\documentclass{amsart}
\usepackage{amssymb,amsfonts,latexsym,amscd}
\usepackage{amsmath}
\usepackage[usenames]{color}

\definecolor{Purple}{rgb}{1, 0., 0.}

\def\H{{\mathcal H}}

\def\gsvar{\Psi} 
\def\gs{\gsvar^{\mbox{\tiny GS}}} 
\def\cU{\gsvar^{\Sgsa}} 

\def\gsp{\gsvar^\perp} 
\def\gsps{\gsp_{1}} 
\def\gspl{\gsp_{2}} 

\def\cE{\Phi} 

\def\cEp{\cE^{3}_{*}}
\def\cEm{\cE^{1}_*}
\def\cEpp{\cE^{2}_*}
\def\cEi{\cE^{i}_*}

\def\Sgs{u} 
\def\Sgsa{{\Sgs_\alpha}} 
\def\Sgso{{\Sgs_1}}
\def\fF{\Phi^{{\Sgsa}}_{*}}
\def\fFz{\Phi^{{\Sgsa}}_{\sharp}}

\def\dvar{e} 

\def\gvar{\Phi}
\def\gvara{\Phi}

\def\Proj{\pi}

\def\Psivar{\Psi}
\def\Psz{\Psi_0}
\def\Psp{\widetilde\Psi_0^\perp}

\def\rvar{\Delta^\perp_\sharp}
\def\Rvar{\Delta^{\Sgsa}_*}
\def\Rnvar{\Delta^{0}_*}

\def\muvar{\kappa}

\def\vac{\Omega_f}  

\def\Aa{A^-} 
\def\Ac{A^+} 
\def\ttH{{K}} 

\newcommand{\gH}{{\mathfrak H}}  
\newcommand{\gF}{{\mathfrak F}}  

\def\C{{\mathbb{C}}} 
\def\N{{\mathbb{N}}} 
\def\R{{\mathbb{R}}} 

\def\cno{c_{\mathrm{n.o.}}} 

\def\1{{\mathbf{1}}}
\def\la{\langle}
\def\ra{\rangle}
\renewcommand\d{\mathrm{d}}
\renewcommand\Re{\mathrm{Re}\,}

\newtheorem{theorem}{Theorem}[section]

\newtheorem{definition}{Definition}[section]
\newtheorem{corollary}{Corollary}[section]
\newtheorem{lemma}{Lemma}[section]

\newtheorem{proposition}{Proposition}[section]
\newtheorem{remark}{Remark}[section]

\begin{document}

\title[Hydrogen binding energy in QED]
{Quantitative estimates on the binding energy for Hydrogen in
non-relativistic QED}

\author{J.-M. Barbaroux}

\address{
 Centre de Physique Th\'eorique, Luminy Case 907, 13288
 Marseille Cedex~9, France and D\'epartement de Math\'ematiques,
 Universit\'e du Sud Toulon-Var, 83957 La
 Garde Cedex, France.
} \email{barbarou@univ-tln.fr}

\author{T. Chen}

\address{Department of Mathematics, University of Texas at Austin,
1 University Station C1200, Austin, TX 78712, U.S.A. }
\email{tc@math.utexas.edu}

\author{V. Vougalter}

\address{University of Toronto,
                Department of Mathematics,
                Toronto, ON, M5S 2E4,
                Canada
} \email{vitali@math.toronto.edu}

\author{S. Vugalter}

\address{Mathematisches Institut, Ludwig-Maximilians-Universit\"at
M\"unchen, Theresienstrasse 39, 80333 M\"unchen and Institut f\"ur
Analysis, Dynamik und Modellierung, Universit\"at Stuttgart. }
\email{wugalter@mathematik.uni-muenchen.de}

\maketitle

\begin{abstract}
In this paper, we determine the exact expression for the hydrogen
binding energy in the Pauli-Fierz model up to the order
$O(\alpha^5\log\alpha^{-1})$, where $\alpha$ denotes the
finestructure constant, and prove rigorous bounds  on the
remainder term of the order $o(\alpha^5\log\alpha^{-1})$. As a
consequence, we prove that the binding energy is not a real
analytic function of $\alpha$, and verify the existence of
logarithmic corrections to the expansion of the ground state
energy in powers of $\alpha$, as conjectured in the recent
literature.
\end{abstract}


\section{Introduction}

For a hydrogen-like atom consisting of an electron interacting
with a static nucleus of charge $eZ$, described by the
Schr\"odinger Hamiltonian $-\Delta -\frac{\alpha Z}{|x|}$,
\begin{equation}\nonumber
 \inf\sigma(-\Delta) - \inf\sigma(-\Delta -\frac{\alpha Z }{|x|} ) \, = \,
 \frac{(Z \alpha)^2}{4}\ .
\end{equation}
corresponds to the binding energy necessary to remove the electron
to spatial infinity.

The interaction of the electron with the quantized electromagnetic
field is accounted for by adding to $-\Delta -\frac{\alpha
Z}{|x|}$ the photon field energy operator $H_f$, and an operator
$I(\alpha)$ which describes the coupling of the electron to the
quantized electromagnetic field; the small parameter $\alpha$ is
the fine structure constant. Thereby, one obtains the Pauli-Fierz
Hamiltonian described in detail in Section~\ref{the-model}. The
systematic study of this operator, in a more general case
involving more than one electron, was initiated by Bach,
Fr\"ohlich and Sigal \cite{Bachetal1995, Bachetal1999,
Bachetal1999bis}.

In this case, determining the binding energy
\begin{equation}\label{eq:def-binding}
 \inf\sigma\big(\, -\Delta + H_f + I(\alpha)\,\big)
 - \inf\sigma\big(\,-\Delta - \frac{\alpha Z}{|x|} + H_f + I(\alpha)\,\big)\
\end{equation}
is a very hard problem. A main obstacle emerges from the fact that
the ground state energy is not an isolated eigenvalue of the
Hamiltonian, and can not be determined with ordinary perturbation
theory. Furthermore, the photon form factor in the quantized
electromagnetic vector potential occurring in the interaction term
$I(\alpha)$ contains a critical frequency space singularity that
is responsible for the infamous infrared problem in quantum
electrodynamics. As a consequence, quantities such as the ground
state energy do not a priori exist as a convergent power series in
the fine structure constant $\alpha$ with coefficients independent
of $\alpha$.

In recent years, several rigorous results addressing the
computation of the binding energy have been obtained, \cite{HVV,
Hainzletal2005, Chenetal2003}. In particular, the coupling to the
photon field has been shown to increase the binding energy of the
electron to the nucleus, and that up to normal ordering, the
leading term is $\frac{(\alpha Z)^2}{4}$ \cite{Hainzl2001, HVV,
Chenetal2003}.

Moreover, for a model with scalar bosons, the binding energy is
determined in \cite{Hainzletal2005}, in the first subleading order
in powers of $\alpha$, up to $\alpha^3$, with error term
$\alpha^3\log\alpha^{-1}$. This result has inspired the question
of a possible emergence of logarithmic terms in the expansion  of
the binding energy; however, this question has so far remained
open.

In \cite{Bachetal2006}, a sophisticated rigorous renormalization
group analysis is developed in order to determine the ground state
energy (and the renormalized electron mass) up to any arbitrary
precision in powers of $\alpha$, with an expansion of the form
\begin{equation}\label{eq:bfp1}
 \varepsilon_0 + \sum_{k=1}^{2N}
 \varepsilon_k(\alpha)\alpha^{k/2} + o(\alpha^N)\ ,
\end{equation}
(for any given $N$) where the coefficients $\varepsilon_k(\alpha)$
diverge as $\alpha\rightarrow0$, but are smaller in magnitude than
any positive power of $\alpha^{-1}$. The recursive algorithms
developed in \cite{Bachetal2006} are highly complex, and
explicitly computing the ground state energy to any subleading
order in powers of $\alpha$ is an extensive task. While it is
expected that the rate of divergence of these coefficient
functions is proportional to a power of $\log\alpha^{-1}$,
this is not explicitly exhibited in the current literature; 
for instance,
it can a priori not be ruled out 
that terms involving logarithmic corrections cancel mutually.

The choice of atomic units in \cite{Bachetal2006}, inherited from
\cite{Bachetal1999}, and motivated by the necessity to keep the
$\alpha$-dependence only in the interaction term but not in the
Coulomb term, introduces a dependence between the ultraviolet
cutoff and the fine structure constant. This makes it challenging
to compare the result derived for \eqref{eq:bfp1} to our estimate
of the binding energy stated in Theorem~\ref{thm:main}, since we
employ different units with an ultraviolet cutoff independent of
$\alpha$. However, such comparison would require to apply the
procedure used in this paper to the model described in
\cite{Bachetal2006}.

The goal of the current paper is to develop an alternative method
(as a continuation of \cite{BCVVi}) that determines the binding
energy up to several subleading orders in powers of $\alpha$, with
rigorous error bounds, and proving the presence of terms
logarithmic in $\alpha$.

The main result established in the present paper (for $Z=1$)
states that the binding energy can be estimated as
\begin{equation}\label{eq:mrbe}
 \frac{\alpha^2}{4} + \dvar^{(1)}\alpha^3 + \dvar^{(2)}\alpha^4
 + \dvar^{(3)}\alpha^5\log\alpha^{-1} + o(\alpha^5\log\alpha^{-1})\ ,
\end{equation}
where $\dvar^{(i)}$ ($i=1,2,3$) are independent of $\alpha$,
$\dvar^{(1)}>0$, and $\dvar^{(3)}\neq 0$. Their explicit values
are given in Theorem~\ref{thm:main}.

As a consequence, we conclude that the binding energy is {\em not
analytic} in $\alpha$. In addition, our proof clarifies how the
logarithmic factor in $\alpha^5\log\alpha^{-1}$ is linked to the
infrared singularity of the photon form factor in the interaction
term $I(\alpha)$. We note that for some models with a mild
infrared behavior, \cite{GrHa}, the ground state energy is proven
to be analytic in $\alpha$.

Notice that on the basis of the present computations, it is
impossible to determine whether the $\alpha^5\log\alpha^{-1}$ term
comes from $\inf\sigma(-\Delta + H_f + I(\alpha))$, or
$\inf\sigma(-\Delta - \frac{\alpha Z}{|x|} + H_f +
I(\alpha))$ or both. However, this question has been
recently answered in \cite{BarbarouxWugalter2010}.\\

\subsection*{Organization of the proof}

The strategy consists mainly in an iteration procedure based on
variational estimates. The derivation of the main result
\eqref{eq:mrbe} is accomplished in two main steps, with first an
estimate up to the order $\alpha^3$ with an error of the order
$\mathcal{O}(\alpha^4)$, and then up to the order
$\alpha^5\log\alpha$. Each of these steps requires both an upper
and a lower bound on the binding energy. From the knowledge of an
approximate ground state for $-\Delta + H_f + I(\alpha)$ (see
\cite{BCVVi}), the remaining hard task consists mainly in
establishing a lower bound for the ground state energy for $H$
yielding an upper bound for the binding energy. A lower bound for
the binding energy can be obtained by choosing a bona fide trial
state.

It is easy to see that the first term in the expansion for the
binding energy is not smaller than the Coulomb energy
$\frac{\alpha^2}4$ by taking a trial state, which is a product
state of the electronic ground state $\Sgsa$ \eqref{eq:gs-schrod}
of the Schr\"odinger-Coulomb operator $-\Delta
-\frac{\alpha}{|x|}$ and the photon ground state of the
self-energy operator given by \eqref{def:T(P)} with zero total
momentum. Indeed, as shown in \cite{Chenetal2003} there is an
increase of the binding energy. However the first term is exactly
the Coulomb energy \cite{Hainzletal2005}

The proof of the upper bound, up to the order $\alpha^3$ and to
the order $\alpha^5\log\alpha^{-1}$ reduces to a minimization
process.

Because of the soft photon problem, the set of minimizing
sequences does not have a clear structure which helps to find the
minimizer, or approximate minimizers. Therefore, prior to the
above two steps, we have to appropriately restrict the set of
states which we are looking at. For that purpose, we first
establish, in Lemma~\ref{Nf-exp-lemma-1}, an a priori bound on the
expected photon number in the ground state $\gs$ of the
Pauli-Fierz operator $(-\Delta -\frac{\alpha}{|x|} + H_f +
I(\alpha))$, namely $\langle \gs,\, N_f\gs\rangle =
\mathcal{O}(\alpha^2\log\alpha^{-1})$, where $N_f$ is the photon
number operator. Moreover, in Theorem~\ref{lem:4-1}, we estimate
the contribution to the binding energy stemming from states
orthogonal to the ground state $\Sgsa$ of the Schr\"odinger
operator. It enables us to show that the projection of the ground
state $\gs$ onto the subspace of functions orthogonal to $\Sgsa$
is small in several norms.

\underline{Estimate up to the order $\alpha^3$:} To find a lower
bound for the binding energy (Theorem~\ref{thm:apriori-bound}), we
pick a trial function $\cE^{\rm trial}$ which is the sum of two
states. The first state is the product of the ground state
$u_\alpha$ of the Schr\"odinger-Coulomb operator with the
approximate ground state $\Psi_0$ of the self-energy operator
(given by \eqref{def:T(P)} below) with zero total momentum which was
derived in \cite{BCVVi}. The second state $\fFz$
\eqref{def:def-F0} is a term orthogonal to $\Sgsa$ and has been
chosen to minimize, up to the order $\alpha^3$, the sum of the
cross term $\la (-\Delta -\frac{\alpha}{|x|} + H_f + I(\alpha))
\Sgsa\Psi_0,\, \fFz\ra$ and the term $\langle\, ( H_f + P_f^2
-\Delta -\frac{\alpha}{|x|} + \frac{\alpha^2}4) \fFz,\,
\fFz\rangle$ stemming from the quadratic form for $H$ on the trial
state $\cE^{\rm trial}=\Sgsa \Psi_0 + \fFz$.

At that stage, we emphasize that the contribution of these two
states yields not only an $\alpha^3$ term, but also an
$\alpha^5\log\alpha^{-1}$ term, which is in fact the only 
contribution to the binding energy for this order as the detailed
proof of the next order estimate shall show.

To recover the upper bound, we take an arbitrary state $\Psi$
satisfying the condition on the expected photon number. Then we
consider $\Sgsa \cU$, the projection of this state onto
$u_\alpha$, and $\gsp$ the projection onto the subspace orthogonal
to $u_\alpha$. We estimate in Lemma~\ref{lem:lem-direct-fpsi} the
quadratic form of $H$ on $u_\alpha \cU$, in
Lemma~\ref{thm:thm-direct-gg} the quadratic form for $H$ on
$\gsp$, and in Lemmata~\ref{lem:lem-cross-1}-\ref{lem:lem-cross-2}
the cross term. Collecting this estimates in
\eqref{eq:thm:thm-bis-1} below yields that the function $\Psi$ can be an
approximate minimizer up to the order $\alpha^3$, with error term
$\mathcal{O}(\alpha^4\log\alpha^{-1})$ of the functional $\langle
H \Psi,\, \Psi\rangle$ only if the difference $R$ between this
function $\Psi$ and the previous trial state $\cE^{\rm trial}$
satisfies the condition $\|H_f^\frac12 R\|^2 =
\mathcal{O}(\alpha^4\log\alpha^{-1})$ (see
\eqref{eq:add-est-w-log})

According to Lemma~\ref{qmqm}, the last estimate itself allows to
improve the expected photon number estimate for the ground state
of $H$ given by Lemma~\ref{Nf-exp-lemma-1}, replacing there
$\alpha^2\log\alpha^{-1}$ by $\alpha^2$. Then, repeating the above
steps for the upper bound with this improved expected photon
number yields the upper bound of the binding energy with an error
$\mathcal{O}(\alpha^4)$ instead of
$\mathcal{O}(\alpha^4\log\alpha^{-1})$.

\underline{Estimate up to the order $\alpha^5\log\alpha^{-1}$}: In
order to derive the lower bound for the binding energy, we improve
the previous trial function by adding several terms that where
irrelevant for the estimate up to the order $\alpha^3$. Due to the
improved estimates in Lemma~\ref{qmqm}, the expected photon number
for the difference of the true ground state and this new trial
function can not exceed $\mathcal{O}(\alpha^{\frac{33}{16}})$.
Furthermore, we infer from Theorem~\ref{thm:thm-main2} that this
difference satisfies several smallness conditions. Using these
conditions we estimate once more the quadratic form and minimize
it with respect to this difference in a similar way to what is
done in section~\ref{section-S5} for the previous step.

The paper is organized as follows: In section~\ref{the-model}, we give a
detailed definition of the model and state the main result. In
section~\ref{S-photon-bound} we establish an a priori bound on the
expected photon number in the ground state of the Pauli-Fierz
operator. In section~\ref{section-S4} we estimate the contribution
to the binding energy stemming from states orthogonal to the
ground state of the Schr\"odinger operator. The
section~\ref{section-S5} is devoted to the statements of the lower
and the upper bounds for the binding energy up to the order
$\alpha^3$, as well as to the proof of the lower bound. The
difficult part of the proof, namely the upper bound on the binding
energy, is presented in four parts and postponed to
section~\ref{S7}. In subsections~\ref{S5.1}-\ref{S5.3}, we
estimate separately the terms according to the splitting of the
variational state. We collect these results in
subsection~\ref{S5.4} and establish then the proof of the upper
bound up to the order $\alpha^3$. In section~\ref{section-S6}, we
establish the main steps of the proof of the upper bound for the
binding energy up to the order $\alpha^5\log\alpha^{-1}$. We start
with some useful definitions, the statement for the upper bound,
and two propositions (proved in appendices
subsections~\ref{subsection:6-1} and \ref{prf-hgg}), that gives
estimates of the contributions to the binding energy according to
a refined splitting of the variational state. Details of the
remaining straightforward computations are given in
section~\ref{subsection:6-3}. Finally, subsection~\ref{S6.4}
provide the proof for a lower bound on the binding energy.
In  appendix~\ref{app-Thm21prf-1}, we provide some
technical lemmata whose proof are straightforward but rather long.

\section{The model}
\label{the-model}

We study a scalar electron interacting with the quantized
electromagnetic field in the Coulomb gauge, and with the
electrostatic potential generated by a fixed nucleus. The Hilbert
space accounting for the Schr\"odinger electron is given by
$\gH_{el}=L^2(\R^3)$. The Fock space of photon states is given by
 $$
   \gF \; = \; \bigoplus_{n \in \N} \gF_n ,
 $$
where the $n$-photon space $\gF_n =
\bigotimes_s^n\left(L^2(\R^3)\otimes\C^2\right)$ is the symmetric
tensor product of $n$ copies of one-photon Hilbert spaces
$L^2(\R^3)\otimes\C^2$. The factor $\C^2$ accounts for the two
independent transversal polarizations of the photon. On $\gF$, we
introduce creation- and annihilation operators $a_\lambda^*(k)$,
$a_\lambda(k)$ satisfying the distributional commutation relations
 $$
  [ \, a_{\lambda}(k) \, , \, a^\ast_{\lambda'}(k') \, ] \; = \;
  \delta_{\lambda, \lambda'} \, \delta (k-k')
  \; \;   ,
  \quad [ \, a_\lambda^\sharp(k) \, ,
  \, a_{\lambda'}^\sharp(k') \, ] \; = \; 0 ,
 $$
where $a^\sharp_\lambda$ denotes either $a_\lambda$ or
$a_\lambda^*$. There exists a unique unit ray $\vac\in\gF$, the
Fock vacuum, which satisfies $a_\lambda(k) \, \vac=0$ for all
$k\in\R^3$ and $\lambda\in\{1,2\}$.

The Hilbert space of states of the system consisting of both the
electron and the radiation field is given by
$$
    \gH \; := \; \gH_{el} \, \otimes \, \gF .
$$
We use units such that $\hbar = c = 1$, and where the mass of the
electron equals $m=1/2$. The electron charge is then given by
$e=\sqrt{\alpha}$, where the fine structure constant $\alpha$ will
here be considered as a small parameter.

Let $x\in\R^3$ be the position vector of the electron and let
$y_i\in\R^3$ be the position vector of the $i$-th photon.

We consider the normal ordered Pauli-Fierz Hamiltonian on $\gH$
for Hydrogen,
\begin{equation}\label{rpf}
 : \left(i\nabla_{\! x}\otimes I_f -
 \sqrt{\alpha} A(x)\right)^2\! : +V(x) \otimes I_f
 + I_{el}\otimes H_f .
\end{equation}
where $:\cdots:$ denotes normal ordering, corresponding to the
subtraction of a normal ordering constant $\cno\alpha$, with $\cno
I_f := [\Aa(x) , \Ac(x)]$ is independent of $x$.

The electrostatic potential $V(x)$ is the Coulomb potential for a
static point nucleus of charge $e = \sqrt{\alpha}$ (i.e., $Z=1$)
 $$
 V(x) = -\frac{\alpha}{|x|} .
 $$

We will describe the quantized electromagnetic field by use of the
Coulomb gauge condition.

The operator that couples an electron to the quantized vector
potential is given by
\begin{equation}\nonumber
\begin{split}
  A(x) & = \sum_{\lambda = 1,2} \int_{\R^3}
  \frac{\chi_\Lambda(|k|)}{2\pi|k|^{1/2}}
  \varepsilon_\lambda(k)\Big[ e^{ikx} \otimes a_\lambda(k)  +
  e^{-ikx} \otimes a_\lambda^\ast
  (k) \Big] \d k
\end{split}
\end{equation}
where by the Coulomb gauge condition, ${\rm div}A =0$.

The vectors $\varepsilon_\lambda(k)\in\R^3$ are the two
orthonormal polarization vectors perpendicular to $k$,
 $$
  \varepsilon_1(k) = \frac{(k_2, -k_1, 0)}
  {\sqrt{k_1^2 + k_2^2}}\qquad
 {\rm and} \qquad
   \varepsilon_2(k) = \frac{k}{|k|}\wedge \varepsilon_1(k).
 $$
The function $\chi_\Lambda$ implements an {\it ultraviolet cutoff}
on the momentum $k$. We assume $\chi_\Lambda$ to be of class
$C^1$, with compact support in $\{ |k|\leq\Lambda\}$,
$\chi_\Lambda\leq 1$ and $\chi_\Lambda = 1$ for $|k|\leq
\Lambda-1$. For convenience, we shall write
 $$
  A(x) \; = \; \Aa(x) \, + \, \Ac(x) ,
 $$
where
 $$
  \Aa(x) \; = \; \sum_{\lambda=1,2} \int_{\R^3}  \, \frac{\chi_\Lambda(|k|)}
  {2\pi |k|^{1/2}} \, \varepsilon_\lambda(k) \, \mathrm{e}^{i k
  x} \, \otimes \, a_\lambda(k) \, \mathrm{d}k
 $$
is the part of $A(x)$ containing the annihilation operators, and
$\Ac(x)=(\Aa(x))^*$.

The photon field energy operator $H_f$ is given by
\begin{equation}\nonumber
H_f = \sum_{\lambda= 1,2} \int_{\R^3} |k| a_\lambda^\ast (k)
a_\lambda (k) \d k.
\end{equation}

We will, with exception of our discussion in
Section~\ref{S-photon-bound}, study the unitarily equivalent
Hamiltonian
\begin{equation}\label{eq-H-Utrsf-1}
  H = U\Big( : \left(i\nabla_{\! x}\otimes I_f -
 \sqrt{\alpha} A(x)\right)^2\! : +V(x) \otimes I_f
 + I_{el}\otimes H_f  \Big)U^*\ ,
\end{equation}
where the unitary transform $U$ is defined by
\begin{equation}\nonumber
  U = \mathrm{e}^{i P_f.x} ,
\end{equation}
and
\begin{equation}\nonumber
 P_f = \sum_{\lambda =1,2} \int  k \,
 a^\ast_\lambda(k) a_\lambda(k) \d k
\end{equation}
is the photon momentum operator.
We have
 $$
  U i\nabla_x U^* = i\nabla_x + P_f \quad\mbox{and}\quad
  U A(x)U^* = A(0)\ .
 $$
In addition, the Coulomb operator $V$, the photon field energy
$H_f$, and the photon momentum $P_f$ remain unchanged under the
action of $U$. Therefore, in this new system of variables, the
Hamiltonian reads as follows
\begin{equation}\label{def:unitary-tH}
\begin{split}
  H =\ \ : \left( (i\nabla_x\otimes I_f - I_{el}\otimes P_f)
 - \sqrt{\alpha} A(0)\right)^2 :
 + H_f -\frac{\alpha}{|x|} \ ,
\end{split}
\end{equation}
where $:...:$ denotes again the normal ordering. Notice that the
operator $H$ can be rewritten, taking into account the normal
ordering and omitting, by abuse of notations, the operators
$I_{el}$ and $I_f$,
\begin{equation}\nonumber
\begin{split}
 H  =  &\ (-\Delta_x - \frac{\alpha}{|x|}) + (H_f + P_f^2)
  - 2\Re \left(i\nabla_x . P_f\right)\\
 &\ - 2 \sqrt{\alpha} (i\nabla_x
 - P_f). A(0) + 2 \alpha \Ac(0). \Aa(0) + 2\alpha\Re
 (\Aa(0))^2\ .
\end{split}
\end{equation}



For a free spinless electron coupled to the quantized
electromagnetic field, the self-energy operator $T$ is given by
 $$
    T = \ \ :\left(i\nabla_{\! x}\otimes I_f
    - \sqrt{\alpha} A(x)\right)^2 \! : + I_{el}\otimes H_f.
 $$
We note that this system is translationally invariant; that is,
$T$ commutes with the operator of total momentum
 $$
 P_{tot} = p_{el}\otimes I_f + I_{el}\otimes P_f ,
 $$
where $p_{el}$ and $P_f$ denote respectively the electron and the
photon momentum operators.

Let $\C \otimes\mathfrak{F}$ be the fiber Hilbert space
corresponding to conserved total momentum $p\in\R^3$.

For fixed value $p$ of the total momentum, the restriction of $T$
to the fibre space is given by (see e.g. \cite{Chen2008})
\begin{equation}\label{def:T(P)}
  T(p) = \ \ :(p - P_f
  - \sqrt{\alpha} A(0))^2:  + H_f
\end{equation}
where by abuse of notation, we again dropped all tensor products
involving the identity operators $I_f$ and $I_{el}$. Henceforth,
we will write
 $$
 A^\pm\equiv A^\pm(0)\ .
 $$
Moreover, we denote
 $$
 \Sigma_0 = \inf\sigma(T) \quad\mbox{and}
 \quad \Sigma = \inf\sigma(H) = \inf\sigma(T+V).
 $$
It is proven in \cite{BCFS2,Chen2008} that $\Sigma_0 = \inf\sigma
(T(0))$ is an eigenvalue of the operator $T(0)$.

Our main result is the following theorem.
\begin{theorem}\label{thm:main}
The binding energy fulfills
\begin{equation}\label{eq:main1}
 \begin{split}
   \Sigma_0 -\Sigma =  & \frac14\alpha^2 +
   \dvar^{(1)} \alpha^3
   +
   \dvar^{(2)} \alpha^4 +
   \dvar^{(3)} \alpha^5\log\alpha^{-1}
   + o(\alpha^{5}\log\alpha^{-1}),
 \end{split}
\end{equation}
where
 $$
  \dvar^{(1)} = \frac2{\pi} \int_0^\infty
  \frac{\chi_\Lambda^2(t)}{1+t} \mathrm{d}t ,
 $$
 \begin{equation}\nonumber
 \begin{split}
  & \dvar^{(2)} =  \frac{2}{3} \, \Re \sum_{i=1}^3 \la (\Aa)^i (H_f+P_f^2)^{-1}\Ac
   . \Ac \vac, (H_f + P_f^2)^{-1} (\Ac)^i \vac \ra
   \\
   & + \frac13
   \sum_{i=1}^3 \|(H_f + P_f^2)^{-\frac12}
   \Big(
   2 \Ac.P_f (H_f + P_f^2)^{-1} (\Ac)^i
   - P_f^i (H_f+P_f^2)^{-1}\Ac. \Ac \Big)\vac\|^2
   \\
   & - \frac23 \sum_{i=1}^3 \| A^- (H_f + P_f^2)^{-1} (A^+)^i
   \vac\|^2
   + 4 a_0^2 \|  Q_1^\perp (-\Delta -\frac{1}{|x|}+\frac14 )^{-\frac12}
  \Delta \Sgso\|^2 ,
 \end{split}
 \end{equation}
\begin{equation}\nonumber
  a_0 = \int
  \frac{k_1^2+k_2^2}{4\pi^2|k|^3} \frac{2}{|k|^2
  +|k|} \chi_\Lambda (|k|)\, \d k_1 \d k_2 \d k_3 ,
\end{equation}
\begin{equation}\nonumber
 \dvar^{(3)} = - \frac{1}{3\pi}  \| (-\Delta -\frac{1}{|x|} +\frac14)^{\frac12} \nabla
 \Sgso\|^2 ,
\end{equation}
and $Q_1^\perp$ is the projection onto the orthogonal complement
to the ground state $\Sgso$ of the Schr\"odinger operator $-\Delta
-\frac{1}{|x|}$ (for $\alpha=1$).
\end{theorem}


\section{Bounds on the expected photon
number}\label{S-photon-bound}

\begin{lemma}
\label{Nf-exp-lemma-1} Let
\begin{equation}\nonumber
    \ttH \; = \; {(i\nabla-\sqrt\alpha A(x))^2}
    \, + \, H_f \, - \, \frac{\alpha}{|x|}
    \; = \; v^2 + H_f -\frac{\alpha}{|x|} \;,
\end{equation}
be the Pauli-Fierz operator defined without normal ordering, where
$v = i\nabla - \sqrt{\alpha} A(x)$. Let $\Psi   \in \H$ denote the
ground state of $K$,
\begin{equation}\nonumber
     \ttH \, \Psi   \; = \; E \, \Psi   \;,
\end{equation}
normalized by
\begin{equation}\nonumber
    \| \, \Psi   \, \| \; = \; 1 \;.
\end{equation}
Let
\begin{equation}\nonumber
    N_f \; := \; \sum_{\lambda=1,2}\int
    \, a_\lambda^*(k) \, a_\lambda(k) \, \d k
\end{equation}
denote the photon number operator. Then, there exists a constant
$c$ independent of $\alpha$,  such that for any sufficiently small
$\alpha>0$, the estimate
\begin{equation}\nonumber
     \la\Psi   \, , \, N_f \,\Psi  \ra \; \leq \; c \, \alpha^{2 } \,
     \log\alpha^{-1}
\end{equation}
is satisfied.
\end{lemma}

\begin{proof}

Using
\begin{equation}\nonumber
 [a_\lambda(k), H_f] = |k|,\quad\mbox{and}\quad[a_\lambda(k),v]
 =\frac{\epsilon(k)}{2\pi|k|^\frac12} \chi_\lambda(k) \mathrm{e}^{ik.x}
\end{equation}
and the pull-through formula,
\begin{eqnarray*}
    a_\lambda(k) E \Psi  &=&a_\lambda(k) \ttH  \Psi
    \nonumber\\
    &=& \Big[( H_f +|k|)a_\lambda(k) - \frac{1}{|x|}a_\lambda(k)
    + [v, a_\lambda(k)]v
    +v [v, a_\lambda(k)]\Big] \Psi  \;,
    \label{pull-through-1}
\end{eqnarray*}
we get
\begin{eqnarray}
    a_\lambda(k) \, \Psi
    &=&- \, \frac{\sqrt\alpha \chi_\Lambda(|k|)}{2\pi\sqrt{|k|}} \,
     \frac{2}{\ttH +|k|-E} \, \left( \Big( \, i\nabla-\sqrt\alpha A(x) \, \Big)\cdot
    \epsilon_\lambda(k) \mathrm{e}^{ik.x}\right)\, \Psi    \;.
    \label{alk-psi-GLL-1}
\end{eqnarray}
From \eqref{alk-psi-GLL-1}, we obtain
\begin{eqnarray*}
    \Big\| \, a_\lambda(k)\Psi   \, \Big\|^2
    &\leq&  \frac{ \alpha \, \chi_\Lambda(|k|)}{\pi^2 |k| }\,
    \Big\| \, \frac{1}{\ttH +|k|-E} \, \Big\|^2 \,
    \Big\| \, \Big( \, i\nabla-\sqrt\alpha A(x) \, \Big) \Psi   \, \Big\|^2
    \nonumber\\
    &\leq&\frac{ \alpha \,
    \chi_\Lambda(|k|)}{\pi^2  \,|k|^3 }\,
    \Big[\big\langle\Psi   \, , \, \ttH \, \Psi   \big\rangle
    \, + \, \big\langle\Psi   \, , \, \frac{\alpha}{|x|} \, \Psi   \big\rangle \, \Big] \;,
    \label{alk-psi-GLL-2}
\end{eqnarray*}
since $\ttH - E\geq0$, and $(i\nabla-\sqrt\alpha A(x) \big)^2\leq
(\ttH + \frac{\alpha}{|x|})$.

Since $\langle \Psi,\, :K:\Psi\rangle = -\alpha^2/4 + o(\alpha^2)$
(\cite{Hainzl2001, HVV, Chenetal2003}) and $:K: = K -
c_{\mathrm{n.o.}}\alpha$, we have
\begin{equation}\label{eq:est-1}
    \big\langle\Psi   \, , \, \ttH \, \Psi   \big\rangle \leq
    c_{\mathrm{n.o.}} \,\alpha .
\end{equation}
Moreover, for the normal ordered hamiltonian defined in
\eqref{rpf}, we have
\begin{equation}\label{eq:est-3}
\begin{split}
 :K:\ \  & = -\Delta_x
 - 4\sqrt{\alpha} \Re \left( i\nabla_x . \Aa(x)\right)
 + \alpha : A(x)^2 :
 +  H_f - \frac{\alpha}{|x|} \\
 & \geq (1-c\sqrt{\alpha})  (-\Delta_x)  +
 \alpha :A(x)^2: + (1-c\sqrt{\alpha}) H_f -\frac{\alpha}{|x|} \\
 & \geq -c\alpha^2 - \frac12\Delta_x -\frac{\alpha}{|x|}
 = -c\alpha^2 + 2(-\frac14 \Delta_x - \frac{\alpha}{|x|}) +
 \frac{\alpha}{|x|}\ .
\end{split}
\end{equation}
where in the last inequality we used (see e.g. \cite{GLL})
\begin{equation}\nonumber
  \alpha :A(x)^{2}: + (1-c \sqrt{\alpha})H_f \geq - c\alpha^2 .
\end{equation}
Since
\begin{equation}\nonumber
 - \frac14 \Delta_x -\frac{\alpha}{|x|} \geq - 4\alpha^2\ ,
\end{equation}
inequality \eqref{eq:est-3} implies
\begin{equation}\label{eq:est-2}
    \big\langle\Psi   \, , \, \frac{\alpha}{|x|} \Psi  \rangle
    \leq \langle -\Delta_x \Psi  , \Psi  \rangle
    + \frac14 \alpha^2 \|\Psi  \|^2 \leq
    c  \alpha^2 \|\Psi  \|^2.
\end{equation}
Collecting \eqref{eq:est-1}, \eqref{eq:est-2} and
\eqref{alk-psi-GLL-1} we find
\begin{equation}\label{a-psi-fund-1}
    \Big\| \, a_\lambda(k)\Psi   \, \Big\| \; \leq \; \frac{c \, \alpha \,
    \chi_\Lambda(|k|)}{|k|^{\frac32}} \;.
\end{equation}
This a priori bound exhibits the $L^2$-critical singularity in
frequency space. It does not take into consideration the
exponential localization of the ground state due to the confining
Coulomb potential, and appears in a similar form for the free
electron.

To account for the latter, we use the following two results from
the work of Griesemer, Lieb, and Loss, \cite{GLL}. Equation (58)
in \cite{GLL} provides the bound
\begin{equation}\nonumber
    \Big\| \, a_\lambda(k) \Psi   \, \Big\| \; < \;
    \frac{c \, \sqrt\alpha \, \chi_\Lambda(|k|)}{|k|^{\frac12} }
    \Big\| \, |x|\Psi   \, \Big\|\;.
\end{equation}
Moreover, Lemma 6.2 in \cite{GLL} states that
\begin{equation}\nonumber
    \Big\| \, \exp[\beta|x|] \Psi   \, \Big\|^2 \;
    \leq \; c \, \Big[1+\frac{1}{\Sigma_0-E-\beta^2}\Big]
    \, \| \, \Psi   \, \|^2\;,
\end{equation}
for any
\begin{equation}\nonumber
    \beta^2 \; < \; \Sigma_0-E \; = \; O(\alpha^2) \; .
\end{equation}
For the 1-electron case, $\Sigma_0$ is the infimum of the
self-energy operator, and $E$ is the ground state energy of $:K:
$. Choosing $\beta=O(\alpha)$,
\begin{eqnarray*}
    \|\, |x|\Psi   \, \|&\leq& \||x|^4 \Psi  \|^{\frac14}
    \, \|\Psi  \|^{\frac34}
    \leq \frac{(4!)^{\frac14}}{\beta} \, \Big\|
    \, \exp[\beta|x|]\Psi   \, \Big\|^\frac14 \|\Psi  \|^\frac34
    \nonumber\\
    &\leq&\frac{c}{\beta} \,
    \Big[1+\frac{1}{\Sigma_0-E-\beta^2}\Big]^{\frac18}
    \,\| \, \Psi   \, \|
    \nonumber\\
    &\leq& c_1\alpha^{-\frac54} \;.
\end{eqnarray*}
Notably, this bound only depends on the binding energy of the
potential.

Thus,
\begin{equation}\label{a-psi-fund-2}
    \Big\| \, a_\lambda(k) \Psi   \, \Big\|<\frac{c \alpha^{-\frac34}
    \chi_\Lambda(|k|)}{|k|^{\frac12} } \;.
\end{equation}
We see that binding to the Coulomb potential weakens the infrared
singularity by a factor $|k|$, but at the expense of a large
constant factor $\alpha^{-2}$. For the free electron, this
estimate does not exist.

Using \eqref{a-psi-fund-1} and \eqref{a-psi-fund-2}, we find
\begin{eqnarray*}
    \la \, \Psi   \, , \, N_f \, \Psivar  \, \ra&=&\int \d k \, \Big\| \, a_{\lambda}(k)\Psivar  \, \Big\|^2
    \nonumber\\
    &\leq&\int_{|k|<\delta} \d k \, \frac{c \, \alpha^{-\frac32}} {|k| }
    \, + \,
    \int_{\delta\leq |k|\leq\Lambda} \d k \, \frac{c \, \alpha^2}{|k|^3}
    \nonumber\\
    &\leq&c \, \alpha^{-\frac32} \, \delta^2 \, + \, c' \, \alpha^2 \, \log \frac1\delta
    \nonumber\\
    &\leq&c \, \alpha^{\frac94}  \,
    + \, c'' \, \alpha^2 \, \log \alpha^{-1} \;.
\end{eqnarray*}
for $\delta=\alpha^{\frac{15}{8}}$. This proves the lemma.
\end{proof}

\section{Estimates on the quadratic form for states orthogonal
to the ground state of the Schr\"odinger operator}
\label{section-S4}

Throughout this paper, we will denote by $\Proj^n$ the projection
onto the $n$-th photon sector (without distinction for the
$n$-photon sector of $\gF$ and the $n$-photon sector of
$\mathfrak{H}$). We also define $\Proj^{\geq n} =
1-\sum_{j=0}^{n-1}\Proj^j$.

Starting with this section, we study the Hamiltonian $H$ defined
in \eqref{def:unitary-tH}, written in relative coordinates. In
particular, $i\nabla_x$ now stands for the operator unitarily
equivalent to the operator of total momentum, which, by abuse of
notation, will be denoted by $P$.

Let
\begin{equation}\label{eq:gs-schrod}
  \Sgsa  (x)= \frac{1}{\sqrt{8\pi}} \alpha^{3/2}
  \mathrm{e}^{-\alpha|x|/2}
\end{equation}
be the normalized ground state of the Schr\"odinger operator
 $$
  h_\alpha:= -\Delta_x -\frac{\alpha}{|x|}.
 $$
We will also denote by $-e_0 = -\frac{\alpha^2}4$ and
$-e_1=-\frac{\alpha^2}{16}$ the two lowest eigenvalues of this
operator.
%
%
%
\begin{theorem}\label{lem:4-1}
Assume that $\gvara \in\mathfrak{H}$ fulfils $\la \Proj^k\gvara,
\Sgsa \ra_{L^2(\R^3, \mathrm{d}x)} =0$, for all $k\geq 0$. Then
there exists $1>\nu>0$ and $\alpha_0>0 $ such that for all
$0<\alpha<\alpha_0$
\begin{equation}\label{eq:lem-41}
  \la H\gvar,\,\gvar\ra \geq (\Sigma_0 - e_0)\|\gvara\| ^2 + \delta \|\gvara\| ^2
  + \nu\|H_f^{\frac12}\gvara\|^2,
\end{equation}
 where $\delta = (e_0 - e_1)/2 = \frac{3}{32} \alpha^2$.
\end{theorem}
\begin{remark}\label{rem:A}
All photons with momenta larger than the ultraviolet cutoff do not
contribute to lower the energy. More precisely, due to the cutoff
function $\chi_\Lambda(|k|)$ in the definition of $A(x)$, and
since we have
 $$
  H = (i\nabla_x -P_f)^2 -2\sqrt{\alpha}\Re(i\nabla_x -P_f).A(0)
  + \alpha :A(0)^2: + H_f -\frac{\alpha}{|x|},
 $$
it follows that for any given normalized state $\Phi\in\gH$, there
exists a normalized state $\Phi_{\leq\Lambda}$ such that $\forall
x\in\R^3$, for all $n\in\{1,2,\ldots)$, for all $((k_1,
\lambda_1), \ (k_2,\lambda_2), \ldots,\ (k_n, \lambda_n))\in
\left((\R^3\setminus\{k, \ |k|\leq\Lambda+1\})\times\{1,2\}
\right)^n$, we have
\begin{equation}\label{eq:cutoff}
  \Proj^n\Phi_{\leq \Lambda}(x,(k_1, \lambda_1),
  \ (k_2,\lambda_2), \ldots,\ (k_n, \lambda_n)) = 0
\end{equation}
and
 $$
 \langle \Phi_{\leq \Lambda}, H \Phi_{\leq\Lambda} \rangle
 \leq \langle \Phi, H \Phi\rangle\quad\mbox{and}\quad
 \langle \Phi_{\leq \Lambda}, T \Phi_{\leq\Lambda} \rangle
 \leq \langle \Phi, T \Phi\rangle .
 $$

A key consequence of this remark is that throughout the paper, all
states will be implicitly assumed to fulfill condition
\eqref{eq:cutoff}. This is crucial for the proof of
Corollary~\ref{cor:cor-4-2}.
\end{remark}
%
%
%
%
%
%
To prove Theorem~\ref{lem:4-1}, we first need the following Lemma.
\begin{lemma}\label{lem:estimate}
There exists $c_0>0$ such that for all $\alpha$ small enough we
have
\begin{equation}\nonumber
 H  -\frac12 (P - P_f)^2 - \frac12 H_f \geq
 - c_0\alpha^2 .
\end{equation}
\end{lemma}
A straightforward consequence of this lemma is the following
result.
\begin{corollary}\label{coro4.3}
Let $\gs$ be the normalized ground state of $H$. Then
\begin{equation}\label{eq:eq16}
  \la H_f\gs,\gs \ra \leq 2c_0\alpha^2\|\gs\|^2
\end{equation}
\end{corollary}
\begin{proof}
This follows from $\la H\gs, \gs\ra \leq (\Sigma_0-e_0)\|\gs\|^2
<0$. The last inequality holds since $e_0=-\alpha^2/4$ and
$\Sigma_0$ is the infimum for the normal ordered Hamiltonian for
the free electron, and thus $\Sigma_0\leq0$.
\end{proof}

Moreover, from Theorem~\ref{lem:4-1} and Lemma~\ref{lem:estimate},
we obtain
\begin{corollary}\label{cor:cor-4-2}
Assume that $\gvara \in\mathfrak{H}$ is such that $\la
\Proj^n\gvara, \Sgsa  \ra_{L^2(\R^3, \mathrm{d}x)} =0$ holds  for
all $n\geq 0$. Then, for $\nu$ and $\delta$ defined in
Theorem~\ref{lem:4-1}, there exists $\zeta>0$, and $\alpha_0>0 $
such that for all $0<\alpha<\alpha_0$ we have
\begin{equation}\label{eq:added-m0}
  \la H\gvara,\,\gvara\ra \geq (\Sigma_0 - e_0)\| \gvara \| ^2
  + M[\gvara] ,
\end{equation}
where
\begin{equation}\label{eq:def-L0}
  M[\gvara] := \frac{\delta}{2} \| \gvara \| ^2
  + \frac{\nu}{2}\|H_f^{\frac12}\gvara\|^2 + \zeta \|(P-P_f)  \gvara \|^2 +\zeta \|\Proj^{n\leq
  4}P\gvara\|^2.
\end{equation}
\end{corollary}
%
%
%
\begin{proof}
According to Remark~\ref{rem:A}, there exists $\widetilde{c}>1$
such that the operator inequality $P_f^2\Proj^{n\leq 4}\leq
\widetilde{c} H_f\Proj^{n\leq 4}$ holds on the set of states for
which \eqref{eq:cutoff} is satisfied. The value of $\widetilde{c}$
only depends on the ultraviolet cutoff $\Lambda$. Thus,
\begin{equation}\nonumber
 \| \Proj^{n\leq 4}P\gvara\|^2 \leq 2 \| (P-P_f)\gvara\|^2 + 2\|\Proj^{n\leq
 4}P_f  \gvara \|^2\leq 2\| (P-P_f)g\|^2 + 2\widetilde{c} \|H_f^\frac12 \Proj^{n\leq
 4} \gvara\|^2 ,
\end{equation}
which yields
 $$
  \| (P-P_f) \gvara\|^2 \geq \frac12 \| \Proj^{n\leq 4}P \gvara\|^2
  - \widetilde{c} \|H_f^\frac12 \Proj^{n\leq 4} \gvara\|^2\ .
 $$
Therefore, it suffices to prove Corollary~\ref{cor:cor-4-2} with
$M[\gvara] $ replaced by
\begin{equation}\label{eq-def-tM-1}
  \widetilde{M}[\gvara]  := \frac{\delta}{2} \| \gvara \| ^2 +\frac34 \nu
  \|H_f^\frac12 \gvara  \|^2
  + 2 \zeta  \| (P-P_f)  \gvara \|^2,
\end{equation}
and $\zeta$ small enough so that $\widetilde{c}\zeta < \frac\nu4$.

Now we consider two cases. Let $c_1:= \max \{ 8 c_0, 8
\delta/\alpha^2\}$.

If $\|(P-P_f)  \gvara \|^2 \leq c_1\alpha^2 \| \gvara \| ^2$,
Theorem~\ref{lem:4-1} and the above remark imply
\eqref{eq:added-m0}.

If $\|(P-P_f)  \gvara \|^2
>c_1\alpha^2 \|  \gvara \|^2$,
Lemma~\ref{lem:estimate} implies that
\begin{equation}\nonumber
\begin{split}
  & \langle H \gvara ,\gvar\rangle \geq \frac12 \langle (P-P_f)^2\gvara,  \gvara  \rangle
  + \frac12 \langle H_f\gvara,  \gvara  \rangle -c_0\alpha^2 \| \gvara \| ^2\\
  & \geq \frac{1}{4} \langle (P-P_f)^2\gvara,
   \gvara \rangle  +
  \frac12 \langle H_f\gvara,  \gvara  \rangle + \frac18 c_1\alpha^2 \| \gvara \| ^2 \  ,
\end{split}
\end{equation}
which concludes the proof since $\Sigma_0-e_0<0$.
\end{proof}
\noindent\textbf{Proof of Lemma~\ref{lem:estimate}.} Recall the
notation $A^\pm\equiv A^\pm(0)$. The following holds.
\begin{equation}\nonumber
 \begin{split}
 & H - \frac12 (P-P_f)^2-\frac12 H_f \\
 & = \frac12(P-P_f)^2
 -\frac{\alpha}{|x|} - 2 \sqrt{\alpha}\Re \left( (P-P_f) .
 A(0)\right)
 + 2\alpha \Re(\Aa)^2 + 2\alpha\Ac.\Aa + \frac12 H_f ,
 \end{split}
 \end{equation}
\begin{equation}\label{eq:i}
  \frac14(P-P_f)^2 -\frac{\alpha}{|x|} \geq -4 \alpha^2.
\end{equation}
and
\begin{equation}\label{eq:ii}
\begin{split}
 2\sqrt{\alpha} |\la (P-P_f) . A(0)\psi,\psi\ra|
 \leq 2\sqrt{\alpha} \|(P-P_f)\psi\|^2 +
 2\sqrt{\alpha}\|\Aa\psi\|^2\ .
\end{split}
\end{equation}
By the Schwarz inequality, there exists $c_1$ independent of
$\alpha$ such that
\begin{equation}\label{eq:iii}
  \|\Aa\psi\|^2 \leq
  c_1 \| H_f^{\frac12}\psi\|^2.
\end{equation}
Inequalities \eqref{eq:ii}-\eqref{eq:iii} imply that for small
$\alpha$,
\begin{equation}\label{eq:4i}
2\sqrt{\alpha} |\la(P-P_f) . A(0)\psi,\psi\ra|  \leq \frac14 \|
(P-P_f)\psi\|^2 + \frac14 \la H_f\psi, \psi\ra.
\end{equation}
Moreover using \eqref{eq:iii} and
\begin{equation}\label{gll2}
 \|\Ac\psi\|^2 \leq c_2
 \|\psi\|^2 + c_3 \|H_f^{\frac12}\psi\|^2 \,,
\end{equation}
we arrive at
\begin{equation}\label{eq:5i}
\begin{split}
 \alpha \la (\Aa)^2\psi,\psi\ra = \alpha \la \Aa
 \psi, \Ac\psi\ra \leq \epsilon\|\Aa\psi\|^2
 + \epsilon^{-1}\alpha^2
 \|\Ac\psi\|^2 \\
 \leq \epsilon c_1 \| H_f^{\frac12}\psi\|^2 +
 \epsilon^{-1} \alpha^2
 (c_2\|\psi\|^2 + c_3 \|H_f^{\frac12}\psi\|^2) \,.
\end{split}
\end{equation}
Collecting the inequalities \eqref{eq:i}, \eqref{eq:4i} and
\eqref{eq:5i} with $\epsilon < 1/(8c_1)$ and $\alpha$ small
enough, completes the proof. \qed

$\;$

\noindent\textbf{Proof of Theorem~\ref{lem:4-1}:} Let
$\gvar:=\gvar_1 + \gvar_2 := \chi(|P| < \frac{p_c}2) \gvar +
\chi(|P| \geq \frac{p_c}2) \gvar $, where $P=i\nabla_x$ is the
total momentum operator (due to the transformation
(\ref{eq-H-Utrsf-1})) and $p_c=\frac13$ is a lower bound on the
norm of the total momentum for which \cite[Theorem 3.2]{Chen2008}
holds.

Since $P$ commutes with the translation invariant operator $H
+\frac{\alpha}{|x|}$, we have for all $\epsilon\in (0,1)$,
 \begin{equation}
 \begin{split}\label{eq:eq}
  &\la H \gvar , \gvar\ra  =  \la H \gvar_1, \gvar_1 \ra + \la H \gvar_2,\gvar_2\ra
  - 2\Re \la \frac{\alpha}{|x|} \gvar_1, \gvar_2\ra \\
  & \geq  \la H \gvar_1, \gvar_1 \ra + \la H \gvar_2,\gvar_2\ra
  - \epsilon \la \frac{\alpha}{|x|} \gvar_1, \gvar_1\ra
  -\epsilon^{-1} \la \frac{\alpha}{|x|} \gvar_2, \gvar_2\ra .
 \end{split}
 \end{equation}
$\bullet$ First, we have the following estimate
\begin{equation}\nonumber
\begin{split}
 & :(P-P_f -\sqrt{\alpha} A(0))^2: + H_f \\
 & = (P-P_f)^2 - 2\Re (P-P_f) . \sqrt{\alpha} A(0) + \alpha :A(0)^2: +
 H_f  \\
 & \geq  (1-\sqrt{\alpha})(P-P_f)^2 + (\alpha
 -\sqrt{\alpha}):A(0)^2: + H_f - \cno \sqrt{\alpha}
   \\
 & \geq  (1-\sqrt{\alpha})(P-P_f)^2 +
 (1- \mathcal{O}(\sqrt{\alpha})) H_f - \mathcal{O}(\sqrt{\alpha})
\end{split}
\end{equation}
where in the last inequality we used \eqref{eq:iii} and
\eqref{gll2}. Therefore
\begin{equation}\label{eq:lem1-4}
\begin{split}
 \la (H - \epsilon^{-1} \frac{\alpha}{|x|})\gvar_2, \gvar_2\ra
 \geq \la (\frac{1-\sqrt{\alpha}}{2} (P-P_f)^2 -
 (1+\epsilon^{-1})\frac{\alpha}{|x|}) \gvar_2, \gvar_2\ra\\
 + \Big\la \left( \frac{1-\sqrt{\alpha}}{2} (P-P_f)^2
 + (1-\mathcal{O}(\sqrt{\alpha}))H_f - \mathcal{O}(\sqrt{\alpha})\right)
 \gvar_2,\gvar_2
 \Big\ra
\end{split}
\end{equation}
The lowest eigenvalue of the Schr\"odinger operator
$-(1-\mathcal{O}(\sqrt{\alpha})) \frac{\Delta}2 -
\frac{(1+\epsilon^{-1})\alpha}{|x|}$ is larger than
$-c_\epsilon\alpha^2$. Thus, using \eqref{eq:lem1-4} and denoting
 $$
  L:=\frac{1-\sqrt{\alpha}}{2} (P-P_f)^2 +
  (1-\mathcal{O}(\sqrt{\alpha}))H_f - \mathcal{O}(\sqrt{\alpha})
  -c_\epsilon\alpha^2\ ,
 $$
we get
\begin{equation}\label{eq:lem1-5}
 \la H\gvar_2, \gvar_2 ) -\epsilon^{-1} \la \frac{\alpha}{|x|}\gvar_2,\gvar_2\ra
 \geq (L\gvar_2, \gvar_2\ra .
\end{equation}
Now we have the following alternative: Either $|P_f|<
\frac{p_c}3$, in which case we have $\la L \gvar_2, \gvar_2\ra\geq
(\frac{p_c^2}{24} - \mathcal{O}(\sqrt{\alpha}))\|\gvar_2\|^2$, or
$|P_f| \geq \frac{p_c}3$, in which case, using $\gvar_2 =
\chi(|P|>\frac{p_c}{2})\gvar_2$ and $H_f\geq |P_f|$, we have
$L\geq (\frac{p_c}{6} -\mathcal{O}(\sqrt{\alpha}))\|\gvar_2\|^2$.
In both cases, for $\alpha$ small enough, this yields the bound
\begin{equation}\label{eq:lem1-6}
  \la L\gvar_2, \gvar_2\ra \geq \frac{p_c^2}{48} \|\gvar_2\|^2\geq
  (\Sigma_0 - e_0 + \frac78 (e_0-e_1)) \|\gvar_2\|^2
\end{equation}
since, for $\alpha$ small enough, the right hand side tends to
zero, whereas $p_c$ is a constant independent of $\alpha$.
Inequalities~\eqref{eq:lem1-5} and \eqref{eq:lem1-6} yield
\begin{equation}\label{eq:added2}
 \la H\gvar_2, \gvar_2\ra - \epsilon^{-1}\la \frac{\alpha}{|x|}\gvar_2, \gvar_2\ra
 \geq (\Sigma_0 - e_0 + \frac78 (e_0-e_1)) \|\gvar_2\|^2 .
\end{equation}

$\bullet$ For $T(p)$ being the self-energy operator with fixed
total momentum $p\in\R^3$ defined in \eqref{def:T(P)}, we have
from \cite[Theorem~3.1~(B)]{Chen2008}
\begin{equation}\nonumber
 \left|\inf\sigma(T(p)) - {p^2} -\inf\sigma(T(0))\right|
 \leq {c_0 \alpha p^2} .
\end{equation}
Therefore
\begin{equation}\label{eq:lem1-1}
    T(p) \geq (1- o_\alpha(1))p^2 + \Sigma_0  .
\end{equation}
Case 1: If $\|\gvar_2\|^2 \geq 8\|\gvar_1\|^2$, we first do the
following estimate, using \eqref{eq:lem1-1},
\begin{equation}\label{eq:added1}
\begin{split}
& \la H \gvar_1, \gvar_1\ra -\epsilon \la \frac{\alpha}{|x|} \gvar_1, \gvar_1\ra \\
& \geq (1- o_\alpha(1))(P^2\gvar_1, \gvar_1\ra - \la
(1+\epsilon)\frac{\alpha}{|x|}\gvar_1, \gvar_1\ra +  \Sigma_0 \|\gvar_1\|^2\\
& \geq \left(\Sigma_0- (1+ \mathcal{O}(\alpha) +
\mathcal{O}(\epsilon)) e_0\right) \|\gvar_1\|^2 .
\end{split}
\end{equation}
Therefore, together with $\|\gvar_2\|^2 \geq 8\|\gvar_1\|^2$ and
\eqref{eq:added2}, for $\alpha$ and $\epsilon$ small enough this
implies
\begin{equation}\label{eq:eq30}
  \la H \gvar, \gvar\ra \geq (\Sigma_0-e_0)
  \| \gvar\| ^2 + \frac34 (e_0-e_1)\| \gvar\| ^2 .
\end{equation}

Case 2: If $\|\gvar_2\|^2 < 8\|\gvar_1\|^2$, we write the estimate
\begin{equation}\label{eq:added3}
\begin{split}
 & \la H \gvar_1, \gvar_1\ra -\epsilon \la \frac{\alpha}{|x|} \gvar_1, \gvar_1\ra
  \\
 \geq & (1- o_\alpha(1))(P^2\gvar_1, \gvar_1\ra - \la
 (1+\epsilon)\frac{\alpha}{|x|}\gvar_1, \gvar_1\ra +  \Sigma_0 \|\gvar_1\|^2\\
 \geq & (1\!+\! o_\alpha(1)\!+\!\mathcal{O}(\epsilon))
 \Bigg(- e_0\! \sum_{k=0}^\infty \|\, \la \Proj^k \gvar_1,
 \Sgsa \ra_{L^2(\R^3, \mathrm{d}x)}\,\|^2 \!-\! e_1(\|\gvar_1\|^2 \\
 & - \! \sum_{k=0}^\infty
 \|\, \la \Proj^k \gvar_1, \Sgsa \ra_{L^2(\R^3, \mathrm{d}x)}
 \,\|^2) \Bigg)
 +  \Sigma_0 \|\gvar_1\|^2 .
\end{split}
\end{equation}
Now, by orthogonality of $\gvar $  and $\Sgsa $ in the sense that
for all $k$, $\la \Proj^k \gvar, \Sgsa  \ra_{L^2(\R^3,
\mathrm{d}x)} =0$, we get
\begin{equation}\label{eq:eq31}
\begin{split}
  & \sum_{k=0}^\infty \|\la \Sgsa , \Proj^k \gvar_1
  \ra_{L^2(\R^3, \mathrm{d}x)}\|^2
  = \sum_{k=0}^\infty
  \|\la \Sgsa  , \Proj^k
  \gvar_2\ra_{L^2(\R^3, \mathrm{d}x)}\|^2 \\
  & \leq \|\gvar_2\|^2
  \|  \chi(|P| \geq \frac{p_c}2)\Sgsa\|^2
  \leq 8 \|\gvar_1\|^2  \| \chi(|P| \geq \frac{p_c}2)\Sgsa\|^2
   \displaystyle\rightarrow_{\alpha\rightarrow 0} 0
\end{split}
\end{equation}
Thus, for $\alpha$ and $\epsilon$ small enough, \eqref{eq:added3},
\eqref{eq:eq31} and \eqref{eq:added2} imply also \eqref{eq:eq30}
in that case.

$\bullet$ Let $\widetilde{c} = \max\{ \delta, |c_0|\alpha^2\}$.

If $\la H_f \gvar, \gvar\ra < 8 \widetilde{c} \| \gvar\| ^2$,
\eqref{eq:lem-41} follows from \eqref{eq:eq30} with
$\nu=\delta/(16 \widetilde{c})$.

If $\la H_f \gvar, \gvar\ra \geq 8 \widetilde{c} \| \gvar\| ^2$,
using Lemma~\ref{lem:estimate}, we obtain
\begin{equation}
\begin{split}
  \la H \gvar, \gvar\ra& \geq  \frac12 \la H_f \gvar,
  \gvar\ra - c_0 \alpha^2
  \| \gvar \| ^2
  \geq \frac14 \la H_f \gvar, \gvar\ra + \widetilde{c} \| \gvar\| ^2 \\
  & \geq
  \frac14 \la H_f \gvar, \gvar\ra + \delta \| \gvar\| ^2 + (\Sigma_0-e_0)\| \gvar\| ^2,
\end{split}
\end{equation}
since $\Sigma_0-e_0\leq 0$, which yields \eqref{eq:lem-41} with
$\nu=\frac14$.

This concludes the proof of \eqref{eq:lem-41}. \hfill 
%
%
%
%
%

\section{Estimate of the binding energy up to $\alpha^3$ term}
\label{section-S5}

\begin{definition}\label{def:decomposition-1}
Let $\Sgsa $ be the normalized ground state of the Schr\"odinger
operator $h_\alpha$, as defined in \eqref{eq:gs-schrod}.

We define the projection $\cU\in\mathfrak{F}$ of the normalized
ground state $\gs$ of $H$, onto $\Sgsa $ as follows
\begin{equation}\nonumber
 \gs = \Sgsa  \cU +\gsp ,
\end{equation}
where for all $k\geq 0$,
\begin{equation}\label{eq:orthogo}
 \la \Sgsa , \Proj^k \gsp\ra_{L^2(\R^3, \mathrm{d}x)} =0 .
\end{equation}
\end{definition}
%
%
%
%
\begin{remark}
The definition implies that for all $m$
 $$
 (\Proj^m\Psi^{u_\alpha})(k_1,\lambda_1;\, k_2,\lambda_2;\,
 \ldots;\, k_m,\lambda_m) = \int_{\R^3}
 (\Proj^m \gs) (y;\, k_1,\lambda_1;\,
 \ldots;\, k_m,\lambda_m)
 \overline{u_\alpha(y)}\d y  .
 $$
\end{remark}

%
%
%
%
\begin{definition}\label{definition:phi123}
Let
\begin{eqnarray*}
 \cEpp  & : = & - (H_f +P_f^2)^{-1} \Ac \cdot \Ac \vac \\
 \cEp  & : = & - (H_f +P_f^2)^{-1} P_f \cdot \Ac \cEpp  \\
 \cEm  & : = & - (H_f +P_f^2)^{-1} P_f \cdot \Aa \cEpp
\end{eqnarray*}
where evidently, the state $\cEi$ contains $i$ photons.
\end{definition}
%
%
%
\begin{definition}
On $\gF$, we define the positive bilinear form
\begin{equation}\label{eq:def-scalar2}
 \la \, v \, , \, w \, \ra_* \; := \;
 \la \,  v \, , \, (H_f + P_f^2) \, w \, \ra ,
\end{equation}
and its associated semi-norm $\|v\|_* = \la v,v\ra_*^{1/2}$.

We will also use the same notation for this bilinear forms on
$\gF_n$, $\gH$ and $\gH_n$.

Similarly, we define the bilinear form
$\la\,.\,,\,.\,\ra_{\sharp}$ on $\gH$ as
\begin{equation}\nonumber
 \la \,u\,,\, v\, \ra_{\sharp} := \la\, u\,, \,(H_f + P_f^2 + h_\alpha +e_0)
 \,v\,\ra
\end{equation}
and its associated semi-norm $\|v\|_\sharp = \la
v,v\ra_\sharp^{1/2}$.
\end{definition}
%
%
%

\begin{definition}Let
\begin{equation}\label{def:def-F}
  \fF := 2\alpha^{\frac12} \nabla \Sgsa  . (H_f + P_f^2)^{-1} \Ac
  \vac \ ,
 \end{equation}
 and
 \begin{equation}\label{def:def-F0}
  \fFz  := 2\alpha^{\frac12}
  (H_f + P_f^2+h_\alpha +e_0)^{-1} \Ac .\nabla \Sgsa
  \vac  \ .
 \end{equation}
\end{definition}
\begin{remark}
The function $\fF$ is not a vector in the Hilbert space $\gH$
because of the infrared singularity of the photon form factor.
However, in the rest of the paper, we only used the vectors
$H_f^\gamma \fF$ or $P_f^\gamma \fF$, with some $\gamma>0$, which
are always well defined. In particular, all expressions involving
$(H_f+P_f^2)^{-1}$ are always well-defined in the sequel.
\end{remark}

The next theorem gives an upper bound on the binding energy up to
the term $\alpha^3$ with an error term $\mathcal{O}(\alpha^4)$.
\begin{theorem}[Lower bound on the binding
energy]\label{thm:apriori-bound} We have
\begin{equation}\label{eq:eq:apriori-1}
  \Sigma \leq \Sigma_0 - e_0 - \|\fFz \|_{\sharp}^2 + \mathcal{O}(\alpha^4)
\end{equation}
\end{theorem}

\begin{proof}
Using the trial function in $\gH$
  $$
  \cE^{\rm trial} :=
  \Sgsa  (\vac + 2 \alpha^{\frac32}\cEm  + \alpha \cEpp  + 2
  \alpha^{\frac32}\cEp ) + \fFz  \ ,
  $$
and from \cite[Theorem~3.1]{BCVVi} which states
  $$
   \Sigma_0 = - \alpha^2 \|\cEpp \|_*^2
   + \alpha^3 (2 \|\Aa \cEpp \|^2 - 4\|\cEm \|_*^2 -4 \| \cEp \|_*^2)
   +\mathcal{O}(\alpha^4)\ ,
  $$
the result follows  straightforwardly.
\end{proof}

We decompose the function $\cU $ defined in
Definition~\ref{def:decomposition-1} as follows.
\begin{definition}\label{def:decomposition-2}
Let $\eta_1$, $\eta_2$, $\eta_3$ and $\Rvar$ be defined by
 \begin{equation}\nonumber
 \cU  =: \, \Proj^0 \cU  + 2\eta_1 \alpha^{\frac32} \cEm  + \eta_2
 \alpha\cEpp  + 2\eta_3 \alpha^{\frac32} \cEp  +\Rvar ,
 \end{equation}
with the conditions $\Proj^0\Rvar=0$ and $\langle \cEi,
\Proj^i\Rvar\rangle_* =0$ (i=1,2,3), where $\cEm $, $\cEpp $,
$\cEp $ are given in Definition~\ref{definition:phi123}.
\end{definition}

We further decompose $ \gsp $  into two parts.
\begin{definition}\label{def:decomposition-3}
Let
\begin{equation}\nonumber
 \gsp =:  \muvar_1  \fFz  +\rvar,
\end{equation}
be defined by $\langle\fFz , \Proj^1\rvar\rangle_{\sharp} =0$.
\end{definition}

The following theorem provides an upper bound of the binding
energy with an error term of the order $\mathcal{O}(\alpha^4)$.
Together with Theorem~\ref{thm:apriori-bound}, it establishes an
estimate up to the order $\alpha^3$ with an error term
$\mathcal{O}(\alpha^4)$.

\begin{theorem}[Upper bound on the binding energy]\label{thm:thm-main2}
1) Let $\Sigma = \inf\sigma(H)$. Then
\begin{equation}\label{eq:eq-binding-alpha3}
 \Sigma \geq \Sigma_0 - e_0 -\| \fF \|_*^2 + \mathcal{O}(\alpha^4) ,
\end{equation}
\begin{equation}\nonumber
 \Sigma_0 = -\alpha^2 \|\cEpp \|_*^2 + \alpha^3 (2\|\Aa\cEpp \|^2
 - 4\|\cEm \|_*^2 - 4 \|\cEp \|_*^2) + \mathcal{O}(\alpha^4) ,
\end{equation}
and $\fF$ defined by \eqref{def:def-F}.

2) For the components $\rvar$, $\cU $, $\Rvar$ of the ground state
$\gs$, and the coefficients $\eta_1$, $\eta_2$, $\eta_3$ and
$\muvar_1$ defined in
Definitions~\ref{def:decomposition-1}-\ref{def:decomposition-2},
holds
\begin{eqnarray}
 & & \|\rvar\|^2 = \mathcal{O}(\alpha^{\frac{33}{16}}),\
 \|H_f^\frac12 \rvar\|^2 =\mathcal{O}(\alpha^4),\
 \|(P-P_f)\rvar\|^2 = \mathcal{O}(\alpha^4),
 \label{eq:thm:thm-main2-1} \\
 & &  |\Proj^0\cU |^2 \geq
 1-c\alpha^2
 , \label{eq:thm-thm-main2-1bis}\\
 & & \|\Rvar\|_*^2 = \mathcal{O}(\alpha^4),\ \|\Rvar\|^2 =
 \mathcal{O}(\alpha^{\frac{33}{16}}), \\
 & & |\eta_{1,3} - 1|^2 \leq c\alpha,\ |\eta_2-1|^2\leq c\alpha^2,\
 |\muvar_1-1|^2\leq c\alpha\ .\label{eq:thm:thm-main2-2}
\end{eqnarray}
\end{theorem}
%
%

To prove this theorem, we will compute $\la H \gs,\gs\ra$
according to the decomposition of $\gs$ introduced in
Definitions~\ref{def:decomposition-1} to
\ref{def:decomposition-3}. Using
\begin{equation}\nonumber
 H = (P^2 -\frac{\alpha}{|x|})\ +\
 :(P_f + \sqrt{\alpha}A(0))^2:\ + H_f\ -\ 2\Re P \cdot (P_f +
 \sqrt{\alpha} A(0)) ,
\end{equation}
and due to the orthogonality of $\Sgsa $ and $\gsp$ , we obtain
\begin{equation}\label{eq:6.3}
  \langle H \gs , \gs \rangle = \langle H \Sgsa  \cU , \Sgsa
  \cU \rangle + \langle H \gsp,\gsp\rangle -
  4\Re \langle P.(P_f+\sqrt{\alpha} A(0)) \Sgsa  \cU , \gsp\rangle
  .
\end{equation}
We will estimate separately each term in \eqref{eq:6.3} in
subsections~\ref{S5.1}-\ref{S5.3}. These estimates will be used to
establish in subsection~\ref{S5.4} the proof of
Theorem~\ref{thm:thm-main2}.
%
%

\section{Estimate of the binding energy up
to $o(\alpha^5\log\alpha^{-1})$ term}\label{section-S6}

We develop here the proof of the difficult part in
Theorem~\ref{thm:main} which is the upper bound in
\eqref{eq:main1}, and which is stated in
Theorem~\ref{thm:thm-upper-bound} below for convenience.

Some technical aspects of this proof are detailed in
section~\ref{subsection:6-3} and appendices~\ref{subsection:6-1}
and \ref{prf-hgg}.

\begin{theorem}[Upper bound up to the order
$\alpha^5\log\alpha^{-1}$ for the binding
energy]\label{thm:thm-upper-bound}

For $\alpha$ small enough, we have
\begin{equation}\label{eq:main1-prime}
 \begin{split}
   \Sigma_0 -\Sigma \geq   & \frac14\alpha^2 +
   \dvar^{(1)} \alpha^3
   +
   \dvar^{(2)} \alpha^4 +
   \dvar^{(3)} \alpha^5\log\alpha^{-1}
   + o(\alpha^{5}\log\alpha^{-1}),
 \end{split}
\end{equation}
where $\dvar^{(1)}$, $\dvar^{(2)}$ and $\dvar^{(3)}$ are defined
in Theorem~\ref{thm:main}.
\end{theorem}

In order to prove this result, we need to refine the splitting of
the function $\gsp$ orthogonal to $\Sgsa\cU$ defined in
Definition~\ref{def:decomposition-1}. Therefore, we consider the
following decomposition

\begin{definition}\label{def:G-improved}
1) Pick
\begin{equation}\nonumber
 \muvar_2 = \left\{
 \begin{array}{ll}
   \alpha^{-1}
   \frac{\la\gsp,\cEpp \Proj^0 \gsp\ra_{\sharp}}
   {\la \cEpp  \Proj^0\gsp,\cEpp  \Proj^0 \gsp\ra_{\sharp}} & \mbox{ if }
   \| \Proj^0 \gsp\| >\alpha^{\frac32}, \\
   0 & \mbox{ if } \|\Proj^0 \gsp\|\leq \alpha^{\frac32}\ .
   \end{array}
   \right.
\end{equation}

\noindent 2) We split $\gsp$  into $\gsps $ and $\gspl $ as
follows: $\forall n\geq 0$, $\Proj^n\gsp = \Proj^n \gsps  +
\Proj^n \gspl $ and

\noindent for $n=0$,
\begin{equation}\nonumber
 \Proj^0 \gsps  = \Proj^0 \gsp\quad\mbox{and}\quad\Proj^0 \gspl =0  ,
\end{equation}
for $n=1$,
\begin{equation}\nonumber
 \Proj^1 \gsps  = \muvar_1\fFz
 \quad\mbox{and}\quad\la\Proj^1 \gspl , \fFz \ra_{\sharp}=0  ,
\end{equation}
for $n=2$,
\begin{equation}\nonumber
\begin{split}
 \Proj^2 \gsps  & =  \alpha\muvar_2\cEpp \Proj^0 \gsps
 + \sum_{i=1}^3 \alpha \muvar_{2,i} (H_f + P_f^2)^{-1} W_i
 \frac{\partial \Sgsa }{\partial x_i} ,\\
 & \mbox{with } W_i  = P_f^i \cEpp  - 2 \Ac . P_f (H_f + P_f^2)^{-1}
 (\Ac)^i \vac ,\\
 \Proj^2 \gspl  & =  \Proj^2\gsp - \Proj^2 \gsps ,\ \mbox{with }
 \la \Proj^2 \gspl , (H_f + P_f^2)^{-1} W_i
 \frac{\partial \Sgsa }{\partial x_i}\ra_{\sharp} =0 \ (i=1,2,3)
 ,\\
 & \mbox{and } \la\Proj^2 \gspl , \cEpp \Proj^0 \gsps \ra_{\sharp} =0 ,
\end{split}
\end{equation}
for $n=3$,
\begin{equation}\nonumber
\begin{split}
 \Proj^3 \gsps  & = \alpha \muvar_3 (H_f + P_f^2)^{-1} \Ac . \Ac \fFz  , \\
 \Proj^3 \gspl  & = \Proj^3\gsp - \Proj^3 \gsps ,\ \la \Proj^3 \gspl ,
 (H_f + P_f^2)^{-1} \Ac . \Ac \fFz \ra_{\sharp} =0 ,
\end{split}
\end{equation}
and for $n\geq 4$,
\begin{equation}\nonumber
 \Proj^n \gsps  = 0\quad\mbox{and}\quad \Proj^n \gspl  = \Proj^n\gsp .
\end{equation}
\end{definition}

The next step consists in establishing, in the next lemma, some a
priori estimates concerning the function $\gsp$ that give
additional information to those obtained in
\eqref{eq:thm:thm-main2-1} of Theorem~\ref{thm:thm-main2}.

\begin{lemma}\label{lem:improved-estimates}
The following estimates hold
\begin{equation}\nonumber
\begin{split}
 & \muvar_1 = 1 +\mathcal{O}(\alpha^\frac12) ,\\
 & |\muvar_2|\, \|\Proj^0 \gsps \| = \mathcal{O}(\alpha) ,\\
 & \muvar_{2,i} = \mathcal{O}(1),\ (i=1,2,3),\\
 & \|\Proj^0 \gsps \| =\mathcal{O}(\alpha), \\
 & \| P \Proj^0 \gsps \| = \mathcal{O}(\alpha^2) .
\end{split}
\end{equation}
\end{lemma}
%
%
%
\begin{proof}
To derive these estimates, we use Theorem~\ref{thm:thm-main2}.

The first equality is a consequence of \eqref{eq:thm:thm-main2-2}.

To derive the next two estimates, we first notice that
\eqref{eq:thm:thm-main2-1} yields
\begin{eqnarray*}
  \| P\Proj^2\rvar\|^2 &\leq& 2 \|(P-P_f)\rvar\|^2 + 2 \| P_f \Proj^2\rvar\|^2
  \\
  &\leq& 2 \|(P-P_f)\rvar\|^2 + 2 c \| H_f^\frac12 \Proj^2\rvar\|^2
  =\mathcal{O}(\alpha^4) ,
\end{eqnarray*}
therefore, using again \eqref{eq:thm:thm-main2-1}, we obtain
 $$
  \| (h_\alpha + e_0)^\frac12 \Proj^2\rvar\|^2 \leq \| P\Proj^2\rvar\|^2
  + e_0 \|\Proj^2\rvar\|^2 =\mathcal{O}(\alpha^4) ,
 $$
and thus, using from \eqref{eq:thm:thm-main2-1} that $\|
H_f^\frac12 \Proj^2\rvar\|^2 = \mathcal{O}(\alpha^4)$, we get
\begin{equation}\label{eq:estimate-*-1-1}
  \| \Proj^2\rvar\|_{\sharp}^2 = \|\Proj^2 \gsp\|_{\sharp}^2 = \mathcal{O}(\alpha^4) .
\end{equation}
We then write, using \eqref{eq:estimate-*-1-1} and the $\la \,
\cdot \, , \, \cdot \, \ra_\sharp$-orthogonality of $\Proj^2 \gsps
$ and $\Proj^2 \gspl $,
\begin{equation}
\| \Proj^2 \gsps \|_{\sharp}^2 \leq \| \Proj^2 \gsp\|_{\sharp}^2 =
\mathcal{O}(\alpha^4).
\end{equation}
Since $\| \Proj^2 \gsps \|_{\sharp} \leq \| \Proj^2 \gsps \|_{*}$,
and using \eqref{eq:0-with-Pf} of Lemma~\ref{lem:appendix-3}, we
obtain
\begin{equation}
\begin{split}
 & \mathcal{O}(\alpha^4) = \| \Proj^2 \gsps \|_{*}^2
 = \| \alpha\muvar_2\cEpp \Proj^0 \gsps \|_*^2
 + \|\alpha \sum_i \muvar_{2,i} (H_f+P_f^2)^{-1} W_i \frac{\partial
 \Sgsa }{\partial x_i} \|_*^2 \\
 & = \alpha^2 \| \cEpp  \|_*^2 |\muvar_2|^2\|\Proj^0 \gsps \|^2
 + \frac{\alpha^2}{3} \|\nabla \Sgsa \|^2 \sum_i |\muvar_{2,i}|^2
 \|(H_f + P_f^2)^{-1} W_i\|_*^2 .
\end{split}
\end{equation}
which gives
 $$
  \muvar_{2,i} = \mathcal{O}(1) \quad \mbox{and}\quad
  |\muvar_2|\, \|\Proj^0 \gsps \| = \mathcal{O}(\alpha) .
 $$

The last two estimates are consequences of
\eqref{eq:thm:thm-main2-1}.
\end{proof}


Eventually, to derive the lower bound on the quadratic form of
$\la H (\Sgsa \cU  + \gsp) , \Sgsa \cU  + \gsp\ra$, yielding the
upper bound \eqref{eq:main1-prime} of
Theorem~\ref{thm:thm-upper-bound}, we will follow the same
strategy as in Section~\ref{section-S5}, the only difference being
that now we have better a priori estimates on $\cU $ and $\gsp$.

The two main results to achieve this are stated below, with the
computation of the contribution to the ground state energy of the
cross term (Proposition~\ref{prop:prf-Hgpsi}) and of the direct
term $\langle H \gsp,\, \gsp\rangle$
(Proposition~\ref{prop:prop-hgg}). Equipped with this two
propositions, and using Theorem~\ref{thm:thm-main2} and
Lemma~\ref{lem:improved-estimates}, the proof of
Theorem~\ref{thm:thm-upper-bound} is only a straightforward
computation which is detailed in section~\ref{subsection:6-3}.

\begin{proposition}\label{prop:prf-Hgpsi}
We have
\begin{equation}\label{simple}
\begin{split}
 & 2 \Re \la H\gsp ,\Sgsa  \cU \ra
 \geq
 -\frac13 \alpha^4 \Re
 \sum_{i=1}^3 \overline{\muvar_{2,i}}
 \la (H_f + P_f^2)^{-1} P_f^i \cEpp , W_i\ra \\
 &  - 4\alpha \Re \la \nabla \Sgsa .   P_f \cEpp , \Proj^2
 \gspl \ra
 - 2\Re \overline{\muvar_1} \Proj^0\cU  \|\fFz \|_{\sharp}^2
 \\
 & -\frac23 \alpha^4 \Re
 \sum_{i=1}^3 \la\, (H_f+P_f^2)^{-\frac12} (\Aa)^i
 \cEpp , (H_f+P_f^2)^{-\frac12}
 (\Ac)^i \vac\ra \\
 & -\epsilon \alpha^2 \| (\gsps )^a \|^2 - \frac58 M[\gspl ]
 -\epsilon \alpha^5\log\alpha^{-1} -\epsilon \alpha^5|\muvar_3|^2 - |\muvar_1-1|
 c\alpha^4
 +\mathcal{O}(\alpha^5) .
\end{split}
\end{equation}
\end{proposition}

The proof of this Proposition is detailed in
Appendix~\ref{subsection:6-1}

\begin{proposition}\label{prop:prop-hgg}
\begin{equation}\label{eq:prop-hgg-main}
\begin{split}
 & \la H\gsp , \gsp\ra \geq (\Sigma_0 - e_0) \|\gsp\| ^2
 - 4\alpha \|(h_\alpha +e_0)^{-\frac12} Q_\alpha^\perp P \Aa \fF\|^2
 \\
 & + |\muvar_1|^2 \|\fFz \|_{\sharp}^2 + 2\alpha \|\Aa \fF\|^2
 + \frac{\alpha^4}{12} \sum_{i=1}^3 |\muvar_{2,i}|^2 \| (H_f+P_f^2)^{-1}
 W_i \|_{*}^2 \\
 & + \frac23 \alpha^4  \Re \sum_{i=1}^3
  \muvar_{2,i} \la P_f.\Aa
  (H_f+P_f^2)^{-1} W_i, (H_f+P_f^2)^{-1} (\Ac)^i\vac \ra\\
 & + 4 \alpha^\frac12\Re \la \Proj^2 \gspl , \Ac. P_f \fF\ra
 + M_1[\gsp]  ,
\end{split}
\end{equation}
where
\begin{equation}\label{simple2}
\begin{split}
 M_1[\gsp] =
 & (1\!-\!c_0\alpha) \| (h_\alpha +e_0)^\frac12 \Proj^0 (\gsps )^a\|^2
 + \frac{|\muvar_3+1|^2}{2} \alpha^2 \|\cEpp \|_*^2 \|\fFz \|^2 \\
 & - |\muvar_1-1| c\alpha^4 +\frac34 M[\gspl ]
 + o(\alpha^5\log\alpha^{-1})\ ,
\end{split}
\end{equation}
and $Q_\alpha^\perp$ is the projection onto the orthogonal
complement to the ground state $\Sgsa $ of the Schr\"odinger
operator $h_\alpha=-\Delta -\frac{\alpha}{|x|}$
\end{proposition}
%
%
%
The proof of this Proposition is detailed in
Appendix~\ref{prf-hgg}.

\section{Proof of Theorem \ref{thm:thm-main2}}\label{S7}
%

We prove Theorem \ref{thm:thm-main2} by bounding the individual
terms in the expression for the binding energy.
%
\subsection{Estimate of the term $\langle H \Sgsa  \cU ,
\Sgsa  \cU \rangle$}\label{S5.1}
%
%
%
%

\begin{lemma}\label{lem:lem-direct-fpsi}
\begin{equation}\nonumber
 \begin{split}
   \langle H \Sgsa  \cU , \Sgsa  \cU \rangle \geq & -e_0
   \|\cU \|^2 - \alpha^2 \| \cEpp \|_*^2 \|\Proj^0 \cU \|^2 + \alpha^2
   |\eta_2 - \Proj^0 \cU |^2 \|\cEpp \|_*^2 \\
   & + \alpha^3 |\eta_2|^2 \left(2 \| \Aa\cEpp \|^2 - 4 \|\cEm \|_*^2
   - 4\|\cEp \|_*^2\right)\\
   & + 4 \alpha^3 \left( |\eta_1-\eta_2|^2 \|\cEm \|_*^2
   + |\eta_3-\eta_2|^2 \|\cEp \|_*^2\right) \\
   & + c \alpha^{4}\log\alpha^{-1} \left(|\eta_1|^2 + |\eta_2|^2 + |\eta_3|^2
   + \| \Proj^0 \cU \|^2\right) +\frac12 \|\Rvar\|_*^2 .
 \end{split}
\end{equation}
\end{lemma}
%
%
%
\begin{proof}
The proof is a trivial modification of the one given for the lower
bound in \cite[Theorem~3.1]{BCVVi}. The only modification is that
we have a slightly weaker estimate in Lemma~\ref{Nf-exp-lemma-1}
on the photon number for the ground state. This is accounted for
by replacing the term of order $\alpha^4$ in
\cite[Theorem~3.1]{BCVVi} by a term of order
$\alpha^{4}\log\alpha^{-1}$. In addition, we need to use the
equality $\langle P.(P_f +\sqrt{\alpha} A(0))\Sgsa  \cU , \Sgsa
\cU \rangle =0$, due to the symmetry of $\Sgsa $.
\end{proof}

\subsection{Estimates for the cross term $- 4\Re \langle P.(P_f +
\sqrt{\alpha} A(0)) \Sgsa  \cU , \gsp\rangle$}\label{S5.2}
%
%
%
\begin{lemma}[$- 4\Re \langle P.P_f \Sgsa  \cU ,\gsp\rangle$ term]\label{lem:lem-cross-1}
\begin{equation}\nonumber
\begin{split}
  & - 4 \Re \langle P.P_f \Sgsa  \cU , \gsp\rangle \geq  - c\alpha^4
  (|\eta_1|^2+ |\eta_2|^2 +|\eta_3|^2 )  -
  \epsilon \|H_f^{\frac12}\rvar\|^2 -c  \alpha \| H_f^{\frac12}\Rvar\|^2.
\end{split}
\end{equation}
\end{lemma}
%
%
%
\begin{proof}
\noindent $\bullet$ For $n=1$ photon,
\begin{equation}
\begin{split}
 \la P.P_f \Proj^1 \Sgsa  \cU , \gsp\ra
 = \langle P.P_f (\eta_1 \alpha^{\frac32}\cEm  + \Proj^1\Rvar) \Sgsa ,
 \muvar_1  \fFz  + \Proj^1\rvar\rangle .
\end{split}
\end{equation}
Obviously
\begin{equation}
\begin{split}
 |\la P.P_f (\eta_1 \alpha^\frac32 \cEm  + \Proj^1\Rvar)\Sgsa  , \Proj^1\rvar\ra |
 \leq
 \epsilon \|H_f^{\frac12} \Proj^1\rvar\|^2 +
 c\alpha^5 |\eta_1|^2 +  c \alpha^2
 \|H_f^{\frac12} \Proj^1\Rvar\|^2 .
\end{split}
\end{equation}
Due to Lemma~\ref{appendix:lem-A1} holds $\| H_f^\frac12 (\fFz
-\fF)\|^2 = \mathcal{O}(\alpha^5)$, which implies
\begin{equation}\label{eq:c1}
\begin{split}
  & | \la  P.P_f (\eta_1\alpha^{\frac32} \cEm  + \Proj^1\Rvar) \Sgsa ,
  \muvar_1 \fFz  \ra |
  \leq |\muvar_1|^2 c\alpha^5 + c |\eta_1|^2 \alpha^5
  + c \alpha^2 \| H_f^\frac12 \Proj^1\Rvar\|^2 \\
  & + |\la  P.P_f (\eta_1\alpha^{\frac32} \cEm  + \Proj^1\Rvar) \Sgsa ,
  \muvar_1 \fF \ra | .
\end{split}
\end{equation}
For the last term on the right hand side of \eqref{eq:c1}, due to
the orthogonality of $\frac{\partial \Sgsa }{\partial x_i}$ and
$\frac{\partial \Sgsa }{\partial x_j}$, $i\neq j$, and the
equality $\|\frac{\partial \Sgsa }{\partial x_i}\| =
\|\frac{\partial \Sgsa }{\partial x_j} \|$, holds
\begin{equation}
\begin{split}
 & |\la  P.P_f (\eta_1\alpha^{\frac32} \cEm  + \Proj^1\Rvar) \Sgsa ,
  \muvar_1 \fF \ra | = \sum_{i=1}^3 \|\frac{\partial \Sgsa }{\partial
  x_i}\|^2
  |\la \eta_1 \alpha^\frac32 \cEm  + \Proj^1\Rvar , \muvar_1 P_f^i (\Ac
  \vac)^i \ra | \\
  & = c | \la \eta_1 \alpha^\frac32 \cEm  + \Proj^1\Rvar,\muvar_1 \Ac .
  P_f \vac\ra| = 0 .
\end{split}
\end{equation}

\noindent$\bullet$ For $n=2$ photons,
\begin{equation}\nonumber
\begin{split}
 & |\langle P.P_f \Proj^2 \Sgsa  \cU , \gsp\rangle| =
 |\langle P.P_f (\eta_2 \alpha\cEpp  + \Proj^2\Rvar) \Sgsa ,
 \Proj^2\rvar\rangle | \\
 & \leq c \alpha^4 |\eta_2|^2 +  \epsilon \|H_f^{\frac12} \Proj^2\rvar\|^2 +
  \alpha \|H_f^{\frac12} \Proj^2\Rvar\|^2  .
\end{split}
\end{equation}
\noindent$\bullet$ For $n=3$ photons, a similar estimate yields
\begin{equation}\nonumber
\begin{split}
 & |\langle P.P_f \Proj^3 \Sgsa  \cU , \gsp\rangle| =
 |\langle P.P_f (\eta_3 \alpha^{\frac32}\cEp  + \Proj^3\Rvar) \Sgsa ,
 \Proj^3\rvar\rangle | \\
 & \leq c \alpha^4 |\eta_3|^2 +  \epsilon \|H_f^{\frac12} \Proj^3\rvar\|^2 +
  \alpha \|H_f^{\frac12} \Proj^3\Rvar\|^2  .
\end{split}
\end{equation}
\noindent$\bullet$ For $n\geq 4$ photons,
\begin{equation}\nonumber
\begin{split}
 & |\langle P.P_f \Proj^{n\geq 4} \Sgsa  \cU , \gsp\rangle| \leq
 \epsilon \|H_f^{\frac12} \Proj^{n\geq 4}\rvar\|^2 +
  \alpha \|H_f^{\frac12} \Proj^{n\geq 4}\Rvar\|^2  .
\end{split}
\end{equation}
\end{proof}

\begin{lemma}[$-4 Re \langle \sqrt{\alpha} P. A(0) \Sgsa \cU ,\gsp\rangle$ term]
\label{lem:lem-cross-2}
\begin{equation}\nonumber
\begin{split}
 & -4 \Re \langle \sqrt{\alpha} P. A(0)) \Sgsa  \cU ,
 \gsp\rangle \geq -2 \Re \overline{\muvar_1} \Proj^0 \cU  \|\fFz \|_{\sharp}^2
 - \epsilon\|H_f^{\frac12}\rvar  \|^2 - \epsilon \alpha^2 \|\rvar\|^2 \\
 & -c\alpha \|H_f^{\frac12}\Rvar\|^2
 - c\alpha^4( |\eta_1|^2 + |\eta_2|^2 + |\eta_3|^2+
 |\muvar_1|^2)  + \mathcal{O}(\alpha^5\log\alpha^{-1}).
\end{split}
\end{equation}
\end{lemma}
%
%
\begin{proof}
We first estimate the term
 $$
 \Re\alpha^{\frac12} \langle P.\Ac
 \Sgsa  \cU , \gsp\rangle = \alpha\Re \sum_{n=0}^\infty \la
 P. \Ac \Sgsa  \Proj^n \cU , \Proj^{n+1} \gsp\ra .
 $$

\noindent$\bullet$ For $n=0$ photon, using the orthogonality
$\langle \fFz , \Proj^1\rvar \rangle_{\sharp} =0$, yields
\begin{equation}\label{eq:alpha3-PA+}
\begin{split}
 & - \Re\alpha^{\frac12} \langle P.\Ac
 \Proj^0 \Sgsa  \cU , \muvar_1 \fFz  + \Proj^1\rvar\rangle
 = - \frac12\Re \left((\Proj^0\cU ) \langle \fFz , \muvar_1 \fFz  + \Proj^1\rvar\rangle_{\sharp}\right)\\
 & = - \frac12 \Re \left(\overline{\muvar_1} \Proj^0\cU  \langle
 \fFz ,  \fFz \rangle_{\sharp}\right)\ .
\end{split}
\end{equation}
\noindent$\bullet$ For $n\geq 1$ photons,
\begin{equation}\label{eq:alpha3-PA+-2}
\begin{split}
 & | \sum_{n\geq 2} \Re\alpha^{\frac12} \langle P.\Ac
 \Proj^n \Sgsa \Rvar,\Proj^{n+1}\rvar \rangle |
 \leq c\alpha^3 \|\Rvar\|^2 + \epsilon \|H_f^{\frac12}\rvar\|^2\\
 & \leq  \epsilon \|H_f^{\frac12}\rvar\|^2 +
 c\alpha^5 (|\eta_1|^2 + |\eta_2|^2
 + |\eta_3|^2) +\mathcal{O}(\alpha^5\log\alpha^{-1}) ,
\end{split}
\end{equation}
where we used from Lemma~\ref{Nf-exp-lemma-1} that $\|\Rvar\|^2
\leq \mathcal{O}(\alpha^2\log\alpha^{-1}) + c\alpha^3(|\eta_1|^2
+|\eta_3|^2) + c\alpha^2 |\eta_2|^2$. We also have
\begin{equation}
\begin{split}
  & | \Re\alpha^{\frac12} \langle P
  (2\eta_1 \alpha^{\frac32} \cEm + \eta_2 \alpha \cEpp  +
  2\eta_3\alpha^{\frac32}  \cEp ) \Sgsa ,
  \Aa\rvar \rangle | \\
  & \leq \epsilon \|H_f^{\frac12}\rvar\|^2 +
  c (\alpha^6 |\eta_1|^2 + \alpha^5|\eta_2|^2
  \!+\! \alpha^6 |\eta_3|^2) .
\end{split}
\end{equation}

We next estimate the term $2\Re\alpha^{\frac12} \langle P.\Aa
\Sgsa  \cU , \gsp\rangle $. We first get
\begin{equation}
\begin{split}
  | \alpha^{\frac12} \Re \langle P. \Aa\Rvar \Sgsa ,\rvar\rangle |
  \leq \epsilon \alpha^2 \|\rvar\|^2 + c\alpha \|H_f^{\frac12}\Rvar\|^2.
\end{split}
\end{equation}
Then we write
\begin{equation}
\begin{split}
 |\alpha^{\frac12} \Re \langle (2 \eta_1 \alpha^{\frac32} \Aa\cEm
 + 2 \eta_3 \alpha^{\frac32}\Aa\cEp ) \nabla \Sgsa  ,\rvar\rangle|
 \leq \epsilon \alpha^2 \|\rvar\|^2 + c\alpha^4(|\eta_1|^2 + |\eta_3|^2) .
\end{split}
\end{equation}
We also have
\begin{equation}
\begin{split}
 & |\alpha^{\frac12} \Re \langle \eta_2 \alpha \Aa \cEpp
   . \nabla \Sgsa  , \muvar_1 \fFz  \rangle| \\
 & = |\alpha^{\frac12} \Re \langle \eta_2 \alpha H_f^{-\frac12}\Aa \cEpp
   . \nabla \Sgsa  , H_f^{\frac12} \muvar_1 \fFz  \rangle| \leq
 c\alpha^4 (|\eta_2|^2 + |\muvar_1|^2) ,
\end{split}
\end{equation}
since $H_f^{-\frac12}\Aa\cEpp \in L^2(\R^3)$ and $\|H_f^{\frac12}
\fFz \| = \mathcal{O}(\alpha^{\frac32})$. Finally we get
\begin{equation}\label{eq:alpha3-PA+-2-bis}
\begin{split}
  & | \alpha^{\frac12} \Re \langle \eta_2 \alpha \Aa \cEpp
   . \nabla \Sgsa  , \Proj^1\rvar \rangle| \\
  & = |\alpha^{\frac12} \Re \langle \eta_2 \alpha H_f^{-\frac12}\Aa \cEpp
   . \nabla \Sgsa  , H_f^{\frac12} \Proj^1\rvar \rangle|
  \leq \epsilon\|H_f^{\frac12} \Proj^1\rvar\|^2 + c\alpha^5 |\eta_2|^2
  .
\end{split}
\end{equation}
Collecting \eqref{eq:alpha3-PA+} to \eqref{eq:alpha3-PA+-2-bis}
concludes the proof.
\end{proof}
%
%
%
%
\subsection{Estimate for the term $\langle H\gsp,\gsp \rangle$}\label{S5.3}
%

\begin{lemma}\label{thm:thm-direct-gg}
\begin{equation}\label{hgg-15.1}
 \la H \gsp,\gsp\ra \geq (\Sigma_0 -e_0) \|\rvar\|^2 - c\alpha^4 |\muvar_1|^2
 + \frac12 M[\rvar] + |\muvar_1|^2 \|\fFz \|_{\sharp}^2 ,
\end{equation}
where $M[\ .\ ]$ is defined in Corollary~\ref{cor:cor-4-2}
\end{lemma}
%
%
\begin{proof}
Recall that
\begin{equation}\nonumber
\begin{split}
 & H  =  (P^2 - \frac{\alpha}{|x|}) + (H_f + P_f^2)
 - 2\Re \left(P . P_f\right)\\
 & - 2 \sqrt{\alpha} (P
 - P_f). A(0) + 2 \alpha \Ac. \Aa + 2\alpha\Re
 (\Aa)^2\ .
\end{split}
\end{equation}
Due to the orthogonality $\la \fFz ,\rvar\ra_{\sharp} = 0$ we get
\begin{equation}\label{eq:hgg-1}
\begin{split}
  & \la H\gsp,\gsp\ra \\
  &= \la Hr,r\ra +|\muvar_1|^2 \|\fFz \|_{\sharp}^2
  -e_0 |\muvar_1|^2 \| \fFz \|^2 -e_0 \muvar_1 \la\rvar, \fFz \ra
  + 2\alpha |\muvar_1|^2 \|\Aa\fFz \|^2\\
  & -2\Re\la P.P_f\fFz , \fFz \ra
  + 2\alpha\Re \la \Aa . \Aa\rvar, \muvar_1 \fFz \ra - 4 \sqrt{\alpha}
  \Re \la P. \Aa\rvar, \muvar_1\fFz \ra\\
  &  - 4 \sqrt{\alpha} \Re\la P.\Ac\rvar,
  \muvar_1 \fFz \ra
  + 4\Re \sqrt{\alpha} \la P_f. A(0)\rvar,
  \muvar_1 \fFz \ra + 4\alpha \Re\la \Ac. \Aa\rvar, \muvar_1 \fFz \ra \\
  & -4\Re \la P.P_f \fFz ,\rvar\ra .
\end{split}
\end{equation}
For the first term on the right hand side of \eqref{eq:hgg-1}, we
have, from Corollary~\ref{cor:cor-4-2}
\begin{equation}
 \la H\rvar,\rvar\ra \geq (\Sigma_0 -e_0) \|\rvar\|^2 + M[\rvar] .
\end{equation}
According to Lemma~\ref{appendix:lem-A0}, we obtain
\begin{equation}
  - e_0\muvar_1 \la\rvar,\fFz \ra - e_0 |\muvar_1|^2 \|\fFz \|^2
  \geq -\epsilon \alpha^2 \|\rvar\|^2 - c| \muvar_1 |^2 \alpha^5\log\alpha^{-1} .
\end{equation}
The next term, namely $2\alpha|\muvar_1|^2 \|\Aa\fFz \|^2$, is
positive.

Due to the symmetry in $x$-variable,
\begin{equation}
 \la P.P_f \fFz , \fFz \ra =0 .
\end{equation}
The term $2\alpha\Re\la\Aa.\Aa\rvar, \muvar_1 \fFz \ra$ is
estimated as
\begin{equation}
 2\alpha\Re\la\Aa.\Aa\rvar, \muvar_1 \fFz \ra \geq - c\alpha^2
 \|\muvar_1\fFz \|^2 - \frac14 \nu \| H_f^{\frac12}\rvar\|^2
 = -\frac14 \nu \| H_f^{\frac12}\rvar\|^2 - c|\muvar_1|^2
 \alpha^5\log\alpha^{-1} .
\end{equation}
Due to Lemma~\ref{appendix:lem-A0} and \cite[Lemma~A4]{GLL},
\begin{equation}
 |\alpha^{\frac12} \la \Aa\rvar, P\muvar_1\fFz \ra|\leq
 \frac{\nu}{8} \|H_f^\frac12\rvar\|^2 + c |\muvar_1|^2
 \alpha^6\log\alpha^{-1}.
\end{equation}
The next term we have to estimate fulfils
\begin{equation}
 |\alpha^{\frac12} \la P. \Ac\rvar, \muvar_1\fFz \ra|
 \leq \epsilon \|P \Proj^0\rvar\|^2 + c\alpha |\muvar_1|^2 \|\Aa \fFz \|^2
 \leq \epsilon \|(P-P_f)\rvar\|^2 - c|\muvar_1|^2 \alpha^4 .
\end{equation}
We have
\begin{equation}
\begin{split}
  & \Re \sqrt{\alpha} \la P_f. A(0)\rvar,
  \muvar_1 \fFz \ra = \Re \sqrt{\alpha} \la P_f. \Aa \Proj^2\rvar,
  \muvar_1 \fFz \ra + \Re \sqrt{\alpha} \la P_f. \Ac \Proj^0\rvar,
  \muvar_1 \fFz \ra \\
  & = \Re \sqrt{\alpha} \la P_f. \Aa \Proj^2\rvar,
  \muvar_1 \fFz \ra + \Re \sqrt{\alpha} \la \Ac . P_f \Proj^0\rvar,
  \muvar_1 \fFz \ra \\
\end{split}
\end{equation}
Obviously, $\la \Ac . P_f \Proj^0\rvar, \muvar_1 \fFz \ra=0$, and
the first term is bounded by
\begin{equation}
\begin{split}
  & | \Re \sqrt{\alpha} \la P_f. \Aa \Proj^2\rvar,
  \muvar_1 \fFz \ra | \leq \epsilon \|H_f^\frac12 \Proj^2\rvar\|^2
  + c \alpha |\muvar_1|^2 \|P_f \fFz \|^2\\
  & \leq \epsilon \|H_f^\frac12 \Proj^2\rvar\|^2
  + c  |\muvar_1|^2 \alpha^4.
\end{split}
\end{equation}
For the term $\alpha \Re\la \Ac. \Aa\rvar, \muvar_1 \fFz \ra $ we
obtain
\begin{equation}
 \alpha \Re\la \Ac. \Aa\rvar, \muvar_1 \fFz \ra \geq
 -c\alpha^2 \|H_f^\frac12 \muvar_1 \fFz \|^2
 |\muvar_1|^2 -\epsilon \|H_f^{\frac12}\rvar\|^2
 = -c\alpha^5 |\muvar_1|^2 -\epsilon \|H_f^{\frac12}\rvar\|^2.
\end{equation}
According to \eqref{appendix:lem7-2} of
Lemma~\ref{appendix:lem-A0}, the last term we have to estimate
fulfils
\begin{equation}\label{eq:hgg-2}
\begin{split}
\Re\la P.P_f \muvar_1 \fFz ,\rvar\ra \leq \epsilon \|H_f^\frac12
\Proj^1 \rvar\|^2 + c\|P |P_f|^\frac12 \muvar_1 \fFz \|^2 \leq
\epsilon \|H_f^\frac12 \Proj^1\rvar\|^2 + c \alpha^5 |\muvar_1|^2
.
\end{split}
\end{equation}
Collecting the estimates \eqref{eq:hgg-1} to \eqref{eq:hgg-2}
yields \eqref{hgg-15.1}.
\end{proof}

\subsection{Upper bound on the binding energy
with error term $\mathcal{O}(\alpha^4)$}\label{S5.4}

We first establish a lemma that we shall need in the proof of
Theorem~\ref{thm:thm-main2} in order to improve the error term
from $\mathcal{O}(\alpha^4\log\alpha^{-1})$ to
$\mathcal{O}(\alpha^4)$.
\begin{lemma}\label{qmqm}
If $\|\Rvar\|_*^2 =\mathcal{O}(\alpha^4\log\alpha^{-1})$ and
$\|H_f^\frac12\rvar\|^2 =\mathcal{O}(\alpha^4\log\alpha^{-1})$
hold then we have
\begin{eqnarray}
  & & \|N_f^\frac12\Rvar\|^2 = \mathcal{O}(\alpha^{\frac{33}{16}})
 \label{eq:improved-photon-number-1}\\
  & & \|N_f^\frac12\rvar\|^2 = \mathcal{O}(\alpha^{\frac{33}{16}})
 \label{eq:improved-photon-number-2}
\end{eqnarray}
\end{lemma}

\begin{proof}
We note that from Definition~\ref{def:decomposition-1}, we have
\begin{equation}\nonumber
 \| a_\lambda(k) \cU \|^2 \leq \| a_\lambda(k) \gs\|^2
 \leq \frac{c\alpha^{-\frac32} \chi_\Lambda(|k|)}{|k|} ,
\end{equation}
where in the last inequality, we used \eqref{a-psi-fund-2}. Taking
into account that
\begin{equation}\nonumber
\Rvar =\cU -2 \eta_1\alpha^\frac32\cEm - \eta_2\alpha\cEpp -2
\eta_3\alpha^\frac32\cEp ,
\end{equation}
where
\begin{equation}\nonumber
    \| a_\lambda(k)\cEm \|^2\leq \frac{c\chi_\Lambda(|k|)}{|k|}
    \; , \;
    \| a_\lambda(k)\cEpp \|^2\leq
    \frac{c\chi_\Lambda(|k|)(1+\left|\log|k|\,\right|)}{|k|}
\end{equation}
\begin{equation}\nonumber
    \| a_\lambda(k) \cEp \|^2\leq \frac{c\chi_\Lambda(|k|)}{|k|} \,,
\end{equation}
and using \eqref{eq:bound-on-etas}, we arrive at
\begin{equation}\nonumber
 \| a_\lambda(k)\Rvar\|^2 \leq
 \frac{c\alpha^{-\frac32}\chi_\Lambda(|k|)(1+|\log|k|\,|)}{|k|} .
\end{equation}
For the expected photon number of $\Rvar$ thus holds
\begin{equation}\nonumber
\begin{split}
  \| N_f^\frac12\Rvar \|^2 &= \sum_\lambda \int \|a_\lambda(k)\Rvar\|^2 \d
  k\\
  & \leq \sum_\lambda \int_{|k|\leq \alpha^{\frac{15}{8}}}
  \frac{c\alpha^{-\frac32}(1+\left|\log|k|\,\right|)} {|k|} \d k +
  \int_{|k| > \alpha^{\frac{15}{8}}} |k|^{-1}\,
  |k|  \,
  \|a_\lambda(k)\Rvar\|^2 \d k\\
  & \leq c\alpha^{\frac{17}{8}} +
  c\alpha^{-\frac{15}{8}}
  \|H_f^\frac12
 \Rvar\|^2 \leq c\alpha^{\frac{33}{16}} ,
\end{split}
\end{equation}
using \eqref{eq:add-est-w-log} in the last inequality. The
relation \eqref{eq:improved-photon-number-2} can be proved
similarly, using
 $$
  \|a_\lambda(k) \fFz \|^2 \leq c \frac{\alpha^{-1}}{|k|}\ .
 $$
 This concludes the proof of the lemma.
\end{proof}

\subsection{Concluding the proof of Theorem~\ref{thm:thm-main2}}

The proof of Theorem~\ref{thm:thm-main2} is obtained in the
following two steps. We first show that the estimate holds with an
error term $\mathcal{O}(\alpha^4\log\alpha^{-1})$. In a second
step, using Lemma~\ref{qmqm}, we improve to an error term
$\mathcal{O}(\alpha^4)$. Then, we derive the estimates
\eqref{eq:thm:thm-main2-1}-\eqref{eq:thm:thm-main2-2} that shall
be used in the next section for the computation of the binding
energy to higher order.

\noindent $\bullet$ \underline{\textit{Step 1:}} We first show that
\eqref{eq:eq-binding-alpha3} holds with an error estimate of the
order $\alpha^4\log\alpha^{-1}$.

Collecting Lemmata~\ref{lem:lem-direct-fpsi},
\ref{lem:lem-cross-1}, \ref{lem:lem-cross-2} and
Lemma~\ref{thm:thm-direct-gg} yields
\begin{equation}\label{eq:thm:thm-bis-1}
 \begin{split}
   & \langle H \gs, \gs\rangle \\
   & \geq -e_0
   \|\cU \|^2 - \alpha^2 \| \cEpp \|_*^2 |\Proj^0 \cU |^2
   + \alpha^3 |\eta_2|^2 (2 \| \Aa\cEpp \|^2 - 4 \|\cEm \|_*^2 - 4\|\cEp \|_*^2)\\
   &   + \alpha^2
   |\eta_2 - \Proj^0 \cU  |^2 \|\cEpp \|_*^2
   + 4 \alpha^3 ( |\eta_1-\eta_2|^2 \|\cEm \|_*^2
   + |\eta_3-\eta_2|^2 \|\cEp \|_*^2) \\
   & - c \alpha^{4}\log\alpha^{-1} (|\eta_1|^2 + |\eta_2|^2 + |\eta_3|^2
   + | \Proj^0 \cU |^2) +\frac14 \|\Rvar\|_*^2 \\
   & + |\muvar_1|^2 \|\fFz \|_{\sharp}^2 + (\Sigma_0 - e_0)\|\rvar\|^2
   +\frac14
   M[\rvar] - c\alpha^4 |\muvar_1|^2 \\
   & -2\Re (\muvar_1
   \Proj^0\cU ) \|\fFz \|_{\sharp}^2 + \mathcal{O}(\alpha^5\log\alpha^{-1}).
 \end{split}
\end{equation}
We first estimate
\begin{equation}\label{eq:thm:thm-bis-2}
\begin{split}
  |\muvar_1|^2 \|\fFz \|_{\sharp}^2 -2\Re (\overline{\muvar_1}
   \Proj^0\cU ) \|\fFz \|_{\sharp}^2 - c\alpha^4 |\muvar_1|^2
   \\
   \geq -\|\fFz \|_{\sharp}^2 |\Proj^0\cU |^2 + \frac{ |\overline{\muvar_1} -
   \Proj^0\cU |^2}{2} \|\fFz \|_{\sharp}^2 +\mathcal{O}(\alpha^4).
\end{split}
\end{equation}
Moreover, since $|\Proj^0 \cU |\leq 1$, we replace in
\eqref{eq:thm:thm-bis-1} $-\alpha^2\|\cEpp \|_*^2|\Proj^0\cU |^2$
by $-\alpha^2\|\cEpp \|_*^2$ and in \eqref{eq:thm:thm-bis-2}
$-\|\fFz \|_{\sharp}^2 |\Proj^0\cU |^2$ by $-\|\fFz
\|_{\sharp}^2$. In addition, using the inequalities
\begin{equation}\nonumber
  |\eta_2 - \eta_j|^2 \geq \frac12 |\eta_j -\Proj^0\cU |^2 - |\eta_2
  -\Proj^0\cU |^2\quad\mbox{and}\quad |\eta_j|^2 \leq 2
  |\eta_j-\Proj^0\cU |^2 + 2
\end{equation}
for $j=1,2,3$ yields that for some $c>0$,
\begin{equation}\label{eq:est-w-log-1}
 \begin{split}
   & \langle H \gs, \gs\rangle \\
   & \geq -e_0
   \|\cU \|^2 - \alpha^2 \| \cEpp \|_*^2
   + \alpha^3 |\eta_2|^2 (2 \| \Aa\cEpp \|^2
   - 4 \|\cEm \|_*^2 - 4\|\cEp \|_*^2)
   - \|\fFz \|_{\sharp}^2 \\
   &   + c \alpha^2
   |\eta_2 - \Proj^0 \cU  |^2 \|\cEpp \|_*^2
   + c \alpha^3 ( |\eta_1-\Proj^0\cU  |^2  + |\eta_3-\Proj^0\cU |^2
   )+ c\alpha^3 |\muvar_1-\Proj^0\cU |^2\\
   &  +\frac14 \|\Rvar\|_*^2
   + (\Sigma_0 - e_0)\|\rvar\|^2 +\frac14
   M[\rvar] + \mathcal{O}(\alpha^{4}\log\alpha^{-1}).
 \end{split}
\end{equation}
Comparing this expression with \eqref{eq:eq:apriori-1} of
Theorem~\ref{thm:apriori-bound} gives
\begin{equation}\label{eq:bound-on-etas}
   \max \{ |\eta_1|,\  |\eta_2|,\,  |\eta_3 \} \leq 2\ .
\end{equation}
and
\begin{equation}\label{eq:est-w-log-2}
  \Sigma = \Sigma_0 -e_0 -\|\fFz \|_{\sharp}^2
  +\mathcal{O}(\alpha^4\log\alpha^{-1}) .
\end{equation}
Finally, by Lemma~\ref{appendix:lem-A1}, we can replace $\|\fFz
\|_{\sharp}$ by $\|\fF\|_*$ in the above equality, which proves
\eqref{eq:eq-binding-alpha3} with an error term
$\mathcal{O}(\alpha^4\log\alpha^{-1})$.
\\

\noindent $\bullet$ \underline{\textit{Step 2:}} We now show that the error
term does not contain any $\log\alpha^{-1}$ term.

From \eqref{eq:est-w-log-1} we obtain
\begin{equation}\label{eq:add-est-w-log}
 \|\Rvar\|_*^2 =\mathcal{O}(\alpha^4\log\alpha^{-1})
 \quad\mbox{and}\quad
 \|H_f^\frac12\rvar\|^2 =\mathcal{O}(\alpha^4\log\alpha^{-1}) .
\end{equation}
According to Lemma~\ref{qmqm}, this implies
\begin{equation}\nonumber
  \|N_f^\frac12\Rvar\|^2 =
  \mathcal{O}(\alpha^{\frac{33}{16}})\quad\mbox{and}\quad
  \|N_f^\frac12\rvar\|^2 = \mathcal{O}(\alpha^{\frac{33}{16}})\ .
\end{equation}

Thus, we have
\begin{equation}\nonumber
\begin{split}
   & \|N_f^\frac12 \cU \|^2 \\
   & \leq 4|\eta_1|^2 \alpha^3\| N^\frac12 \cEm \|^2
   + |\eta_2|^2 \alpha^2 \| N ^\frac12 \cEpp \|^2
   + 4 |\eta_3|^2 \alpha^3 \| N^\frac12 \cEp \|^2 + \| N^\frac12
   \Rvar\|^2\\
   & =\mathcal{O}(\alpha^2) ,
\end{split}
\end{equation}
which implies that in Lemma~\ref{lem:lem-direct-fpsi} we can
replace the term $c \alpha^4\log\alpha^{-1}(|\eta_1|^2 +
|\eta_2|^2 + |\eta_3|^2 + |\Proj^0\cU |^2)$ with $c \alpha^4
(|\eta_1|^2 + |\eta_2|^2 + |\eta_3|^2 + |\Proj^0\cU |^2)$ and
consequently in \eqref{eq:est-w-log-1} and \eqref{eq:est-w-log-2},
the term $\mathcal{O}(\alpha^4\log\alpha^{-1})$ can be replaced by
$\mathcal{O}(\alpha^4)$. This proves \eqref{eq:eq-binding-alpha3}.
The estimates
\eqref{eq:thm:thm-main2-1}-\eqref{eq:thm:thm-main2-2} follow from
\eqref{eq:eq-binding-alpha3} and \eqref{eq:est-w-log-1} with
$\mathcal{O}(\alpha^4\log\alpha^{-1})$ replaced with
$\mathcal{O}(\alpha^4)$.
We thus arrive at the proof of Theorem~\ref{thm:thm-main2}.
\qed

\newpage

\section{Proof of Theorem~\ref{thm:thm-upper-bound}}
\label{subsection:6-3}

In this section, we prove the
upper bound on the binding energy up to the order
$\alpha^5\log\alpha^{-1}$ provided in Theorem~\ref{thm:thm-upper-bound}.

We have
\begin{equation}\label{eq:collect-1}
 \la H \gs, \gs\ra  \geq  (\Sigma_0 -e_0)\|\cU \|^2 + 2\Re \la H \gsp,
 \Sgsa \cU \ra + \la H \gsp , \gsp\ra .
\end{equation}
The estimates for the last two terms in \eqref{eq:collect-1} are
given in Propositions~\ref{prop:prf-Hgpsi} and
\ref{prop:prop-hgg}. We will bound below this expression by
considering separately the terms involving the parameters
$\muvar_1$, $\muvar_{2,i}$, and $\muvar_3$.

$\bullet$ We first estimate the second term on the right hand side
of \eqref{simple} together with the seventh term on the right hand
side of \eqref{eq:prop-hgg-main}. We have
\begin{equation}
\begin{split}
 & - 4\alpha \Re \la \nabla \Sgsa .   P_f \cEpp , \Proj^2
 \gspl \ra + 4 \alpha^\frac12\Re \la \Proj^2 \gspl , \Ac. P_f \fF\ra \\
 & = -4 \alpha \Re \la \Proj^2 \gspl , (\sum_{i=1}^3 P_f^i \cEpp  -
 2 \Ac. P_f (H_f+P_f^2)^{-1} (\Ac)^i \vac ) \frac{\partial
 \Sgsa }{\partial x_i}\ra \\
 & = - 4 \alpha \sum_{i=1}^3 \Re \la \Proj^2 \gspl , (H_f+P_f^2)^{-1} W_i
 \frac{\partial \Sgsa }{\partial x_i}\ra_* .
\end{split}
\end{equation}
Using the $\la \, \cdot \, , \, \cdot \, \ra_\sharp$-orthogonality
of $\Proj^2 \gspl $ and $(H_f+P_f^2)^{-1} W_i \frac{\partial \Sgsa
}{\partial x_i}$, the last expression can be estimated as
\begin{equation}
\begin{split}
 & - 4 \alpha \sum_{i=1}^3 \Re \la \Proj^2 \gspl , (h_\alpha + e_0)(H_f+P_f^2)^{-1} W_i
 \frac{\partial \Sgsa }{\partial x_i}\ra .
\end{split}
\end{equation}
By the Schwarz inequality, this term is bounded below by
\begin{equation}
\begin{split}
 & - \epsilon \alpha^2 \|\Proj^2 \gspl \|^2
 - c \sum_{i=1}^3 \|(h_\alpha + e_0)
 \frac{\partial \Sgsa }{\partial x_i}\|^2
 = - \epsilon \alpha^2 \|\Proj^2 \gspl \|^2-\mathcal{O}(\alpha^6) .
\end{split}
\end{equation}

$\bullet$ Next, we collect all the terms involving $\muvar_1$ in
\eqref{simple} and \eqref{eq:prop-hgg-main}. This yields
\begin{equation}
\begin{split}
 & - 2\Re \overline{\muvar_1} \Proj^0\cU  \|\fFz \|_{\sharp}^2
 + |\muvar_1|^2 \|\fFz \|_{\sharp}^2 -|\muvar_1-1| c\alpha^4 \\
 & \geq -|\Proj^0\cU |^2  \|\fFz \|_{\sharp}^2
 + |\muvar_1 - \Proj^0\cU |^2 \|\fFz \|_{\sharp}^2 -|\muvar_1-1| c\alpha^4\\
\end{split}
\end{equation}
Notice that from Theorem~\ref{thm:thm-main2} we have $|\Proj^0\cU
|^2 = 1+\mathcal{O}(\alpha^2)$; moreover, we have $|\Proj^0\cU | =
1+\mathcal{O}(\alpha^2)$. This yields
\begin{equation}
\begin{split}
 & - 2\Re \overline{\muvar_1} \Proj^0\cU  \|\fFz \|_{\sharp}^2
 + |\muvar_1|^2 \|\fFz \|_{\sharp}^2 -|\muvar_1-1| c\alpha^4 \\
 & \geq - \|\fFz \|_{\sharp}^2
 + |\muvar_1 -1|^2 c'\alpha^3 -|\muvar_1-1| c\alpha^4
 + \mathcal{O}(\alpha^5) \geq - \|\fFz \|_{\sharp}^2 +
 \mathcal{O}(\alpha^5) .
\end{split}
\end{equation}

$\bullet$ We now collect and estimates the terms in \eqref{simple}
and \eqref{eq:prop-hgg-main} involving $\muvar_{2,i}$. We get
\begin{equation}
\begin{split}
 & -\frac13 \alpha^4 \Re
 \sum_{i=1}^3 \muvar_{2,i}
 \la (H_f+P_f^2)^{-1} W_i, P_f^i \cEpp \ra
 + \frac{\alpha^4}{12} \sum_{i=1}^3 |\muvar_{2,i}|^2 \| (H_f+P_f^2)^{-1}
 W_i \|_{*}^2 \\
 & + \frac23 \alpha^4  \Re \sum_{i=1}^3
  \muvar_{2,i} \la
   (H_f+P_f^2)^{-1} W_i,  P_f.\Ac (H_f+P_f^2)^{-1} (\Ac)^i\vac \ra \\
 & = -\frac13 \alpha^4 \Re\sum_{i=1}^3
 \muvar_{2,i} \|(H_f+P_f^2)^{-1} W_i \|_*^2
  + \frac{\alpha^4}{12} \sum_{i=1}^3 |\muvar_{2,i}|^2 \| (H_f+P_f^2)^{-1}
 W_i \|_{*}^2 \\
 & \geq - \frac{\alpha^4}{3} \sum_{i=1}^3 \| (H_f+P_f^2)^{-1} W_i \|_{*}^2
\end{split}
\end{equation}

$\bullet$ Collecting in \eqref{simple} and
\eqref{eq:prop-hgg-main} the terms containing $\muvar_3$ yields
\begin{equation}
\begin{split}
 \frac{|\muvar_3+1|^2}{2} \alpha^2 \| \cEpp \|_*^2 \|\fFz \|^2
 - \epsilon \alpha^5 |\muvar_3|^2 \geq c_1
 \alpha^5\log\alpha^{-1} |\muvar_3+1|^2-
 \epsilon \alpha^5 |\muvar_3|^2 \geq -c_2 \alpha^5 ,
\end{split}
\end{equation}
where $c_1$ and $c_2$ are positive constants.

$\bullet$ The fifth term on the right hand side of \eqref{simple}
and the first term on the right hand side of \eqref{simple2} are
estimated, for $\alpha$ small enough, as
\begin{equation}
\begin{split}
 (1-c_0\alpha) \| (h_\alpha + e_0)^\frac12 \Proj^0 (\gsps )^a\|^2
 - \epsilon \alpha^2 \|(\gsps )^a\|^2 \geq (\frac{\delta}{2} - \epsilon
 \alpha^2) \|(\gsps )^a\|^2 \geq 0 ,
\end{split}
\end{equation}
with $\delta = \frac{3}{32}\alpha^2$, and where we used that
$(\gsps )^a$ is orthogonal to $\Sgsa $.

$\bullet$ Substituting the above estimates in \eqref{eq:collect-1}
yields
\begin{equation}\label{eq:c-0}
\begin{split}
 & \la H \gs, \gs\ra \\
 & \geq
 (\Sigma_0 - e_0) \|\cU \|^2 + (\Sigma_0-e_0) \|\gsp\| ^2
 - \|\fFz \|_{\sharp}^2
 - \alpha^4 \sum_{i=1}^3 \| (H_f+P_f^2)^{-1} W_i \|_{*}^2 \\
 & -\frac23 \alpha^4 \Re
 \sum_{i=1}^3 \la\, (H_f+P_f^2)^{-\frac12} (\Aa)^i
 \cEpp , (H_f+P_f^2)^{-\frac12}
 (\Ac)^i \vac\ra \\
 & - 4\alpha \| (h_\alpha + e_0)^{-\frac12} Q_\alpha^\perp P \Aa \fF\|^2
 + 2\alpha \|\Aa \fF\|^2 + o(\alpha^5\log\alpha^{-1})\ ,
\end{split}
\end{equation}
where $Q_\alpha^\perp$ is the projection onto the orthogonal
complement to the ground state $\Sgsa $ of the Schr\"odinger
operator $h_\alpha=-\Delta -\frac{\alpha}{|x|}$.

 To complete the proof of
Theorem~\ref{thm:thm-upper-bound} we first note that
\begin{equation}\label{eq:c-1}
  \|\cU \|^2 + \|\gsp\|^2 = \| \gs\|^2.
\end{equation}
Moreover, according to Lemma~\ref{appendix:lem-A1}
\begin{equation}\label{eq:c-2}
\begin{split}
 - \| \fFz \|_{\sharp}^2 = - \| \fF\|_*^2 + \frac{1}{3\pi}
 \| (h_1+\frac14)^\frac12 \nabla \Sgso\|^2 \alpha^5 \log\alpha^{-1}
 + o(\alpha^5 \log\alpha^{-1}), \\
\end{split}
\end{equation}
and
\begin{equation}\label{eq:c-3}
  \| \fF\|_*^2 =
  \frac{\alpha^3}{2\pi} \int_0^\infty \frac{\chi_\Lambda(t)}{1+t} \d t
  = \dvar^{(1)} \alpha^3 .
\end{equation}
In addition, we have the following identities ($i=1,2,3$)
\begin{equation}\label{eq:c-4}
\begin{split}
 & \| (H_f + P_f^2)^{-1} W_i \|_{*}^2 = \\
 & \|(H_f + P_f^2)^{-\frac12}
  \Big(
  2 \Ac.P_f (H_f + P_f^2)^{-1} (\Ac)^i
  - P_f^i (H_f+P_f^2)^{-1}\Ac. \Ac \Big)\vac\|^2 ,
\end{split}
\end{equation}
and
\begin{equation}\label{eq:c-5}
\begin{split}
 & -\frac23 \alpha^4 \Re
 \sum_{i=1}^3 \la\, (H_f+P_f^2)^{-\frac12} \Aa \cEpp , (H_f+P_f^2)^{-\frac12}
 (\Ac)^i \vac\ra \\
 & = -\frac{2}{3} \alpha^4
 \Re \sum_{i=1}^3  \la \Aa (H_f+P_f^2)^{-1}\Ac
   . \Ac \vac, (H_f + P_f^2)^{-1} (\Ac)^i \vac \ra\ .
\end{split}
\end{equation}
We also have
\begin{equation}\label{eq:c-6}
 \begin{split}
   & -4 \alpha \| (h_\alpha + e_0)^{-\frac12}
   Q_\alpha^\perp P \Aa \fF\|^2
   =
   - 4 \alpha^4 a_0^2 \| (-\Delta -\frac{1}{|x|}+\frac14 )^{-\frac12}
   Q_1^\perp \Delta \Sgso\|^2 ,
 \end{split}
\end{equation}
with
\begin{equation}\nonumber
  a_0 = \int
  \frac{k_1^2+k_2^2}{4\pi^2|k|^3} \frac{2}{|k|^2
  +|k|} \chi_\Lambda (|k|)\, \d k_1 \d k_2 \d k_3 ,
\end{equation}
and
\begin{equation}\label{eq:c-7}
 2\alpha \| \Aa \fF \|^2 =
   \frac23 \alpha^4 \sum_{i=1}^3
 \| A^- (H_f + P_f^2)^{-1} (A^+)^i \vac\|^2 .
\end{equation}
Substituting \eqref{eq:c-1}-\eqref{eq:c-7} into \eqref{eq:c-0}
finishes the proof of Theorem~\ref{thm:thm-upper-bound} and thus
the proof of the upper bound in Theorem~\ref{thm:main}. \qed

$\;$\\

\section{Proof of Theorem~\ref{thm:main}: Lower bound up to $o(\alpha^5\log\alpha^{-1})$ for
the binding energy}\label{S6.4}

In this section, we prove a
lower bound for $\Sigma_0 - \Sigma$ in
Theorem~\ref{thm:main} which coincides with the upper bound given
in \eqref{eq:c-0}.
To this end, it suffices to compute
\begin{equation}\nonumber
  \frac{\la (H-\Sigma_0 + e_0) \widetilde\Phi^{\mathrm{trial}},\
  \widetilde\Phi^{\mathrm{trial}}\ra}{\| \widetilde\Phi^{\mathrm{trial}}\|^2},
\end{equation}
with the following trial function
 $$
  \widetilde\Phi^{\mathrm{trial}} = \Sgsa  \Psz + \Psp,
 $$
where $\Psz$ is  ground state of the operator $T(0)$ (defined in
\eqref{def:T(P)}), with the normalization $\Proj^0 \Psz =\vac$,
$\Sgsa $ is the normalized ground state of $h_\alpha=-\Delta -
\frac{\alpha}{|x|}$, and $\Psp$  is defined by
\begin{equation}
\begin{split}
  & \Proj^0\Psp=2{\alpha^\frac12}    (h_\alpha +e_0)^{-1}
  Q_\alpha^\perp P.\Aa\fF,\
  \Proj^1\Psp =\fFz ,\\
  & \Proj^2 \Psp=\alpha\cEpp  \Proj^0\Psp +\sum_{i=1}^3 2 \alpha
  (H_f+P_f^2)^{-1} W_i \frac{\partial \Sgsa }{\partial x_i},\\
  &\Proj^3\Psp = -\alpha (H_f + P_f^2)^{-1} \Ac. \Ac \fFz  .
\end{split}
\end{equation}
Where $\fF$, $\fFz $, $\cEpp $ and $W_i$ are defined as in
Sections~\ref{section-S5} and \ref{section-S6}.
Technical Lemmata used in this proof are given in Appendix~\ref{app-Thm21prf-1}.

We compute
\begin{equation}\label{eq:fin-9}
\begin{split}
  \la H \widetilde\Phi^{\mathrm{trial}},
  \widetilde\Phi^{\mathrm{trial}}\ra
  = \la H \Sgsa  \Psz, \Sgsa  \Psz\ra
  + 2\Re \la H \Sgsa  \Psz, \Psp\ra
  + \la H\Psp ,\Psp \ra ,
\end{split}
\end{equation}
and we recall
\begin{equation}\label{decomp-H}
   H = h_\alpha + (H_f+ P_f^2)
   - 2\Re P. P_f - 2\alpha^\frac12 P. A(0)
   + 2 \alpha^\frac12 P_f. A(0) + 2 \alpha \Ac . \Aa + 2\alpha
   (\Aa)^2 .
\end{equation}

 $\bullet$ For the first term in
\eqref{eq:fin-9}, a straightforward computation shows
\begin{equation}\label{eq:fin-9-1}
\begin{split}
 \la H \Sgsa  \Psz, \Sgsa  \Psz\ra
 = (\Sigma_0 - e_0) \| \Sgsa  \|^2 \|\Psz \|^2  .
\end{split}
\end{equation}

$\bullet$ We estimate the second term on the right hand side of
\eqref{eq:fin-9} by computing each term that occurs in the
decomposition \eqref{decomp-H}.

$\diamond$ Using the symmetry of $\Sgsa $, the only non zero terms
in $2\Re \la H \Sgsa  \Psz, \Psp\ra$ are given by
\begin{equation}\label{eq:u-1}
 2\Re \la H \Sgsa  \Psz, \Psp\ra =
 -4\Re \la P.P_f \Sgsa \Psz, \Psp\ra
 - 4\Re \la P. \Ac \Sgsa  \Psz, \Psp\ra
 - 4\Re \la P. \Aa \Sgsa  \Psz, \Psp\ra .
\end{equation}

$\diamond$ The first term on the right hand side of \eqref{eq:u-1}
is estimated with similar arguments as in
Lemma~\ref{lem:lem-cross-1}, and using
$\|\Rvar\|^2=\mathcal{O}(\alpha^3)$ (Lemma~\ref{appendix:lem-A10})
and $\|\Rvar\|_*^2=\mathcal{O}(\alpha^4)$
(\cite[Theorem~3.2]{BCVVi}). We obtain
\begin{equation}\label{eq:u-2}
 -4\Re \la P.P_f \Sgsa \Psz, \Psp\ra = -\frac{2}{3}\alpha^4
 \sum_{i=1}^3 \la P_f^i \cEpp , (H_f + P_f^2) W_i\ra
 + \mathcal{O}(\alpha^5 \sqrt{\log\alpha^{-1}}) .
\end{equation}

$\diamond$ The second and third terms on the right hand side of
\eqref{eq:u-1} are estimated as in Lemma~\ref{lem:lem-cross-2},
and using again $\|\Rvar\|^2=\mathcal{O}(\alpha^2)$ and
$\|\Rvar\|_*^2=\mathcal{O}(\alpha^4)$. This yields
\begin{equation}\label{eq:u-3}
\begin{split}
 -4\Re \la P.\Ac \Sgsa \Psz, \Psp\ra =
 -2\| \fFz \|_{\sharp}^2 +
 o(\alpha^5\log\alpha^{-1}) ,
\end{split}
\end{equation}
and
\begin{equation}\label{eq:u-4}
\begin{split}
 & -4\Re \la P.\Aa \Sgsa \Psz, \Psp\ra \\
 & = -\frac23 \alpha^4 \sum_{i=1}^3 \la (\Aa)^i\cEpp ,
 (H_f+P_f^2)^{-1} (\Ac)^i \vac\ra +
 \mathcal{O}(\alpha^5 \sqrt{\log\alpha^{-1}}) .
\end{split}
\end{equation}

$\bullet$ Next, we estimate the third term on the right hand side
of \eqref{eq:fin-9}. For that sake, we also use the decomposition
\eqref{decomp-H} for $H$.

$\diamond$ For the term involving $h_\alpha$, using $\|(h_\alpha +
e_0)\Proj^0 \Psp\| = \mathcal{O}(\alpha^3)$ (since $\|P. \Aa \fF\|
= \mathcal{O}(\alpha^\frac52))$, and $\|(h_\alpha + e_0)
\frac{\partial \Sgsa }{\partial x_i}\| = \mathcal{O}(\alpha^3)$,
we directly obtain
\begin{equation}\label{eq:u-5}
\begin{split}
  & \la h_\alpha \Psp,\Psp \ra = \la (h_\alpha+e_0) \Psp,\Psp\ra -e_0 \| \Psp\|^2 \\
  & = 4\alpha \|(h_\alpha +e_0)^{-\frac12} Q_\alpha^\perp P. \Aa  \fF\|^2
  + \la (h_\alpha + e_0) \fFz , \fFz \ra -e_0 \|\Psp\| ^2
  + \mathcal{O}(\alpha^5) .
\end{split}
\end{equation}

$\diamond$ For the term with $H_f + P_f^2$, we use the estimate
\eqref{appendix:lem9-1} of Lemma~\ref{appendix:lem-A4}, and the
$\la \, \cdot \, , \, \cdot \, \ra_*$-orthogonality (see
\eqref{eq:0-with-Pf} of Lemma~\ref{lem:appendix-3}) of the two
vectors $\alpha\cEpp \Proj^0 g$ and $\sum_{i=1}^3 2 \alpha
(H_f+P_f^2)^{-1}W_i \frac{\partial \Sgsa }{\partial x_i}$ that
occur in $\Proj^2 \Psp$. We therefore obtain
\begin{equation}\label{eq:u-6}
\begin{split}
 & \la (H_f + P_f^2) \Psp , \Psp\ra = \la (H_f+P_f^2) \fFz , \fFz \ra
 + \alpha^2 \|\cEpp  \Proj^0 \Psp\|_*^2 \\
 & + \|\sum_{i=1}^3 2 \alpha
 (H_f+P_f^2)^{-1}W_i \frac{\partial \Sgsa }{\partial x_i}\|_*^2
 + \|\cEpp \|_*^2 \|\fFz \|^2
 + o(\alpha^5 \log\alpha^{-1}) .
\end{split}
\end{equation}

$\diamond$ Using the symmetry of $\Sgsa $, all terms in $\Re \la
P.P_f \Psp , \Psp\ra$ are zero, except the expression $\Re \la
P.P_f \Proj^2 \Psp, \Proj^2 \Psp\ra$, which is estimated as
follows
\begin{equation}\nonumber
\begin{split}
 &  \Re \la P.P_f \Proj^2\Psp,\Proj^2 \Psp\ra
 = 2\Re \la P. P_f \alpha\cEpp \Proj^0\Psp ,
 \sum_{i=1}^3 2 \alpha (H_f+P_f^2)^{-1}
 W_i \frac{\partial \Sgsa }{\partial x_i}\ra
 = \mathcal{O}(\alpha^5) ,
\end{split}
\end{equation}
where we used Lemma~\ref{lem:appendix-3} in the first equality to
prove that only the crossed term remains. Therefore, we obtain
\begin{equation}\label{eq:u-7}
  -2\Re \la P.P_f \Psp , \Psp\ra = \mathcal{O}(\alpha^5) .
\end{equation}

$\diamond$ The terms involving $-2\alpha^\frac12 \Re \la P. A(0)
\Psp , \Psp \ra$ is estimated as in the proof of
Lemma~\ref{lem:g1g1-1}. This yields
\begin{equation}\label{eq:u-8}
\begin{split}
 & -2\alpha^\frac12\la P. A(0) \Psp , \Psp\ra = -4\Re \alpha^\frac12
 \la P. \Ac\Psp,\Psp \ra \\
 &  = -4\Re \la P.\Ac \Proj^0\Psp,\Proj^1 \Psp\ra +
 \mathcal{O}(\alpha^5\sqrt{\log\alpha^{-1}}) \\
 & = -8\alpha \| (h_\alpha +e_0)^{-\frac12} Q_\alpha^\perp P. \Aa \fF\|^2
 + \mathcal{O}(\alpha^5\sqrt{\log\alpha^{-1}}) .
\end{split}
\end{equation}

$\diamond$ For $2\alpha^\frac12 \Re \la P_f.A(0)\Psp,\Psp \ra$, we
proceed as in the proof of Lemma~\ref{lem:g1g1-2}, and obtain
\begin{equation}\label{eq:u-9}
\begin{split}
 & 2\alpha^\frac12\la P_f. A(0) \Psp , \Psp\ra = 4\Re \alpha^\frac12
 \la P_f. \Aa\Psp,\Psp \ra \\
 & = 4\alpha^\frac12 \la P_f.\Aa \alpha \sum_{i=1}^3 2(H_f +
 P_f^2)^{-1} W_i \frac{\partial \Sgsa }{\partial x_i}, \fF\ra
 + \mathcal{O}(\alpha^5\sqrt{\log\alpha^{-1}}) .
\end{split}
\end{equation}

$\diamond$ Using the symmetry of $\Sgsa $ and $\Proj^0 \Psp$, the
term $2\alpha \Re \la \Aa . \Aa \Psp , \Psp\ra$ is estimated as
follows,
\begin{equation}\label{eq:u-10}
\begin{split}
 & 2\alpha \Re \la \Aa .\Aa \Psp , \Psp\ra \\
 & = 2\alpha \Re \la \Aa. \Aa
 (\alpha \cEpp  \Proj^0\Psp + \sum_{i=1}^3 2\alpha (H_f+P_f^2)^{-1}
 W_i \frac{\partial \Sgsa }{\partial x_i}), \Proj^0 \Psp\ra \\
 & + 2\alpha \Re\la \Aa . \Aa (-\alpha (H_f+P_f^2)^{-1}\Ac .\Ac
 \fFz ), \fFz \ra \\
 & = -2\alpha^2 \|\Proj^0 \Psp\|^2 \|\cEpp \|_*^2
 -2 \alpha^2 \|(H_f+P_f^2)^{-1} \Ac. \Ac \fFz  \|^2 \\
 & = -2\alpha^2 \|\Proj^0 \Psp\|^2 \|\cEpp \|_*^2
 - 2\alpha^2 \|\cEpp \|_*^2 \|\fFz \|^2 + o(\alpha^5\log\alpha^{-1}) ,
\end{split}
\end{equation}
where in the last inequality we used \eqref{appendix:lem9-1} of
Lemma~\ref{appendix:lem-A4}.

$\diamond$ Finally, a straightforward computation yields
\begin{equation}\label{eq:u-11}
 2\alpha \Re \la \Aa. \Ac\Psp , \Psp\ra = 2\alpha \| \Aa \fFz \|^2 +
 \mathcal{O}(\alpha^5) = 2\alpha \|\Aa \fF\|^2 +
 \mathcal{O}(\alpha^5) ,
\end{equation}
where in the last equality, we used Lemma~\ref{appendix:lem-A1}.

$\bullet$ Before collecting \eqref{eq:u-1}-\eqref{eq:u-11}, we
show that gathering some terms yield simpler expressions. Namely,
we have
\begin{equation}\label{eq:u-12}
\begin{split}
  & -\frac{2}{3}\alpha^4 \sum_{i=1}^3
  \la P_f^i \cEpp , (H_f + P_f^2) W_i\ra
  + 4\alpha^\frac12 \la P_f.\Aa \alpha \sum_{i=1}^3 2(H_f +
  P_f^2)^{-1} W_i \frac{\partial \Sgsa }{\partial x_i}, \fF\ra
  \\
  & + \|\sum_{i=1}^3 2 \alpha
 (H_f+P_f^2)^{-1}W_i \frac{\partial \Sgsa }{\partial x_i}\|_*^2
  = -\frac13 \alpha^4 \sum_{i=1}^3 \|(H_f+P_f^2)^{-1}W_i \|_*^2 .
\end{split}
\end{equation}
We also have, using $-\alpha^2\|\cEpp \|_*^2 = \Sigma_0
+\mathcal{O}(\alpha^3)$ (see e.g. \cite{BCVVi})
\begin{equation}\label{eq:u-13}
\begin{split}
 & (\Sigma_0 - e_0)\|\Psz \|^2 - e_0\|\Psp\| ^2 -\alpha^2 \|\cEpp \|_*^2
 \| \Proj^0 \Psp\|^2 - \alpha^2 \|\cEpp \|_*^2 \|\fFz \|^2 \\
 & = (\Sigma_0 - e_0) (\|\Psz\|^2 + \|\Psp\| ^2) +
 \mathcal{O}(\alpha^5).
\end{split}
\end{equation}
Therefore, collecting \eqref{eq:u-1}-\eqref{eq:u-11}, and using
the two equalities \eqref{eq:u-12}-\eqref{eq:u-13}, we obtain
\begin{equation}\nonumber
\begin{split}
  & \la H(\Sgsa \Psz + \Psp), \Sgsa \Psz +\Psp \ra
  = (\Sigma_0 -e_0) (\|\Psz\|^2 + \|\Psp\| ^2)
  -\frac{1}{3}\alpha^4
 \sum_{i=1}^3 \|(H_f + P_f^2)^{-1} W_i\|_*^2 \\
  & -\frac23 \alpha^4 \sum_{i=1}^3 \la (\Aa)^i\cEpp ,
 (H_f+P_f^2)^{-1} (\Ac)^i \vac\ra -\|\fFz \|_{\sharp}^2
  -4\alpha \| (h_\alpha + e_0)^{-\frac12} Q_\alpha^\perp P.
 \Aa\fF\|^2 \\
 & + 2\alpha \|\Aa \fF\|^2 + o(\alpha^5\log\alpha^{-1}) .
\end{split}
\end{equation}

With the definition $\dvar^{(1)}$, $\dvar^{(2)}$, and
$\dvar^{(3)}$, of Theorem~\ref{thm:main} this expression can be
rewritten as
\begin{equation}\label{eq:u-15}
\begin{split}
  \la (H-\Sigma_0 + e_0) \widetilde\Phi^{\mathrm{trial}},\
  \widetilde\Phi^{\mathrm{trial}}\ra =
  \dvar^{(1)}\alpha^3 + \dvar^{(2)}\alpha^4 + \dvar^{(3)}\alpha^5\log\alpha^{-1} +
  o(\alpha^5\log\alpha^{-1}) .
\end{split}
\end{equation}
Using Lemma~\ref{appendix:lem-A10} yields $\|\Psz\|^2
=1+\mathcal{O}(\alpha^2)$, which implies, due to the orthogonality
of $\Psp$  and $\Sgsa $ in $L^2(\R^3, \d x)$,
\begin{equation}\nonumber
\begin{split}
 \| \widetilde\Phi^{\mathrm{trial}}\|^2 = \|\Sgsa \|^2 \|\Psz\|^2 +
 \|\Psp\| ^2 = 1 + \mathcal{O}(\alpha^2).
\end{split}
\end{equation}
Therefore, together with \eqref{eq:u-15}, this gives
\begin{equation}\nonumber
\begin{split}
  \frac{\la (H-\Sigma_0 + e_0) \widetilde\Phi^{\mathrm{trial}},\
  \widetilde\Phi^{\mathrm{trial}}\ra}{\|\widetilde\Phi^{\mathrm{trial}}\|^2 } =
  \dvar^{(1)}\alpha^3 + \dvar^{(2)}\alpha^4 + \dvar^{(3)}\alpha^5\log\alpha^{-1} +
  o(\alpha^5\log\alpha^{-1}) .
\end{split}
\end{equation}
which concludes the proof of the lower bound in
Theorem~\ref{thm:main}.
\qed

\newpage

\begin{appendix}


\section{Proof of Proposition~\ref{prop:prf-Hgpsi}}\label{subsection:6-1}

In this Appendix, we provide   proofs of results that have a
high level of technicality.
To begin with, we establish Proposition~\ref{prop:prf-Hgpsi}.

\begin{lemma}\label{lem:lem-PPf}
The following holds
\begin{equation}\label{eq:lem-PPf}
\begin{split}
 & - 4\Re\la P.P_f \Sgsa  \cU , \gsp\ra \geq  -\frac43 \alpha^2 \Re
 \sum_{i=1}^3 \overline{\muvar_{2,i}}\, \|\nabla \Sgsa \|^2\,
 \la (H_f + P_f^2)^{-1} P_f^i \cEpp , W_i\ra \\
 &  - 4\alpha \Re \la \nabla \Sgsa .   P_f \cEpp , \Proj^2
 \gspl \ra -\epsilon  \|H_f^{\frac12} \gspl \|^2
 - \epsilon \alpha^5 |\muvar_3|^2 + \mathcal{O}(\alpha^5)
\end{split}
\end{equation}
\end{lemma}
%
%
%
\begin{proof}
For $n\neq 2,3$, with the estimates from the proof of
Lemma~\ref{lem:lem-cross-1} and using that due to
Theorem~\ref{thm:thm-main2} we have
 $$
  \|H_f^\frac12\Rvar\|^2 =\mathcal{O}(\alpha^4),\
  |\eta_1|=\mathcal{O}(1), \mbox{ and } |\muvar_1|=\mathcal{O}(1),
 $$
and since $\Proj^1\rvar= \Proj^1 \gspl $ and $\Proj^{n\geq4}\rvar
=\Proj^{n\geq4} \gspl $, we obtain
\begin{equation}\label{eq:PPf-0}
 \sum_{n\neq 2,3} -4\Re \la P.P_f \Sgsa  \Proj^n \cU , \Proj^n \gsp\ra
 \geq -\epsilon \| H_f^{\frac12} \gspl \|^2 + \mathcal{O}(\alpha^5) .
\end{equation}

For $n=2$,
\begin{equation}\label{eq:PPf-1}
\begin{split}
 & -4\Re \la \nabla \Sgsa  . (\alpha \eta_2 P_f\cEpp  + P_f \Proj^2
\Rvar), \Proj^2 \gsp\ra \geq -4\Re \la \nabla \Sgsa  . \alpha
\eta_2
 P_f \cEpp , \Proj^2 \gsps \ra \\
 & - 4\Re \la \nabla \Sgsa . \alpha \eta_2 P_f \cEpp , \Proj^2
 \gspl \ra - c\alpha \|H_f^{\frac12} \Proj^2\Rvar\|^2 - c\alpha
 \|H_f^{\frac12} \Proj^2 \gsp\|^2 .
\end{split}
\end{equation}
Using Theorem~\ref{thm:thm-main2}, the last two terms on the right
hand side of \eqref{eq:PPf-1} can be estimated by
$\mathcal{O}(\alpha^5)$. For the first term on the right hand side
of \eqref{eq:PPf-1}, using from Lemma~\ref{lem:appendix-3} that
$\la P_f^i\cEpp , \cEpp \ra=0$, from Theorem~\ref{thm:thm-main2}
that $\eta_2 = 1+\mathcal{O}(\alpha)$, and from
Lemma~\ref{lem:improved-estimates} that $\muvar_{2,i} =
\mathcal{O}(1)$, holds
\begin{equation}
\begin{split}
 & - 4 \alpha \Re \la \nabla \Sgsa  . \eta_2 P_f \cEpp , \Proj^2
 \gsps \ra \\
 & = - 4 \Re \alpha \la \nabla \Sgsa  . \eta_2 P_f \cEpp , \alpha
 \muvar_2 \cEpp  \Proj^0 \gsp\ra \\
  & - 4 \Re \alpha \la \nabla \Sgsa  . \eta_2 P_f \cEpp , \sum_i
  \alpha \muvar_{2,i} (H_f+P_f^2)^{-1} W_i \frac{\partial
  \Sgsa }{\partial x_i}\ra \\
  & = - \frac43 \Re \alpha^2 \|\nabla \Sgsa \|^2
  \sum_i \overline{\muvar_{2,i}} \la (H_f+P_f^2)^{-1} P_f^i \cEpp ,
  W_i\ra \\
  & = - \frac13 \alpha^4\Re
  \sum_{i=1}^3 \overline{\muvar_{2,i}} \la (H_f+P_f^2)^{-1} P_f^i \cEpp ,
  W_i\ra.
\end{split}
\end{equation}
We also used that $\la \frac{\partial \Sgsa }{\partial x_i},
\frac{\partial \Sgsa }{\partial x_j} \ra =0$ for $i\neq j$ and
$\|\frac{\partial \Sgsa }{\partial x_i}\| = \|\frac{\partial \Sgsa
}{\partial x_j}\|$ for all $i$ and $j$.

Finally, the second term on the right hand side of
\eqref{eq:PPf-1} gives the second term on the right hand side of
\eqref{eq:lem-PPf} plus $\mathcal{O}(\alpha^5)$, using from
Theorem~\ref{thm:thm-main2} that $|\eta_2 - 1|^2 =
\mathcal{O}(\alpha^2)$ and $\| H_f^\frac12 \Proj^2 \gspl  \| =
\mathcal{O}(\alpha^2)$.

To complete the proof, we shall estimate now the term for $n=3$,
\begin{equation}\label{eq:PPf-2}
\begin{split}
 & 4\Re \la P.P_f \Sgsa  \Proj^3\cU , \Proj^3 \gsp\ra =
 4 \Re \alpha^{\frac32} 2\eta_3 \la P.P_f \Sgsa  \cEp , \Proj^3
 \gsps \ra\\
 & + 4\Re\alpha^{\frac32} 2\eta_3 \la P.P_f\Sgsa  \cEp ,
 \Proj^3 \gspl \ra
 + 4\Re \la P.P_f \Sgsa  \Proj^3\Rvar,\Proj^3 \gsp\ra .
\end{split}
\end{equation}
The inequalities $\| H_f^{\frac12}\Rvar\| \leq c\alpha^2$ and
$\|H_f^{\frac12} \Proj^3 \gsp\| \leq c\alpha^2$ (see
Theorem~\ref{thm:thm-main2}) imply that the last term on the right
hand side of \eqref{eq:PPf-2} is $\mathcal{O}(\alpha^5)$. For the
second term on the right hand side of \eqref{eq:PPf-2} holds
\begin{equation}
 \Re \alpha^{\frac32} \eta_3 \la P.P_f \Sgsa  \cEp , \Proj^3
 \gspl \ra
 \geq -\epsilon \| H_f^{\frac12}\Proj^3 \gspl \|^2 +
 \mathcal{O}(\alpha^5) ,
\end{equation}
since from Theorem~\ref{thm:thm-main2} we have $\eta_3 =
\mathcal{O}(1)$.

Finally to estimate the first term on the right hand side of
\eqref{eq:PPf-2}, we note that
 $$
  |k_1|^{- \frac16}|k_2|^{- \frac16}|k_3|^{- \frac16}
  \cEp  (k_1, k_2, k_3) \in L^2(\R^9, \C^6)\ ,
 $$
and from Lemma~\ref{appendix:lem-A4},
 $$
  \|\, |k_1|^{\frac16}|k_2|^{\frac16}|k_3|^{\frac16}
  (H_f + P_f^2)^{-1}
  \Ac .\Ac \fFz \|^2 = \mathcal{O}(\alpha^3) .
 $$
This implies, using again $|\eta_3|=\mathcal{O}(1)$, and the
explicit expression of $\Proj^3 \gsps $
\begin{equation}\label{eq:PPf-last}
  |\alpha^{\frac32} 2\eta_3 \la P.P_f \Sgsa  \cEp , \Proj^3
  \gsps \ra| \leq c\alpha^{\frac52} |\muvar_3| \alpha^{\frac32} |\eta_3|
  \, \|P\Sgsa \| \leq \epsilon \alpha^5 |\muvar_3|^2 +
  \mathcal{O}(\alpha^5) .
\end{equation}
Collecting \eqref{eq:PPf-0}-\eqref{eq:PPf-last} concludes the
proof.
\end{proof}

\begin{lemma}\label{lem:lem-A(0)}
The following estimate holds
\begin{equation}\label{eq:lem-A(0)}
\begin{split}
 & -4\sqrt {\alpha} \Re \la P. \Ac \Sgsa  \cU , \gsp\ra\\
 &  =
 - 2\Re \overline{\muvar_1} \Proj^0\cU  \|\fFz \|_{\sharp}^2
 -\frac{1}{4} M[\gspl ]
 -\alpha^5|\muvar_3|^2
 +\mathcal{O}(\alpha^5) ,
\end{split}
\end{equation}
\end{lemma}
%
%
%
\begin{proof}
Obviously
\begin{equation}\label{eq:PA+-1}
\begin{split}
 & -4 \Re\alpha^{\frac12} \la P. \Ac \Sgsa  \cU , \gsp\ra
 \\
 & =
 -4 \Re \alpha^{\frac12} \la P.\Ac \Sgsa  (\Proj^0 \cU  + 2\eta_1
 \alpha^{\frac32}\cEm  +\eta_2\alpha\cEpp  +
 2\eta_3\alpha^{\frac32} \cEp  +\Rvar), \gsp\ra .
\end{split}
\end{equation}
%
%
\noindent$\diamond$ \textit{Step 1} From \eqref{eq:PA+-1}, let us
first estimate the term
\begin{equation}\label{eq:PA+-1bis}
 -4 \Re \alpha^{\frac12}\la P.\Ac \Sgsa \Rvar,\gsp\ra =
 -4 \alpha^{\frac12} \Re \sum_{n=0}^\infty \la P.\Ac
 \Proj^n \Sgsa \Rvar,
 \Proj^{n+1} \gsp\ra.
\end{equation}

For $n=0$, the corresponding term vanishes since $\Proj^0\Rvar=0$.

For $n> 2$, we can use \eqref{eq:alpha3-PA+-2} where the term
$\mathcal{O}(\alpha^5\log\alpha^{-1})$ can be replaced with
$\mathcal{O}(\alpha^5)$ because we know from
Theorem~\ref{thm:thm-main2} that $\|\Rvar\|^2 =
\mathcal{O}(\alpha^{\frac{33}{16}})$.

For $n=1$, we have
\begin{equation}\label{eq:PA+-2}
\begin{split}
 & |4\alpha^{\frac12} \Re \la P.\Ac
 \Proj^1 \Sgsa \Rvar,\Proj^2 \gsps  + \Proj^2 \gspl \ra |
 \leq c\alpha^3 \|\Proj^1\Rvar\|^2 + \epsilon \|H_f^{\frac12} \Proj^2
 \gspl \|^2 \\
 & + | 4\alpha^{\frac12} \la \nabla \Sgsa  \Proj^1\Rvar,\Aa \left(
 \alpha\muvar_2 \cEpp  \Proj^0\gsps  + \alpha \sum_{i=1}^3 \muvar_{2,i}
 (H_f + P_f^2)^{-1} W_i \frac{\partial \Sgsa }{\partial
 x_i}\right)\ra | .
\end{split}
\end{equation}
To estimate the last term on the right hand side we note that
$H_f^{-\frac12} \Aa\cEpp \in L^2$ and $H_f^{-\frac12} \Aa
(H_f+P_f^2)^{-1} W_i \in L^2$ which thus gives for this term the
bound
\begin{equation}\label{eq:PA+-3}
 c\alpha \|H_f^{\frac12}\Rvar\|^2 +\epsilon \alpha^4 |\muvar_2|^2
 \|\Proj^0 \gsps \|^2 + \epsilon \alpha^6 \sum_i |\muvar_{2,i}|^2
 = \mathcal{O}(\alpha^5) ,
\end{equation}
using Theorem~\ref{thm:thm-main2} and
Lemma~\ref{lem:improved-estimates}. The inequalities
\eqref{eq:PA+-2} and \eqref{eq:PA+-3} imply
\begin{equation}
 | \Re \alpha^{\frac12} \la P.\Ac \Proj^1 \Sgsa \Rvar,
 \Proj^2 \gsps  + \Proj^2 \gspl \ra| \leq
 \epsilon \| H_f^{\frac12} \Proj^2 \gspl \|^2 +\mathcal{O}(\alpha^5) .
\end{equation}
To complete the estimate of the term $4\alpha^{\frac12} \Re\la
P.\Ac \Sgsa \Rvar,\gsp\ra$ we have to estimate the term for $n=2$
in \eqref{eq:PA+-1bis}, namely $4\Re\alpha^\frac12\la P.\Ac \Sgsa
\Proj^2\Rvar,\Proj^3\gsps  +\Proj^3 \gspl \ra$. Obviously,
\begin{equation}
 |\Re\alpha^{\frac12} \la P. \Ac \Sgsa  \Proj^2\Rvar,\Proj^3 \gspl \ra|
 \leq \epsilon \|H_f^\frac12 \gspl \|^2 +\mathcal{O}(\alpha^5) .
\end{equation}
For the term involving $\Proj^3 \gsps $ we have
\begin{equation}\label{eq:PA+-4}
\begin{split}
 & |\Re\alpha^\frac12 \la P \Sgsa  \Proj^2\Rvar,\alpha\muvar_3
 \Aa(H_f + P_f^2)^{-1} \Ac. \Ac \fFz \ra| \\
 & \leq c\alpha^{3-\frac{1}{16}} \|\Proj^2
\Rvar\|^2 + \epsilon |\muvar_3|^2 \alpha^{2+\frac{1}{16}}
 \|\fFz \|^2 = |\muvar_3|^2
 \alpha^{5} + \mathcal{O}(\alpha^5) ,
\end{split}
\end{equation}
using Theorem~\ref{thm:thm-main2} and Lemma~\ref{appendix:lem-A4}.
Collecting \eqref{eq:PA+-1bis}-\eqref{eq:PA+-4} yields
\begin{equation}\label{eq:PA+-5}
 |\Re\alpha^\frac12 \la P.\Ac \Sgsa \Rvar,g\ra|\leq
 \epsilon\|H_f^\frac12 \gspl \|^2
 + \alpha^5 |\muvar_3|^2 +\mathcal{O}(\alpha^5) .
\end{equation}
%
%
\noindent$\diamond$ \textit{Step 2} We next estimate in
\eqref{eq:PA+-1} the term $-4\Re \alpha^\frac12 \la P.\Ac \Sgsa
(\Proj^0\cU + \alpha^\frac32 2\eta_1 \cEm  +\alpha \eta_2 \cEpp  +
\alpha^\frac32 2\eta_3\cEp ), \gsp\ra$. First using
\eqref{eq:alpha3-PA+-2} yields
\begin{equation}\label{eq:PA+-5bis}
 -4\Re \alpha^\frac12 \la P.\Ac \Sgsa
 \Proj^0\cU , \gsp\ra = - 2 \Re (\overline{\muvar_1} \Proj^0\cU )
 \|\fFz \|_{\sharp}^2 .
\end{equation}
We also have, using Theorem~\ref{thm:thm-main2}
\begin{equation}\label{eq:PA+-6}
\begin{split}
 & |\alpha^\frac12 \la P \Sgsa  (\alpha^\frac32 2\eta_1\cEm  +
 \alpha^\frac32 2\eta_3 \cEp ) , \Aa \gsp\ra | \\
 & \leq \alpha \| H_f^\frac12 \Proj^2 \gsp\|^2 + \alpha \|H_f^\frac12
 \Proj^4 \gsp\|^2 +\mathcal{O}(\alpha^5) = \mathcal{O}(\alpha^5) ,
\end{split}
\end{equation}
and
\begin{equation}\label{eq:PA+-7}
\begin{split}
 & |\alpha^\frac12 \la P \Sgsa  \alpha\eta_2\cEpp
 , \Aa (\Proj^3 \gsps  +\Proj^3 \gspl )\ra | \\
 & \leq \epsilon \| H_f^\frac12 \Proj^3 \gspl \|^2
 + |\alpha^\frac32 \la P\Sgsa  \eta_2 |k_1|^{-\frac14}
 |k_2|^{-\frac14} \cEpp , |k_1|^{\frac14} |k_2|^{\frac14}
 \Aa \Proj^3 \gsps \ra |
 +\mathcal{O}(\alpha^5) \\
 & \leq \epsilon \| H_f^\frac12 \Proj^3 \gspl \|^2 +
 \epsilon |\muvar_3|^2\alpha^5 +\mathcal{O}(\alpha^5) .
\end{split}
\end{equation}
Here we used $|k_1|^{-\frac14} |k_2|^{-\frac14}\cEpp  \in L^2$ and
$\||k_1|^\frac14 |k_2|^\frac14 \Aa (H_f+P_f^2)^{-1}\Ac .\Ac \fFz
\|^2 = \mathcal{O}(\alpha^3)$ (see Lemma~\ref{appendix:lem-A4}).

Collecting \eqref{eq:PA+-5}-\eqref{eq:PA+-7} yields
\begin{equation}\label{eq:PA+-7bis}
 -4\Re\alpha^\frac12 \la P.\Ac \Sgsa  \cU , \gsp\ra \geq
 -2\Re \overline{\muvar_1} \Proj^0\cU  \|\fFz \|_{\sharp}^2
 - \epsilon \| H_f^\frac12 \gspl \|^2 - \epsilon \alpha^5 |\muvar_3|^2
 +\mathcal{O}(\alpha^5) .
\end{equation}
\end{proof}

\begin{lemma}\label{lem:lem-A(0)-bis}
\begin{equation}
\begin{split}
 & -4\sqrt{\alpha} \Re \la P. \Aa \Sgsa  \cU , \gsp\ra \geq \\
 & -\frac23 \alpha^4  \Re
 \sum_{i=1}^3 \la\, (H_f+P_f^2)^{-\frac12} (\Aa)^i \cEpp , (H_f+P_f^2)^{-\frac12}
 (\Ac)^i \vac\ra \\
 & -\frac14 M[\gspl ]
 -\epsilon \alpha^2 \| (\gsps )^a \|^2
 -\epsilon \alpha^5\log\alpha^{-1} (\alpha |\muvar_3|^2 +1)
 - |\muvar_1 -1| c\alpha^4 + \mathcal{O}(\alpha^5) \ ,
\end{split}
\end{equation}
where $(\gsps )^a(x) := (\Proj^0 \gsps (x) - \Proj^0 \gsps
(-x))/2$ is the odd part of $\Proj^0 \gsps $.
\end{lemma}
%
%
\begin{proof}
Since from Lemma~\ref{lem:appendix-3} we have $\nabla \Sgsa  .
 \Aa\cEm =0$, we have
\begin{equation}\label{eq:cross-lem2-6}
\begin{split}
  & 4\Re \alpha^\frac12 \la P.\Aa \Sgsa  \cU ,\gsp \ra =
  4\alpha^\frac32 \Re \eta_2 \la \Aa\cEpp . P \Sgsa , \Proj^1 \gsp\ra
  \\
  & + 4\alpha^2 \Re 2\eta_3 \la \Aa \cEp . P\Sgsa , \Proj^2 \gsp\ra
  + 4\alpha^\frac12 \la \Aa \Rvar.P\Sgsa , \gsp\ra .
\end{split}
\end{equation}
For the first term on the right hand side of
\eqref{eq:cross-lem2-6} we have
\begin{equation}\label{eq:cross-lem2-7}
\begin{split}
  & 4\alpha^\frac32\Re \eta_2 \la \Aa\cEpp  .P\Sgsa , \Proj^1 \gspl
  +\muvar_1 \fFz \ra \\
  & = 4\alpha^\frac32\Re\eta_2 \la H_f^{-\frac12} \Aa
  \cEpp .P\Sgsa , H_f^\frac12 \Proj^1\gspl \ra
  + 4\alpha^\frac32 \Re\eta_2 \la \Aa\cEpp  .P \Sgsa , \muvar_1
  \fFz \ra .
\end{split}
\end{equation}
The first term on the right hand side of \eqref{eq:cross-lem2-7}
is bounded from below by $-\epsilon \|H_f^\frac12 \Proj^1 \gspl
\|^2 + \mathcal{O}(\alpha^5)$.

Applying Lemma~\ref{appendix:lem-A1}, we can replace $\fFz $ in
the second term of the right hand side of \eqref{eq:cross-lem2-7}
by $\fF$, at the expense of $\mathcal{O}(\alpha^5)$. More
precisely
\begin{equation}\nonumber
\begin{split}
 & | \alpha^\frac32 \la \eta_2 \Aa \cEpp . P \Sgsa , \muvar_1
 (\fFz -\fF)\ra| \\
 & \leq c\alpha^5 |\eta_2|^2 \|(H_f+P_f^2)^{-\frac12} \Aa\cEpp \|^2
 + |\muvar_1|^2 \|\fFz  -\fF\|_*^2 = \mathcal{O}(\alpha^5) .
\end{split}
\end{equation}

Moreover
\begin{equation}\label{eq:cross-lem2-8}
\begin{split}
 & 4\alpha^\frac32 \Re\eta_2 \overline{\muvar_1} \la \Aa \cEpp . P
 \Sgsa  , \fF\ra \\
 & =\frac83 \alpha^2 \|\nabla \Sgsa \|^2 \Re
 \eta_2\overline{\muvar_1}\sum_{i=1}^3\la(\Aa)^i\cEpp ,
 (H_f+P_f^2)^{-1} (\Ac)^i \vac\ra \\
 & = \frac23 \alpha^4 \Re
 \eta_2\overline{\muvar_1}\sum_{i=1}^3\la(\Aa)^i\cEpp ,
 (H_f+P_f^2)^{-1} (\Ac)^i \vac\ra \\
 & \geq \frac23 \alpha^4 \Re
 \sum_{i=1}^3\la(\Aa)^i\cEpp ,
 (H_f+P_f^2)^{-1} (\Ac)^i \vac\ra - |\muvar_1-1|c\alpha^4,
\end{split}
\end{equation}
where we used $\muvar_1=\mathcal{O}(1)$
(Lemma~\ref{lem:improved-estimates}) and $\eta_2= 1 +
\mathcal{O}(\alpha)$ (Theorem~\ref{thm:thm-main2}). Note that the
right hand side of \eqref{eq:cross-lem2-8} is well defined since
$(H_f+P_f^2)^{-\frac12}\Ac \vac\in \gF$ and
$(H_f+P_f^2)^{-\frac12}\Aa \cEpp \in\gF$.

Collecting the estimates for the first and the second term in the
right hand side of \eqref{eq:cross-lem2-7}, we arrive at
\begin{equation}\label{eq:PA+-7ter}
\begin{split}
  & -4\alpha^\frac32\Re \eta_2 \la \Aa\cEpp  . P\Sgsa , \Proj^1
  \gsp\ra\\
  & \geq -\frac83 \alpha^2 \|\nabla \Sgsa \|^2 \Re \overline{\muvar_1}
  \la (\Aa)^i \cEpp , (H_f+P_f^2)^{-1} \Ac \vac\ra -\epsilon\|H_f^\frac12
  \Proj^1 \gspl \|^2 +\mathcal{O}(\alpha^5) .
\end{split}
\end{equation}
Here we used also $\eta_2 =1+\mathcal{O}(\alpha)$.

As the next step, we return to \eqref{eq:cross-lem2-6} and
estimate the second term on the right hand side as
\begin{equation}\label{eq:PA+-7quatro}
\begin{split}
 & 4\alpha^2\Re 2\eta_3 \la \Aa \cEp . P \Sgsa , \Proj^2 \gsp\ra
 = 8\alpha^2\Re \eta_3 \la H_f^{-\frac12} \Aa \cEp . P\Sgsa ,
 H_f^\frac12 \Proj^2 \gsp\ra =\mathcal{O}(\alpha^5) ,
\end{split}
\end{equation}
where we used $H_f^{-\frac12} \Aa\cEp  \in L^2$ and $\| \Proj^2
H_f^\frac12\rvar\|^2 = \| \Proj^2 H_f^\frac12 \gsp\|^2 =
\mathcal{O}(\alpha^4)$ from Theorem~\ref{thm:thm-main2}. For the
last term on the right hand side of \eqref{eq:cross-lem2-6}, we
have
\begin{equation}\label{eq:cross-lem2-10}
\begin{split}
  & 4\alpha^{\frac12} \Re \langle \Aa\Rvar \cdot \nabla \Sgsa  , \gsp\rangle \\
  & =
  4 \alpha^{\frac12} \Re \langle \Aa\Rvar \cdot \nabla \Sgsa  , \Proj^0
  \gsps \rangle + 4\alpha^{\frac12} \Re \overline{\muvar_1} \langle \Aa
 \Rvar\cdot \nabla \Sgsa ,  \fFz \rangle \\
  & + 4\alpha^{\frac12} \Re
  \langle \Aa\Rvar \cdot \nabla \Sgsa , \Proj^1 \gspl \rangle +
  4 \alpha^{\frac12} \Re \langle \Aa\Rvar \cdot \nabla \Sgsa  , \Proj^2
  \gsp\rangle  \\
  & + 4\alpha^{\frac12} \Re \langle \Aa\Rvar \cdot \nabla \Sgsa ,
  \Proj^3 \gsp\rangle
  + 4\alpha^\frac12\Re\la \Aa \Rvar. \nabla \Sgsa , \Proj^{n\geq 4} \gsp\ra.
\end{split}
\end{equation}
We write the function $\Proj^0 \gsps  =  (\gsps )^s + (\gsps )^a$
where $(\gsps )^s$ (respectively $(\gsps )^a$) denotes the even
(respectively odd) part of $\Proj^0 \gsps $. Obviously, we have
\begin{equation}\label{eq:cross-lem2-11}
 |\alpha^{\frac12} \Re \langle \Aa\Rvar\cdot \nabla \Sgsa , \Proj^0
 \gsps \rangle | \leq c\alpha \|H_f^{\frac12}R\|^2 + \epsilon
 \alpha^2 \| (\gsps )^a\|^2 = \epsilon
 \alpha^2 \| (\gsps )^a\|^2 + \mathcal{O}(\alpha^5).
\end{equation}
The constant $\epsilon$ can be chosen small for large $c$.

For the second term on the right hand side of
\eqref{eq:cross-lem2-10}, we have
\begin{equation}\label{eq:cross-lem2-12}
 |\alpha^{\frac12} \overline{\muvar_1} \langle \Aa\Rvar\cdot \nabla \Sgsa ,
 \fFz  \rangle | \leq \epsilon  \alpha^{5}\log\alpha^{-1}
 |\muvar_1|^2 + c\alpha
 \| H_f^{\frac12}\Rvar\|^2 = \epsilon\alpha^5\log\alpha^{-1}
 + \mathcal{O}(\alpha^{5}).
\end{equation}

For the third term on the right hand side of
\eqref{eq:cross-lem2-10}, we have, since $\delta
=\frac{3}{32}\alpha^2$
\begin{equation}\label{eq:cross-lem2-13}
 |4\alpha^{\frac12} \langle \Aa\Rvar\cdot \nabla \Sgsa , \Proj^1
 \gspl \rangle | \leq \frac{\delta}{8} \|\Proj^1 \gspl \|^2 + c\alpha
 \|H_f^{\frac12}\Rvar\|^2 = \frac{\delta}{8} \|\Proj^1 \gspl \|^2
 + \mathcal{O}(\alpha^{5}).
\end{equation}
Similarly
\begin{equation}\label{eq:cross-lem2-14}
 | 4\alpha^{\frac12} \langle \Aa\Rvar \cdot \nabla \Sgsa , \Proj^{n\geq 4}
 \gsp\rangle | \leq \frac{\delta}{8} \|\Proj^{n\geq 4} \gspl \|^2 +
  \mathcal{O}(\alpha^{5}).
\end{equation}

To complete the estimate of the last term in
\eqref{eq:cross-lem2-10}, we have to estimate two terms:
$-4\Re\alpha^\frac12 \la \Aa \Rvar. P\Sgsa , \Proj^2 \gsp\ra$ and
$-4\Re\alpha^\frac12 \la \Aa \Rvar. P\Sgsa , \Proj^3 \gsp\ra$. For
the first one we have
\begin{equation}\nonumber
\begin{split}
 & |\Re \alpha^\frac12 \la \Aa \Rvar. P\Sgsa , \Proj^2 \gsp\ra| \leq
 c\alpha
 \|H_f^\frac12\Rvar\|^2 + \epsilon\alpha^2 \|\Proj^2 \gspl \|^2
 +\epsilon \alpha^4|\muvar_2|^2 \|\Proj^0 \gsp\|^2\\
 &  + \epsilon \alpha^6
 \sum_{i=1}^3 |\muvar_{2,i}|^2 = \epsilon \alpha^2 \|\Proj^2 \gspl \|^2
 +\mathcal{O}(\alpha^5) .
\end{split}
\end{equation}
Similarly,
\begin{equation}\label{eq:PA+-8}
\begin{split}
 & |\Re \alpha^\frac12 \la \Aa \Rvar. P\Sgsa , \Proj^3 \gsp\ra| \leq
 c\alpha
 \|H_f^\frac12\Rvar\|^2 + \epsilon\alpha^2 \|\Proj^3 \gspl \|^2
 +\epsilon \alpha^6\log\alpha^{-1}|\muvar_3|^2
 +\mathcal{O}(\alpha^5) .
\end{split}
\end{equation}
Collecting the estimates \eqref{eq:cross-lem2-10}-\eqref{eq:PA+-8}
yields
\begin{equation}\label{eq:PA+-9}
\begin{split}
|4\Re \alpha^\frac12  \la \Aa\Rvar . P\Sgsa , \gsp\ra| \leq
\frac{\delta}{8} \|\gspl \|^2 + \epsilon \|H_f^\frac12 \gspl \|^2
+ \epsilon \alpha^5\log\alpha^{-1} (1 +\alpha |\muvar_3|^2)
+\mathcal{O}(\alpha^5) .
\end{split}
\end{equation}
Collecting \eqref{eq:PA+-5}, \eqref{eq:PA+-7bis},
\eqref{eq:PA+-7ter}, \eqref{eq:PA+-7quatro} and \eqref{eq:PA+-9}
concludes the proof.
\end{proof}

\subsection{Concluding the proof of Proposition~\ref{prop:prf-Hgpsi}}

We can now prove the estimate on $\Re\la H g,\, \Sgsa  \cU \ra$
asserted in
Proposition~\ref{prop:prf-Hgpsi}.

Using the
orthogonality \eqref{eq:orthogo} of $\Sgsa $ and $\gsp$ , yields
\begin{equation}\nonumber
2\Re \la H \gsp,\cU \ra = - 4\Re\la P.P_f \Sgsa  \cU , \gsp\ra
-4\sqrt{\alpha} \Re \la P. A(0) \Sgsa  \cU , \gsp\ra .
\end{equation}
Together with Lemmata~\ref{lem:lem-PPf}-\ref{lem:lem-A(0)-bis},
this concludes the proof of Proposition~\ref{prop:prf-Hgpsi}. \qed
%
%
%
%

\section{Proof of Proposition~\ref{prop:prop-hgg}}\label{prf-hgg}

In this section, we present the proof of Proposition~\ref{prop:prop-hgg}.

To begin with, we establish the estimate
\begin{proposition}\label{prop:sub1}
We have
\begin{equation}\label{eq:n-hgg-0}
\begin{split}
 & \la H\gsp , \gsp\ra
 \\ & \geq  \la H \gsps , \gsps \ra + \la H \gspl , \gspl \ra
 +2\alpha \left( \|\Aa \gsp\|^2 - \|\Aa \gsps \|^2 - \|\Aa
 \gspl \|^2\right) \\
 & - 4\Re \la P. P_f \gspl , \gsps \ra - 4\alpha^\frac12 \Re \la P. A(0)
 \gspl , \gsps \ra
 + 4\Re \la P_f. A(0) \gspl , \gsps \ra \\
 & - \epsilon M[\gspl ] - c\alpha^6\log\alpha^{-1} |\muvar_3|^2
 - c_0\alpha \| (h_\alpha+e_0)^\frac12 \Proj^0 \gsps \|^2
 +\mathcal{O}(\alpha^5) .
\end{split}
\end{equation}
\end{proposition}
%
%
%
\begin{proof}
Recall that
\begin{equation}\label{eq:n-hgg-1}
  H = P^2 - \frac{\alpha}{|x|} + T(0) -2\Re P.
  \left(P_f +\alpha^\frac12 A(0)\right) ,
\end{equation}
and
\begin{equation}\label{eq:n-hgg-2}
  T(0) = :(P_f +\alpha^\frac12 A(0))^2: + H_f .
\end{equation}
Due to the orthogonality
 $$
  \la \Proj^n \gsps , \Proj^n \gspl \ra_{\sharp} =0,\quad n=0,1,\ldots\ ,
 $$
and \eqref{eq:n-hgg-1}, \eqref{eq:n-hgg-2}, we obtain
\begin{equation}\label{eq:n-hgg-3}
\begin{split}
 & \la (H+e_0)\gsp,\gsp\ra \\
 & = \la (H+e_0) \gspl , \gspl \ra
 + \la (H+e_0) \gsps , \gsps \ra
 +\sum_{n=0}^3 2\alpha\Re \la \Aa. \Aa \Proj^{n+2} \gspl , \Proj^n
 \gsps \ra\\
 & + 2\alpha \Re \la \Aa. \Aa \Proj^3 \gsps , \Proj^1 \gspl \ra
 + 2\alpha (\|\Aa \gsp\|^2 - \|\Aa \gsps \|^2 -\|\Aa \gspl \|^2) \\
 & - 4\Re \la P.P_f \gspl , \gsps \ra - 4\alpha^\frac12 \Re \la
 P. A(0) \gspl , \gsps \ra + 4\alpha^\frac12 \Re \la P_f. A(0)
 \gspl , \gsps \ra .
\end{split}
\end{equation}
We have
\begin{equation}
\begin{split}
 & 2\alpha\Re \la \Aa .\Aa \Proj^5 \gspl , \Proj^3 \gsps \ra
 \geq -\epsilon \|H_f^\frac12 \Proj^5 \gspl \|^2 - c\alpha^2 \|\Proj^3 \gsps \|^2 \\
 & \geq -\epsilon \|H_f^\frac12 \Proj^5 \gspl \|^2
 - c\alpha^7\log\alpha^{-1} |\muvar_3|^2 .
\end{split}
\end{equation}
Similarly,
\begin{equation}\nonumber
\begin{split}
 & 2\alpha\Re \la \Aa .\Aa \Proj^4 \gspl , \Proj^2 \gsps \ra\\
 & \geq -\epsilon
 \| H_f^\frac12 \Proj^4 \gspl \|^2
 - c\alpha^4 |\muvar_2|^2 \|\Proj^0 \gsps \|^2 - \sum_{i=1}^3 c\alpha^6
 |\muvar_{2,i}|^2 \\
 & \geq -\epsilon \| H_f^\frac12 \Proj^4 \gspl \|^2
 +\mathcal{O}(\alpha^5) .
\end{split}
\end{equation}
To estimate the term $2\alpha\Re \la \Aa . \Aa \Proj^3\gspl  ,
\Proj^1 \gsps \ra$ we rewrite it as $2\alpha \Re \la \Proj^3 \gspl
, \Ac .\Ac \muvar_1\fFz \ra$ and use that $\la \Proj^3 \gspl ,
(H_f+P_f^2)^{-1} \Ac .\Ac \fFz \ra_{\sharp} =0$. This yields,
using Lemma~\ref{appendix:lem-A4}
\begin{equation}\label{eq:usedhere}
\begin{split}
 \alpha\Re \la \Aa . \Aa \Proj^3 \gspl , \Proj^1 \gsps \ra
 & = -\alpha \Re \la \Proj^3 \gspl  , (h_\alpha + e_0) (H_f+P_f^2)^{-1}
 \Ac .\Ac \muvar_1\fFz \ra \\
 & \geq -\epsilon\alpha^2 \|\Proj^3\gspl \|^2 +
 c\alpha^7\log\alpha^{-1} .
\end{split}
\end{equation}
Similarly, if $\|\Proj^0 \gsps \| > \alpha^\frac32$,
\begin{equation}
\begin{split}
 & 2\alpha\Re \la \Aa .\Aa \Proj^2 \gspl , \Proj^0 \gsps \ra
 = -2 \alpha\Re \la \Proj^2 \gspl , (h_\alpha + e_0)
 (H_f+P_f^2)^{-1} \Ac .\Ac \Proj^0 \gsps \ra \\
 & \geq - c\alpha \| (h_\alpha +e_0)^\frac12 \Proj^2 \gspl \|^2
 - c\alpha \| (h_\alpha + e_0)^\frac12 (H_f +P_f^2)^{-1} \Ac .\Ac
 \Proj^0 \gsps \|^2 \\
 & \geq -c\alpha \|P \Proj^2 \gspl \|^2 + c\alpha \|
 |x|^{-\frac12} \Proj^2 \gspl \|^2 - c\alpha e_0 \|\Proj^2
 \gspl \|^2
 - c_0 \alpha \| (h_\alpha +e_0)^\frac12 \Proj^0 \gsps \|^2 \\
 & \geq -c\alpha \| P\Proj^2 \gspl \|^2
 -\epsilon\alpha^2 \| \Proj^2 \gspl \|^2 - c_0\alpha
 \| (h_\alpha  + e_0)^\frac12 \Proj^0 \gsps \|^2\ .
\end{split}
\end{equation}
If $\|\Proj^0 \gsps \| \leq \alpha^\frac32$, we have instead
\begin{equation}
\begin{split}
 2\alpha\Re \la \Aa .\Aa \Proj^2 \gspl , \Proj^0 \gsps \ra
 & \geq
 -\epsilon \|H_f^\frac12 \Proj^2 \gspl \|^2
 -c\alpha^2\|\Proj^0\gsps \|^2 \\
 & \geq -\epsilon \|H_f^\frac12 \Proj^2 \gspl \|^2
 +\mathcal{O}(\alpha^5) .
\end{split}
\end{equation}

Finally, using Lemma~\ref{appendix:lem-A5} yields
\begin{equation}\label{eq:n-hgg-4}
 2\alpha \Re \la \Aa. \Aa \Proj^3 \gsps , \Proj^1 \gspl \ra \geq
 - \epsilon \|H_f^{\frac12} \gspl  \|^2 +\mathcal{O}(\alpha^5) .
\end{equation}

Collecting \eqref{eq:n-hgg-3}-\eqref{eq:n-hgg-4} concludes the
proof of the proposition.
\end{proof}

In the rest of this section, we estimate further terms in
\eqref{eq:n-hgg-0}.

\subsection{Estimate of crossed terms involving $\gsps $ and
$\gspl $}

\begin{lemma}\label{lem:sub-1}
\begin{equation}\label{eq:n-hgg-5}
  2\alpha ( \|\Aa \gsp\|^2 - \|\Aa \gsps \|^2 - \|\Aa \gspl \|^2 )
  \geq -\epsilon M[\gspl ] +\mathcal{O}(\alpha^5) .
\end{equation}
\end{lemma}
%
%
%
\begin{proof}
Obviously, the left hand side of \eqref{eq:n-hgg-5} is equal to
\begin{equation}
\begin{split}
 & 4\alpha\Re \la \Aa \gsps , \Aa \gspl \ra \geq
 -c\alpha \sum_n \| H_f^\frac12 \Proj^n \gsps \|\,
 \|H_f^\frac12 \Proj^n \gspl \| \\
 & \geq -
 \epsilon \| H_f^\frac12 \Proj^1 \gspl \|^2 - c\alpha^2 |\muvar_1|^2
 \|H_f^\frac12 \fFz \|^2
 -\sum_{n\neq 1} c\alpha \left(
 3 \| H_f^\frac12\Proj^n \gspl \|^2 + 2\|H_f^\frac12 \Proj^n \gsp\|^2
 \right)\\
 & \geq -\epsilon \|H_f^\frac12 \gspl \|^2 +\mathcal{O}(\alpha^5) ,
\end{split}
\end{equation}
where in the last inequality we used \eqref{eq:app-ii} of
Lemma~\ref{appendix:lem-A0}, and \eqref{eq:thm:thm-main2-1} of
Theorem~\ref{thm:thm-main2}.
\end{proof}

\begin{lemma}\label{lem:sub-2}
\begin{equation}\nonumber
 |\la P. P_f \gspl , \gsps \ra |
 \leq \epsilon M[\gspl ]
 + c\alpha^7 \log\alpha^{-1} |\muvar_3|^2 +\mathcal{O}(\alpha^5)
\end{equation}
\end{lemma}
%
%
%
\begin{proof}
We have
\begin{equation}\label{eq:n-hgg-6}
 \la P.P_f \gspl , \gsps \ra = \la P_f \Proj^1 \gspl , P \Proj^1 \gsps \ra + \la
 P_f \Proj^2 \gspl , P \Proj^2 \gsps \ra +\la P_f \Proj^3 \gspl , P\Proj^3 \gsps \ra .
\end{equation}
Obviously, using Lemma~\ref{appendix:lem-A0} and the equality
$\muvar_1 = \mathcal{O}(1)$ from
Lemma~\ref{lem:improved-estimates}, yields
\begin{equation}\label{eq:n-hgg-7}
 \la P_f \Proj^1 \gspl  , P\Proj^1 \gsps \ra| \leq \epsilon \|H_f
 ^\frac12 \gspl \|^2 + |\muvar_1|^2 \| P |P_f|^\frac12 \fFz \|^2 \leq
 \epsilon \|H_f
 ^\frac12 \gspl \|^2  + \mathcal{O}(\alpha^5).
\end{equation}
We also have, by definition of $\Proj^2 \gsps $ and using the
estimates $\muvar_{2,i} = \mathcal{O}(1)$ from
Lemma~\ref{lem:improved-estimates},
\begin{equation}\label{eq:n-hgg-7}
 \la P_f \Proj^2 \gspl  , P\Proj^2 \gsps \ra| \leq \epsilon \|H_f
 ^\frac12 \gspl \|^2
 + c|\muvar_2|^2\alpha^2 \|P \Proj^0 \gsps \|^2
 +\mathcal{O}(\alpha^6).
\end{equation}
We next bound the second term on the right hand side of
\eqref{eq:n-hgg-7}. Notice that by definition of $\muvar_2$, this
term is nonzero only if $\|\Proj^0 \gsps \|^2
>\alpha^3$, which implies, with Lemma~\ref{lem:improved-estimates},
that $\| P\Proj^0 \gsps \|^2 \leq c\alpha \|\Proj^0 \gsps \|^2$.
The inequality \eqref{eq:n-hgg-7} can thus be rewritten as
\begin{equation}\label{eq:n-hgg-8}
 |\la P. P_f \Proj^2 \gspl , \Proj^2 \gsps \ra| \leq
 \epsilon \| H_f^\frac12 \gspl \|^2 +
 c |\muvar_2|^2 \alpha^3 \|\Proj^0 \gsps \|^2 +\mathcal{O}(\alpha^6)
 \leq \epsilon \| H_f^\frac12 \gspl \|^2 + \mathcal{O}(\alpha^5) ,
\end{equation}
using in the last inequality that $|\muvar_2|\, \|\Proj^0 \gsps \|
=\mathcal{O}(\alpha)$ (see Lemma~\ref{lem:improved-estimates}).

For the second term on the right hand side of \eqref{eq:n-hgg-6},
using \eqref{appendix:lem7-1} from Lemma~\ref{appendix:lem-A0}
yields
\begin{equation}\label{eq:n-hgg-9}
  \la \Proj^3 P_f  \gspl , \Proj^3 P  \gsps \ra
  \leq \epsilon \| \Proj^3 H_f^\frac12  \gspl \|^2
  + c\alpha^7\log\alpha^{-1} |\muvar_3|^2 .
\end{equation}
The inequalities \eqref{eq:n-hgg-6}, \eqref{eq:n-hgg-8} and
\eqref{eq:n-hgg-9} prove the lemma.
\end{proof}

\begin{lemma}\label{lem:sub-3}
\begin{equation}\nonumber
  - 4\alpha^\frac12 \Re \la P. \Ac \gspl , \gsps \ra \geq
  -\epsilon M[\gspl ]
  +\mathcal{O}(\alpha^5) .
\end{equation}
\end{lemma}
%
%
%
\begin{proof}
Since $\Proj^{n > 3} \gsps =0$, $\Proj^0 \gspl =0$ and
 $$
 \| H_f^\frac12
 \Proj^{n\neq 1}\gsps \|^2 \leq 2 \|H_f^\frac12 \Proj^{n\neq1} \gsp\|^2 + 2
 \|H_f^\frac12 \Proj^{n\neq1} \gspl \|^2 \leq c M[\gspl ] +
 \mathcal{O}(\alpha^4) ,
 $$
(see \eqref{eq:thm:thm-main2-1} of Theorem~\ref{thm:thm-main2}) we
obtain
\begin{equation}\nonumber
\begin{split}
 4\alpha^\frac12 \Re\la P.\Ac \gspl , \gsps \ra & \leq
 \epsilon \|  \Proj^{n\leq 2} P \gspl \|^2 + c\alpha \| \Proj^{n\geq 2} H_f^\frac12
 \gsps \|^2
 \leq \epsilon M[\gspl ] +\mathcal{O}(\alpha^5) \ .
\end{split}
\end{equation}
\end{proof}

\begin{lemma}\label{lem:sub-4}
\begin{equation}\nonumber
 - 4\alpha^\frac12\Re \la P.\Aa \gspl , \gsps \ra \geq
 -\epsilon M[\gspl ]
 +\mathcal{O}(\alpha^5) .
\end{equation}
\end{lemma}
%
%
%
\begin{proof}
\begin{equation}\nonumber
\begin{split}
 & - 4\alpha^\frac12 \Re \la P. \Aa \gspl , \gsps \ra \geq
 -\epsilon \| H_f^\frac12 \gspl \|^2 - c\alpha \|P \gsps \|^2 \\
 & \geq -\epsilon \|H_f^\frac12 \gspl \|^2 - c\alpha \left( \|
 \Proj^{n=0,2,3,4 } P \gsp\|^2 + \| \Proj^{n=0,2,3,4} P \gspl \|^2\right)
 -c\alpha |\muvar_1|^2 \| P\fFz \|^2 \\
 & \geq -\epsilon \| H_f^\frac12 \gspl \|^2
 - \epsilon \| \Proj^{n\leq 4}P \gspl \|^2 +\mathcal{O}(\alpha^6\log\alpha^{-1})
 \geq -\epsilon M[\gspl ]
 +\mathcal{O}(\alpha^6\log\alpha^{-1}),
\end{split}
\end{equation}
using \eqref{appendix:lem7-1} of Lemma~\ref{appendix:lem-A0}.
\end{proof}

\begin{lemma}\label{lem:sub-5}
\begin{equation}\nonumber
  4\alpha^\frac12 \Re \la P_f. A(0) \gspl , \gsps \ra \geq
  4\alpha^\frac12 \Re \la  \Proj^2 \gspl , P_f. \Ac\fF\ra
  - \epsilon M[\gspl ] + \mathcal{O}(\alpha^5) .
\end{equation}
\end{lemma}
%
%
%
\begin{proof}
We have
\begin{equation}\nonumber
\begin{split}
 & 4\alpha^\frac12 \Re \la P_f.\Aa \gspl , \gsps \ra \\
 & \geq
 -\epsilon \|H_f^\frac12 \gspl \|^2 -c\alpha \|\Proj^{n\geq 2} P_f
 \gsps \|^2 + 4\alpha^\frac12 \Re \la P_f. \Aa \Proj^2 \gspl , \muvar_1\fFz \ra \\
 &  \geq -\epsilon \|H_f^\frac12 \gspl \|^2
 - c \alpha \| \Proj^{n=2,3} H_f^\frac12 \gsp\|^2
 + 4\alpha^\frac12 \Re \la P_f. \Aa \Proj^2 \gspl , \muvar_1\fFz \ra \\
 & \geq -\epsilon M[\gspl ] +
 4\alpha^\frac12 \Re \la P_f. \Aa \Proj^2 \gspl , \muvar_1\fFz \ra +
 \mathcal{O}(\alpha^5).
\end{split}
\end{equation}
We estimate the second term on the right hand side as follows,
\begin{equation}\nonumber
\begin{split}
 & 4\alpha^\frac12 \Re \la P_f. \Aa \Proj^2 \gspl , \muvar_1\fFz \ra \\
 & \geq
 4\alpha^\frac12 \Re \la P_f. \Aa \Proj^2 \gspl , \fFz \ra -\epsilon
 \|H_f^\frac12 \Proj^2 \gspl \|^2 - c\alpha |\muvar-1|^2 \| P_f \fFz \|^2
 \\
 & \geq 4\alpha^\frac12 \Re \la P_f. \Aa \Proj^2 \gspl , \fF\ra
 -\epsilon
 \|H_f^\frac12 \Proj^2 \gspl \|^2 + \mathcal{O}(\alpha^5) ,
\end{split}
\end{equation}
where we used $|\muvar_1-1|^2 =\mathcal{O}(\alpha)$ from
Lemma~\ref{lem:improved-estimates}, $\|\fFz \|_*^2 =
\mathcal{O}(\alpha^3)$ from Lemma~\ref{appendix:lem-A0}, and $\|
P_f(\fFz  -\fF)\|^2 = \mathcal{O}(\alpha^4)$ from
Lemma~\ref{appendix:lem-A1}.

We also have, using $P_f.\Ac = \Ac.P_f$,
\begin{equation}\nonumber
\begin{split}
 & 4\alpha^\frac12 \Re \la P_f . \Ac \gspl , \gsps \ra
 = 4\alpha^\frac12 \Re
 \la \Proj^{n\leq 2} P_f \gspl ,  \Aa \Proj^{n\geq 2}  \gsps \ra \\
 & \geq -\epsilon
 \|H_f^\frac12 \gspl \|^2
 - c\alpha \|  H_f^\frac12 \Proj^{n\geq 2} \gsps \|^2 \geq
 -\epsilon M[\gspl ] +\mathcal{O}(\alpha^5) ,
\end{split}
\end{equation}
where $\| H_f^\frac12 \Proj^{n\geq 2} \gsps \|$ has been estimated
as $\| P \gsps  \|$ in the proof of Lemma~\ref{lem:sub-4}.

\end{proof}

\subsection{Estimates of the term $\la (H+e_0) \gsps , \gsps \ra$}
$\ $

Due to \eqref{eq:n-hgg-1} and \eqref{eq:n-hgg-2}, one finds
\begin{equation}\label{eq:n-hgg-10}
\begin{split}
 & \la (H+e_0) \gsps , \gsps \ra = \la \gsps , \gsps \ra_{\sharp} - 2\Re \la
 P. (P_f +\alpha^\frac12 A(0)) \gsps , \gsps \ra \\
 & + 2\alpha^\frac12
 \Re\la P_f. A(0) \gsps , \gsps \ra
 + 2\alpha \|\Aa \gsps \|^2 + 2\alpha \Re\la \Aa . \Aa \gsps , \gsps \ra
\end{split}
\end{equation}

We estimate the terms in \eqref{eq:n-hgg-10} below.

\begin{lemma}\label{lem:g1g1-1}
We have
\begin{equation}\nonumber
\begin{split}
 & -2\Re \la P.(P_f +\alpha^\frac12 A(0)) \gsps , \gsps \ra \\
 & \geq
 -4\alpha^\frac12 \Re \la \Aa\fFz , P\Proj^0 (\gsps )^s\ra
 - \epsilon M[\gspl ] - |\muvar_1 -1| c\alpha^4 +\mathcal{O}(\alpha^5) ,
\end{split}
\end{equation}
where $\Proj^0 (\gsps )^s = \Proj^0 \gsps  - \Proj^0 (\gsps )^a$
is the even part of $\Proj^0 \gsps $.
\end{lemma}
%
%
%
\begin{proof} Using $\la\cEpp , P_f^i\cEpp \ra=0$ (see Lemma~\ref{lem:appendix-3}),
the symmetry of $\Sgsa $, and $\la P.P_f \fFz , \fFz \ra=0$, we
obtain
\begin{equation}\label{eq:n-hgg-11}
\begin{split}
 & | 2\la P.P_f \gsps , \gsps  \ra| = |2\la P_f \alpha \muvar_2
 \cEpp \Proj^0 \gsps , P\alpha \sum_{i=1}^3 \muvar_{2,i} (H_f +
 P_f^2)^{-1} W_i \frac{\partial \Sgsa }{\partial x_i}
 \ra | \\
 & \leq c\alpha^3 |\muvar_2|^2 \|\Proj^0 \gsps \|^2 + c\alpha^5
 \sum_{i=1}^3 |\muvar_{2,i}|^2 =\mathcal{O}(\alpha^5) ,
\end{split}
\end{equation}
where in the last inequality we used
Lemma~\ref{lem:improved-estimates}.

We also have, using again the symmetry of $\Sgsa$,
\begin{equation}\label{eq:n-hgg-12}
\begin{split}
 & - 2\alpha^\frac12 \Re \la P. A(0) \gsps , \gsps \ra
 = - 4\alpha^\frac12 \Re \la P.\Aa \gsps , \gsps \ra\\
 & = -4\alpha^\frac12 \Re \la \Aa \muvar_1 \fFz , P\Proj^0 (\gsps )^s\ra
 -4\alpha^\frac12 \Re \la \Aa \muvar_2 \alpha\cEpp  \Proj^0 \gsps , P\muvar_1
 \fFz \ra  \\
 & - 4\alpha^\frac12 \Re \la \Aa \Proj^3 \gsps , P\alpha \muvar_2 \cEpp
 \Proj^0 \gsps \ra
 \\
 & \geq -4\alpha^\frac12 \Re \la \Aa \muvar_1 \fFz , P\Proj^0 (\gsps )^s\ra
 - c\alpha \| H_f^\frac12 \Proj^3 \gsps \|^2 - c\alpha^2 \|P\Proj^0
 \gsps \|^2 |\muvar_2|^2 \\
 &  - 4\alpha^\frac12 \Re \la |k|^{-\frac16}
 \Aa\muvar_2\alpha\cEpp \Proj^0 \gsps , |k|^\frac16 \muvar_1 P\fFz  \ra \\
 & \geq -4\alpha^\frac12 \Re \la \Aa \muvar_1 \fFz , P\Proj^0 (\gsps )^s\ra
 - c\alpha \| H_f^\frac12 \Proj^3 \gsp\|^2
 - c\alpha \| H_f^\frac12 \Proj^3 \gspl \|^2 \\
 & -c\alpha^5
 - c \alpha^\frac32 |\muvar_2| \, \|\Proj^0 \gsps \| \, \|\, |k|^\frac16 P \fFz \|
  \\
 & \geq -4\alpha^\frac12 \Re \la \Aa \muvar_1 \fFz , P\Proj^0 (\gsps )^s\ra
 - \epsilon M[\gspl ] - c\alpha^5 ,
\end{split}
\end{equation}
where we used Theorem~\ref{thm:thm-main2} and
Lemma~\ref{appendix:lem-A0}.

Moreover, because $\|\Aa \fF\| = \mathcal{O}(\alpha^\frac32)$ and
$\|P\Proj^0 \gsps \|=\mathcal{O}(\alpha^2)$, we obtain
\begin{equation}\label{eq:estimate-diff-mu1}
  -4\Re\alpha^\frac12 \Re \la \Aa \muvar_1\fFz , P\Proj^0 (\gsps )^s\ra
  \geq - 4 \alpha^\frac12 \Re \la \Aa \fFz , P\Proj^0 (\gsps )^s\ra
  - |\muvar_1 -1| c\alpha^4 .
\end{equation}

This estimate, together with \eqref{eq:n-hgg-11} and
\eqref{eq:n-hgg-12}, proves the lemma.
\end{proof}

\begin{lemma}\label{lem:g1g1-2}
We have
\begin{equation}\nonumber
\begin{split}
 & 2 \alpha^\frac12 \Re \la P_f . A(0) \gsps , \gsps \ra \\
 & \geq
 \frac23 \alpha^4 \Re \sum_i
  \muvar_{2,i} \la P_f.\Aa
  (H_f+P_f^2)^{-1} W_i, (H_f+P_f^2)^{-1} (\Ac)^i\vac \ra\\
 & - |\muvar_1 -1| c\alpha^4 + \mathcal{O}(\alpha^5) .
\end{split}
\end{equation}
\end{lemma}
%
%
%
\begin{proof} The following holds
\begin{equation}\label{eq:n-hgg-13}
\begin{split}
 & 2\alpha^\frac12 \Re \la P_f. A(0) \gsps , \gsps \ra
 =4\alpha^\frac12 \Re\la P_f. \Aa \Proj^3 \gsps , \Proj^2 \gsps \ra
 + 4\alpha \Re \la P_f.\Aa \Proj^2 \gsps , \muvar_1 \fFz \ra \\
 & = 4\alpha^\frac32 \Re \sum_{i=1}^3 \overline{\muvar_{2,i}}
 \la |k_1|^\frac16 |k_2|^\frac16 \Aa\Proj^3 \gsps , P_f
 |k_1|^{-\frac16} |k_2|^{-\frac16} (H_f+P_f^2)^{-1} W_i
 \frac{\partial \Sgsa }{\partial x_i}\ra\\
 & + 4\alpha^\frac32 \Re \muvar_2 \la |k_1|^\frac16 |k_2|^\frac16
 \Aa \Proj^3 \gsps , |k_1|^{-\frac16} |k_2|^{-\frac16} P_f \cEpp  \Proj^0
 \gsps  \ra \\
 &
 + 4\alpha^\frac12 \Re \la P_f.\Aa \Proj^2 \gsps , \muvar_1 \fFz \ra .
\end{split}
\end{equation}

Applying the Schwarz inequality and the estimates $\|
|k_1|^\frac16 |k_2|^\frac16 \Aa\Proj^3 \gsps
\|^2=\mathcal{O}(\alpha^5)$ (Lemma~\ref{appendix:lem-A4}),
$\muvar_{2,i} =\mathcal{O}(1)$
(Lemma~\ref{lem:improved-estimates}), and $\|\nabla \Sgsa  \|^2
=\mathcal{O}(\alpha^2)$, we see that the first term on the right
hand side of \eqref{eq:n-hgg-13} is $\mathcal{O}(\alpha^5)$.
Applying also the estimate $|\muvar_2|\, \|\Proj^0 \gsps \| =
\mathcal{O}(\alpha)$ (Lemma~\ref{lem:improved-estimates}), we
obtain that the second term on the right hand side of
 \eqref{eq:n-hgg-13}  is also $\mathcal{O}(\alpha^5)$.

Finally, we estimate $4\alpha^\frac12 \Re \la P_f.\Aa \Proj^2
\gsps , \muvar_1 \fFz \ra$. The following inequality holds,
\begin{equation}\label{eq:n-hgg-20}
  4\alpha^\frac12 \Re \la P_f.\Aa \Proj^2 \gsps , \muvar_1 \fFz \ra | \geq
  4\alpha^\frac12 \Re \la P_f.\Aa \Proj^2 \gsps , \fFz \ra
  - |\muvar_1 -1| c\alpha^4 ,
\end{equation}
whose proof is similar to the one of \eqref{eq:estimate-diff-mu1}.
Next we get
\begin{equation}\label{eq:n-hgg-21}
\begin{split}
  & | 4\alpha^\frac12 \Re \la P_f.\Aa \Proj^2 \gsps , \fFz \ra -
  4\alpha^\frac12 \Re \la P_f.\Aa \Proj^2 \gsps , \fF\ra | \\
  & \leq
  \alpha^\frac12 \|H_f^\frac12 \Proj^2 \gsps \|\, \|P_f(\fFz -\fF)\| =
  \mathcal{O}(\alpha^6|\log\alpha|^\frac12),
\end{split}
\end{equation}
using $\|P_f(\fFz  -\fF)\| =
\mathcal{O}(\alpha^\frac72|\log\alpha|^\frac12)$ (see
Lemma~\ref{appendix:lem-A1}) and $\|H_f^\frac12 \Proj^2 \gsps \| =
\mathcal{O}(\alpha^2)$. Moreover,
\begin{equation}\label{eq:n-hgg-22}
\begin{split}
  & 4\alpha^\frac12 \Re \la P_f.\Aa \Proj^2 \gsps , \fF \ra \\
  & = 4 \alpha^\frac32 \Re \la P_f.\Aa (\muvar_2\cEpp  \Proj^0 \gsps  +
  \sum_i \muvar_{2,i} (H_f+P_f^2)^{-1} W_i \frac{\partial
  \Sgsa }{\partial x_i}), \fF\ra \\
  & = 4 \alpha^\frac32 \Re \la P_f.\Aa
  \sum_i \muvar_{2,i} (H_f+P_f^2)^{-1} W_i \frac{\partial
  \Sgsa }{\partial x_i}, \fF\ra \\
  & = \frac23 \alpha^4  \Re \sum_i
  \muvar_{2,i} \la P_f.\Aa
  (H_f+P_f^2)^{-1} W_i, (H_f+P_f^2)^{-1} (\Ac)^i\vac \ra\,
\end{split}
\end{equation}
where we used \eqref{eq:Pf-phi2-F} of Lemma~\ref{lem:appendix-3}
in the second equality.

Collecting \eqref{eq:n-hgg-13}-\eqref{eq:n-hgg-22} concludes the
proof.
\end{proof}

\begin{lemma}\label{lem:g1g1-3}
We have
\begin{equation}\nonumber
\begin{split}
 & \la \gsps , \gsps \ra_{\sharp} + 2\alpha\Re \la \Aa . \Aa \gsps , \gsps \ra
 \geq \Sigma_0 \left( \|\Proj^0 \gsps \|^2 + \|\Proj^1 \gsps \|^2\right)
 \\
 & + \frac{\alpha^4}{12} \sum_{i=1}^3 |\muvar_{2,i}|^2 \| (H_f+P_f^2)^{-1}
 W_i \|_{*}^2
 + \alpha^2 |\muvar_3 +1|^2  \|\cEpp \|_*^2 \|\fFz \|^2 \\
 & + \|\Proj^1 \gsps \|_{\sharp}^2 + \|\Proj^0 \gsps \|_{\sharp}^2 +
 o(\alpha^5\log\alpha^{-1}).
\end{split}
\end{equation}
\end{lemma}
%
%
%
\begin{proof}
Obviously we have
\begin{equation}\label{eq:n-hgg-14}
 \la \gsps , \gsps \ra_{\sharp} = \sum_{i=0}^3 \la \Proj^i \gsps , \Proj^i
 \gsps \ra_{\sharp} ,
\end{equation}
and using Lemma~\ref{lem:appendix-3}
\begin{equation}\label{eq:n-hgg-14bis}
 \la \Proj^2 \gsps , \Proj^2 \gsps \ra_{\sharp}
 = \alpha^2 |\muvar_2|^2 \|\cEpp  \Proj^0 \gsps \|_{\sharp}^2 +\alpha^2
 \sum_{i=1}^3 |\muvar_{2,i}|^2 \| (H_f + P_f^2)^{-1} W_i
 \frac{\partial \Sgsa }{\partial x_i}\|_{\sharp}^2 ,
\end{equation}

Moreover, from the inequality $\|\cEpp  \Proj^0 \gsps
\|_{\sharp}^2 > \|\cEpp \|_*^2 \|\Proj^0 \gsps \|^2$ we obtain,
\begin{equation}\label{eq:n-hgg-15}
\begin{split}
 & \alpha^2 |\muvar_2|^2 \|\cEpp  \Proj^0 \gsps \|_{\sharp}^2
 + 2\alpha\Re \la \Aa .\Aa \Proj^2 \gsps , \Proj^0 \gsps  \ra\\
 & \geq \alpha^2 |\muvar_2|^2 \|\cEpp \|_*^2 \|\Proj^0 \gsps \|^2
 + 2\alpha\Re \la \Aa . \Aa \Proj^2 \gsps , \Proj^0 \gsps \ra \\
 & = \alpha^2 |\muvar_2|^2 \|\cEpp \|_*^2 \|\Proj^0\gsps \|^2
 - 2\alpha^2 \Re \muvar_2 \|\cEpp \|_*^2 \|\Proj^0 \gsps \|^2\\
 & +
 2\alpha\Re\sum_{i=1}^3 \alpha \la \Aa . \Aa \muvar_{2,i}
 (H_f+P_f^2)^{-1} W_i \frac{\partial \Sgsa }{\partial x_i},
 \Proj^0 \gsps \ra \\
 & \geq \Sigma_0 \|\Proj^0 \gsps \|^2 + |\muvar_2 -1|^2 c\alpha^2 \|\Proj^0
 \gsps \|^2
 +\mathcal{O}(\alpha^5)
 \geq \Sigma_0 \|\Proj^0 \gsps \|^2
 +\mathcal{O}(\alpha^5).
\end{split}
\end{equation}
where we used $\Sigma_0 = - \alpha^2 \|\cEpp \|_*^2
+\mathcal{O}(\alpha^3)$ and $\la \Aa . \Aa (H_f+P_f^2)^{-1} W_i ,
\Proj^0 \gsps \ra = -\overline{\Proj^0 \gsps } \la W_i, \cEpp \ra
=0$.

Similarly $\|\Proj^3 \gsps \|_{\sharp} > \|\Proj^3 \gsps \|_*$
yields
\begin{equation}
\begin{split}
 & \|\Proj^3 \gsps \|_{\sharp}^2 +2\alpha\Re \la \Aa . \Aa \Proj^3 \gsps , \Proj^1
 \gsps \ra
 \geq \alpha^2 |\muvar_3|^2 \| (H_f+P_f^2)^{-1} \Ac .\Ac \fFz \|_*^2 \\
 & + 2\Re \alpha^2 \muvar_3 \la \Aa .\Aa (H_f+P_f^2)^{-1} \Ac .\Ac
 \fFz , \muvar_1 \fFz \ra \\
 & \geq -\alpha^2 \|(H_f+P_f^2)^{-1} \Ac .\Ac \muvar_1 \fFz \|_*^2
 + |\muvar_3 +1 |^2\alpha^2 \|(H_f+P_f^2)^{-1} \Ac .\Ac \fFz \|_*^2
 \\
 & \geq \Sigma_0 \|\Proj^1 \gsps \|^2 + \alpha^2 |\muvar_3+1|^2
 \|\cEpp \|_*^2 \|\fFz \|^2 + o(\alpha^5\log\alpha^{-1}) ,
\end{split}
\end{equation}
where in the last inequality, we used \eqref{appendix:lem9-1} from
Lemma~\ref{appendix:lem-A4}, $\muvar_1
=1+\mathcal{O}(\alpha^\frac12)$ from
Lemma~\ref{lem:improved-estimates}, $\|\fFz
\|^2=\mathcal{O}(\alpha^3\log\alpha^{-1})$ from
Lemma~\ref{appendix:lem-A0} and $-\alpha^2 \|\cEpp \|_*^2 =
\Sigma_0 + \mathcal{O}(\alpha^3)$.

The second term on the right hand side of \eqref{eq:n-hgg-14bis}
is estimated as
\begin{equation}\label{eq:n-hgg-16}
\begin{split}
 & \alpha^2
 \sum_{i=1}^3 |\muvar_{2,i}|^2 \| (H_f + P_f^2)^{-1} W_i
 \frac{\partial \Sgsa }{\partial x_i}\|_{\sharp}^2
 \geq \alpha^2
 \sum_{i=1}^3 |\muvar_{2,i}|^2 \| (H_f + P_f^2)^{-1} W_i
 \frac{\partial \Sgsa }{\partial x_i}\|_{*}^2 \\
 & = \alpha^2 \frac13\|\nabla \Sgsa \|^2
 \sum_{i=1}^3 |\muvar_{2,i}|^2 \| (H_f + P_f^2)^{-1} W_i
 \|_{*}^2 = \frac{\alpha^4}{12} \sum_{i=1}^3
 |\muvar_{2,i}|^2 \| (H_f + P_f^2)^{-1} W_i
 \|_{*}^2 .
\end{split}
\end{equation}

Collecting \eqref{eq:n-hgg-14}-\eqref{eq:n-hgg-16} concludes the
proof.
\end{proof}

\begin{proposition}\label{prop:n-hgg-*}
We have
\begin{equation}\label{eq:n-hgg-*}
\begin{split}
  & \la (H + e_0) \gsps , \gsps \ra \\
  & \geq -4\alpha  \|
  (h_\alpha + e_0)^{-\frac12} Q_\alpha^\perp P. \Aa \fF\|^2
  + \| (h_\alpha + e_0)^\frac12 \Proj^0 (\gsps )^a\|^2\\
  & + c_0 \alpha \| (h_\alpha +e_0)^\frac12 \Proj^0 (\gsps )^s\|^2
  + \Sigma_0 (\|\Proj^0 \gsps \|^2 + \|\Proj^1 \gsps \|^2) \\
  & + \frac{\alpha^4}{12} \sum_{i=1}^3 |\muvar_{2,i}|^2 \| (H_f+P_f^2)^{-1}
  W_i \|_{*}^2
  + 4 \alpha^\frac12\Re \la \Proj^2 \gsps , \Ac. P_f \fFz \ra \\
  & + \|\muvar_1 \fFz \|_{\sharp}^2 + 2 \alpha  \|\Aa \fF\|^2
  + \alpha^2 |\muvar_3+1|^2 \|\cEpp \|_*^2 \|\fFz \|^2 \\
  & - \epsilon M[\gspl ] - |\muvar_1-1| c\alpha^4+ o(\alpha^5\log\alpha^{-1}) ,
\end{split}
\end{equation}
where $c_0$ is the same positive constant as in
Proposition~\ref{prop:sub1} and
\begin{itemize}
 \item $Q_\alpha^\perp$ is the orthogonal projection onto
 $\mathrm{Span} (\Sgsa )^\perp$,
 \item $(\gsps )^a$ is the odd part
 of $\Proj^0 \gsps $.
\end{itemize}

\end{proposition}
%
%
%
\begin{proof}
Collecting Lemmata~\ref{lem:g1g1-1}, \ref{lem:g1g1-2}, and
\ref{lem:g1g1-3} yields
\begin{equation}\label{eq:n-hgg-17}
\begin{split}
 & \la (H+e_0)\gsps , \gsps \ra \geq
   - 4\alpha^\frac12 \Re \la \Aa \fFz , P\Proj^0 (\gsps )^s\ra
   + 4 \alpha\Re \la \Proj^2 \gsps , \Ac. P_f \fFz \ra\\
 &  + \Sigma_0 (\|\Proj^0 \gsps \|^2 +\|\Proj^1 \gsps \|^2)
 + \alpha^2 |\muvar_3 +1|^2
 \|\cEpp \|_*^2 \|\fFz \|^2  \\
 & +  \| (h_\alpha + e_0)^\frac12 \Proj^0 \gsps \|^2
 + \| \Proj^1 \gsps \|^2_{\sharp} +2 \alpha \|\Aa \muvar_1 \fFz \|^2
 - |\muvar_1 - 1|c\alpha^4 \\
 & + \frac{\alpha^4}{12} \sum_{i=1}^3 |\muvar_{2,i}|^2 \| (H_f+P_f^2)^{-1}
 W_i \|_{*}^2
 - \epsilon M[\gspl ] + o(\alpha^5\log\alpha^{-1}) .
\end{split}
\end{equation}
Obviously,
\begin{equation}\label{eq:n-hgg-18}
 \| (h_\alpha +e_0)^\frac12 \Proj^0 \gsps \|^2 = \| (h_\alpha
 +e_0)^\frac12 \Proj^0 (\gsps )^s \|^2 + \| (h_\alpha +e_0)^\frac12 \Proj^0
 (\gsps )^a\|^2 .
\end{equation}
As before, we write $\Proj^0 \gsps $ as the sum of its odd part
$(\gsps )^a$ and its even part $(\gsps )^s$. Since $\Proj^0 \gsps
$ is orthogonal to $\Sgsa $ by definition of $\gsp$ , and $(\gsps
)^a$ is orthogonal to $\Sgsa $ by symmetry of $\Sgsa $, we also
have $(\gsps )^s$ orthogonal to $\Sgsa $. Therefore, one can
replace $\Proj^0 (\gsps )^s$ by $Q_\alpha^\perp \Proj^0 (\gsps
)^s$ in \eqref{eq:n-hgg-17} and \eqref{eq:n-hgg-18}. Thus, as the
next step, given a constant $c_0>0$, we minimize
\begin{equation}\label{eq:n-hgg-18-a}
\begin{split}
 & (1-c_0\alpha) \|(h_\alpha +e_0)^\frac12 Q_\alpha^\perp\Proj^0 (\gsps )^s\|^2
 - 4\alpha^\frac12 \Re \la Q_\alpha^\perp P. \Aa \fFz ,
 Q_\alpha^\perp\Proj^0 (\gsps )^s\ra \\
 & \geq -\frac{4\alpha}{1 - c_0\alpha}
 \| (h_\alpha +e_0)^{-\frac12} Q_\alpha^\perp P. \Aa \fFz \|^2 .
\end{split}
\end{equation}
Obviously,
\begin{equation}\label{eq:n-hgg-19}
\begin{split}
  &  -\frac{4\alpha}{1 - c_0\alpha}
   \|(h_\alpha + e_0)^{-\frac12} Q_\alpha^\perp P.\Aa \fFz \|^2\\
  & \geq  - \frac{4(1+\alpha)\alpha}{1 - c_0\alpha}
  \|(h_\alpha + e_0)^{-\frac12} Q_\alpha^\perp P.\Aa \fF\|^2 \\
  &  - \frac{4(1+\alpha^{-1}) \alpha}{1 - c_0\alpha}
  \| (h_\alpha +e_0)^{-\frac12} Q_\alpha^\perp P. \Aa (\fFz  - \fF)\|^2.
\end{split}
\end{equation}
There exist $\gamma_1$ and $\gamma_2$ positive, independent of
$\alpha$, such that
 $$
  Q_\alpha^\perp (h_\alpha + e_0)^{-1} Q_\alpha^\perp
  \leq ( \gamma_1 P^2 + \gamma_2 \alpha^2)^{-1},
 $$
and thus $P Q_\alpha^\perp (h_\alpha +e_0)^{-1}Q_\alpha^\perp P$
is a bounded operator. In addition, since $\|\Aa (\fFz
-\fF)\|^2=\mathcal{O}(\alpha^7\log\alpha^{-1})$
(Lemma~\ref{appendix:lem-A0}), this shows that
 \begin{equation}\label{eq:prop-6.2-1}
  \frac{4(1+\alpha^{-1}) \alpha }{1 - c_0\alpha}
  \| (h_\alpha +e_0)^{-\frac12} Q_\alpha^\perp P. \Aa (\fFz  - \fF)\|^2
  = \mathcal{O}(\alpha^6\log\alpha^{-1}) .
 \end{equation}
In addition, using $\|\Aa \fF\| \leq \|\Aa (\fF-\fFz )\| + \|\Aa
\fFz \| \leq c\alpha^\frac32$ (Lemma~\ref{appendix:lem-A0} and
Lemma~\ref{appendix:lem-A1}), and the fact that $PQ_\alpha^\perp
(h_\alpha +e_0)^{-1}Q_\alpha^\perp P$ is bounded, yields
 \begin{equation}\label{eq:prop-6.2-2}
 \begin{split}
  &  - \frac{(1+\alpha)\alpha}{1 - c_0\alpha}
  \|(h_\alpha + e_0)^{-\frac12} Q_\alpha^\perp P.\Aa \fF\|^2 \\
  & = -4 \|(h_\alpha + e_0)^{-\frac12} Q_\alpha^\perp P.\Aa \fF\|^2
  +\mathcal{O}(\alpha^5).
 \end{split}
 \end{equation}

Collecting \eqref{eq:n-hgg-19}-\eqref{eq:prop-6.2-2}, one gets
\begin{equation}\label{eq:n-hgg-19-a}
\begin{split}
  & -\frac{4\alpha }{1 - c_0\alpha}
   \|(h_\alpha + e_0)^{-\frac12} Q_\alpha^\perp P.\Aa \fFz \|^2
  \geq -4
  \|(h_\alpha + e_0)^{-\frac12} Q_\alpha^\perp P.\Aa \fF\|^2
  +\mathcal{O}(\alpha^5).
\end{split}
\end{equation}

Finally, using $\|\Aa \fF\|^2=\mathcal{O}(\alpha^3)$, we obtain
\begin{equation}\label{eq:n-hgg-19-b}
 \begin{split}
  2\alpha |\muvar_1|^2 \|\Aa \fF\|^2 \geq
  2\alpha \|\Aa \fF\|^2 - c|\muvar_1 -1|\alpha^4 .
 \end{split}
\end{equation}

Substituting \eqref{eq:n-hgg-18}, \eqref{eq:n-hgg-18-a},
\eqref{eq:n-hgg-19-a} and \eqref{eq:n-hgg-19-b} into
\eqref{eq:n-hgg-17} concludes the proof.
\end{proof}

We can now collect the above results to prove Proposition~\ref{prop:prop-hgg}.

\subsection{Concluding the proof of Proposition~\ref{prop:prop-hgg}}

Collecting the
results of Proposition~\ref{prop:sub1},
Lemmata~\ref{lem:sub-1}-\ref{lem:sub-5} and
Proposition~\ref{prop:n-hgg-*} yields directly the following
bound,
\begin{equation}\label{eq:prop-hgg-0}
\begin{split}
 & \la H \gsp , \gsp\ra \geq
 \la H \gspl , \gspl \ra
 -\epsilon M[\gspl ] - c |\muvar_2|^2 \alpha^4 \|\Proj^0 \gsps \|^2
 - c\alpha^6\log\alpha^{-1} |\muvar_3|^2 - e_0\|\gsps \|^2 \\
 & - 4\alpha \|(h_\alpha+e_0)^{-\frac12} Q_\alpha^\perp P.\Aa
 \fF\|^2 + (1-c_0\alpha) \|(h_\alpha +e_0)^\frac12 \Proj^0
 (\gsps )^a\|^2\\
 & + \Sigma_0 (\|\Proj^0 \gsps \|^2 + \|\Proj^1 \gsps \|^2)
 + \frac{\alpha^4}{12} \sum_{i=1}^3 |\muvar_{2,i}|^2 \| (H_f+P_f^2)^{-1}
 W_i \|_{*}^2 \\
 & + 4 \alpha^\frac12\Re \la \Proj^2 \gspl , \Ac. P_f \fF\ra
 + 4 \alpha^\frac12\Re \la \Proj^2 \gsps , \Ac. P_f \fFz \ra
 + |\muvar_1|^2 \|\fFz \|_{\sharp}^2 + 2\alpha  \|\Aa\fF\|^2 \\
 & + \alpha^2 |\muvar_3+1|^2
 \|\cEpp \|_*^2 \|\fFz \|^2 - |\muvar_1 -1| c\alpha^4
 + o(\alpha^5\log\alpha^{-1}) .
\end{split}
\end{equation}

Comparing this expression with the statement of the Proposition we
see that it suffices to show that
\begin{equation}\label{eq:prop-hgg-3}
\begin{split}
 & -\epsilon M[\gspl ] - c |\muvar_2|^2 \alpha^4 \|\Proj^0 \gsps \|^2
 - c\alpha^6\log\alpha^{-1} |\muvar_3|^2 -e_0\|\gsps \|^2
 +  \la H \gspl , \gspl \ra\\
 & + \Sigma_0 (\|\Proj^0 \gsps \|^2 + \|\Proj^1 \gsps \|^2)
 + \alpha^2 |\muvar_3+1|^2
 \|\cEpp \|_*^2 \|\fFz \|^2 \\
 & \geq (\Sigma_0 - e_0) \|\gsp\| ^2 + (1-\epsilon) M[\gspl ]
 + \frac{|\muvar_3+1|^2}{2} \alpha^2 \|\cEpp \|_*^2 \|\fFz \|^2
 + \mathcal{O}(\alpha^5) .
\end{split}
\end{equation}

Using from Corollary~\ref{cor:cor-4-2} that $\la H \gspl , \gspl
\ra \geq (\Sigma_0 - e_0)\|\gspl \|^2 + M[\gspl ]$, we first
estimate the following terms in \eqref{eq:prop-hgg-3},
\begin{equation}\label{eq:prop-hgg-4}
\begin{split}
 & \Sigma_0 \left(\|\Proj^0 \gsps \|^2 + \|\Proj^1 \gsps \|^2\right)
 -e_0\|\gsps \|^2 +  \la H \gspl , \gspl \ra   -\epsilon M[\gspl ]\\
 & \geq (1 - \epsilon) M[\gspl ]
 + (\Sigma_0 - e_0) (\|\gsps \|^2 + \|\gspl \|^2) -
 \Sigma_0 \| \Proj^{n\geq 2} \gsps \|^2 \\
 & \geq (1-\epsilon) M[\gspl ] + (\Sigma_0 - e_0) \|\gsp\| ^2 \\
 & - (\Sigma_0 - e_0) (\|\gsp\| ^2 -\|\gsps \|^2 -\|\gspl \|^2) -
 \Sigma_0 \| \Proj^{n\geq 2} \gsps \|^2 .
\end{split}
\end{equation}

We have obviously
\begin{equation}\label{eq:n-hgg-4a}
 \|\gsp\| ^2 - \|\gsps \|^2 - \|\gspl \|^2 = 2\Re (\la \Proj^1 \gsps , \Proj^1
 \gspl \ra + \la \Proj^2 \gsps , \Proj^2 \gspl \ra + \la \Proj^3 \gsps , \Proj^3
 \gspl \ra ).
\end{equation}
Since $|\Sigma_0 -e_0|\leq c\alpha^2$, by definition of $\Proj^3
\gsps $ and Lemma~\ref{appendix:lem-A4}, we obtain
\begin{equation}\label{eq:n-hgg-21}
\begin{split}
 & | (\Sigma_0 -e_0) 2\Re \la \Proj^3 \gsps , \Proj^3 \gspl \ra|
 \leq \epsilon\alpha^2 \|\Proj^3 \gspl \|^2
 + c\alpha^4 |\muvar_3|^2 \|(H_f+P_f^2)^{-1} \Ac .\Ac \fFz  \|^2 \\
 & \leq \epsilon\alpha^2 \|\Proj^3 \gspl \|^2
 + c|\muvar_3|^2 \alpha^7\log\alpha^{-1} .
\end{split}
\end{equation}
Similarly, for the two-photon sector, we find
\begin{equation}\label{eq:n-hgg-22}
\begin{split}
 & |(\Sigma_0 -e_0) \la \Proj^2 \gsps , \Proj^2 \gspl \ra|\leq\epsilon\alpha^2
 \|\Proj^2 \gspl \|^2
 + c\alpha^4 |\muvar_2|^2 \|\Proj^0 \gsps \|^2
 + c\sum_{i=1}^3 |\muvar_{2,i}|^2 \alpha^6 \\
 & = \epsilon \alpha^2 \|\Proj^2 \gspl \|^2 + \mathcal{O}(\alpha^6) ,
\end{split}
\end{equation}
where we used Lemma~\ref{lem:improved-estimates}. For the term
$\la \Proj^1 \gsps , \Proj^1 \gspl \ra$, one gets
\begin{equation}\label{eq:n-hgg-23}
\begin{split}
 & |(\Sigma_0-e_0)\la \Proj^1 \gspl , \muvar_1\fFz \ra|
 = |(\Sigma_0 -e_0) (|k|^\frac12 \Proj^1 \gspl , |k|^{-\frac12} \muvar_1
 \fFz \ra|\\
 & \leq \epsilon\| H_f^\frac12 \Proj^1 \gspl \|^2 +
 c\alpha^4|\muvar_1|^2 \| \, |k|^{-\frac12} \fFz \|^2
 \leq \epsilon\| H_f^\frac12 \Proj^1 \gspl \|^2 +
 \mathcal{O}(\alpha^5),
\end{split}
\end{equation}
since  $\|\,|k|^{-\frac12} \fFz \|^2 = \mathcal{O}(\alpha)$
(Lemma~\ref{appendix:lem-A0}) and $\muvar_1=\mathcal{O}(1)$
(Lemma~\ref{lem:improved-estimates}).

Collecting \eqref{eq:n-hgg-4a}, \eqref{eq:n-hgg-21},
\eqref{eq:n-hgg-22} and \eqref{eq:n-hgg-23} yields
\begin{equation}\label{eq:n-hgg-24}
 |\Sigma_0 -e_0|\, \left| \, \|\gsp\| ^2 -\|\gsps \|^2 -\|\gspl \|^2\right|
 \leq \epsilon M[\gspl ] + c|\muvar_3|^2\alpha^7\log\alpha^{-1}
 +\mathcal{O}(\alpha^5) .
\end{equation}
Therefore, together with \eqref{eq:prop-hgg-4}, one finds
\begin{equation}\label{eq:prop-hgg-5}
\begin{split}
& \Sigma_0 \left(\|\Proj^0 \gsps \|^2 + \|\Proj^1 \gsps
\|^2\right)
 -e_0\|\gsps \|^2 +  \la H \gspl , \gspl \ra   -\epsilon M[\gspl ] \\
& \geq (1-2\epsilon) M[\gspl ] + (\Sigma_0 - e_0)\|\gsp\| ^2
 - c|\muvar_3|^2\alpha^7\log\alpha^{-1} - \Sigma_0\|\Proj^{n\geq
 2}\gsps \|^2
 +\mathcal{O}(\alpha^5) ,
\end{split}
\end{equation}
By definition of $\gsps $ and using
$\Sigma_0=\mathcal{O}(\alpha^2)$, $|\muvar_2| \|\Proj^0 \gsps
\|=\mathcal{O}(\alpha)$ (Lemma~\ref{lem:improved-estimates}) and
Inequality \eqref{appendix:lem9-1-bis} of
Lemma~\ref{appendix:lem-A4}, we straightforwardly obtain
 $$
  \Sigma_0\|\Proj^{n\geq
 2}\gsps \|^2 \leq
 c \alpha^{6} + c |\muvar_3|^2 \alpha^{7}\log\alpha^{-1}.
 $$
Substituting this in \eqref{eq:prop-hgg-5} yields
\begin{equation}\label{eq:prop-hgg-6}
\begin{split}
& \Sigma_0 \left(\|\Proj^0 \gsps \|^2 + \|\Proj^1 \gsps
\|^2\right)
 -e_0\|\gsps \|^2 +  \la H \gspl , \gspl \ra   -\epsilon M[\gspl ] \\
& \geq (1-2\epsilon) M[\gspl ] + (\Sigma_0 - e_0)\|\gsp\| ^2
 - 2c|\muvar_3|^2\alpha^7\log\alpha^{-1} + \mathcal{O}(\alpha^5) .
\end{split}
\end{equation}

To conclude the proof of \eqref{eq:prop-hgg-3}, and thus of the
Proposition, we first note that according to
Lemma~\ref{lem:improved-estimates},
\begin{equation}\label{eq:prop-hgg-1}
 - c |\muvar_2|^2 \alpha^4 \|\Proj^0 \gsps \|^2
 = \mathcal{O}(\alpha^6) .
\end{equation}
Similarly, taking into account that $\|\fFz \|^2=c
\alpha^3\log\alpha^{-1}$ (see \eqref{eq:app-i} in
Lemma~\ref{appendix:lem-A0}), we get for some $c_2>0$,
\begin{equation}\label{eq:prop-hgg-2}
  \alpha^2 \frac{|\muvar_3+1|^2}{2}
 \|\cEpp \|_*^2 \|\fFz \|^2 - c\alpha^6\log\alpha^{-1} |\muvar_3|^2
 - 2 c \alpha^7\log\alpha^{-1} |\muvar_3|^2
 \geq -c_2 \alpha^6\log\alpha^{-1} .
\end{equation}

Collecting \eqref{eq:prop-hgg-6}, \eqref{eq:prop-hgg-1} and
\eqref{eq:prop-hgg-2} yields the bound \eqref{eq:prop-hgg-3}, and
thus concludes the proof of the Proposition~\ref{prop:prop-hgg}.
\qed

\section{Proof of Theorem~\ref{thm:main}: Auxiliary Lemmata}
\label{app-Thm21prf-1}

The proof of Theorem~\ref{thm:main} uses Lemmata~\ref{lem:appendix-3} $\sim$
~\ref{appendix:lem-A5} below.


\begin{lemma}\label{appendix:lem-A10}
Let $\Psi_0$ be the ground state of $T(0)$, with
$|\Proj^0\Psi_0|=1$. Then we have
\begin{equation}\nonumber
\begin{split}
 &\Psi_0 = \Omega_f + 2\eta_1 \alpha^\frac32 \cEm  +
 \eta_2 \alpha \cEpp  + 2\eta_3 \alpha^\frac32 \cEp  +\Rnvar,
 \\
 & \mbox{with }\la\Rnvar,\cEi\ra_* =0 \ (i=1,2,3) , \Proj^0\Rnvar=0 \
\end{split}
\end{equation}
and
\begin{equation}\nonumber
 \|\Rnvar\|^2 =\mathcal{O}(\alpha^3) .
\end{equation}
\end{lemma}
%
%
%
\begin{proof}
It follows from a similar argument as in
\cite[Proposition~5.1]{CF2007} that
 $$
   \| a_\lambda(k)\Rnvar\| \leq \frac{c\alpha}{| k |} .
 $$
This yields
\begin{equation}\nonumber
\begin{split}
 \|\Rnvar\|^2 & \leq \int \|a_\lambda(k)\Rnvar\|^2 \d k \leq \int_{|k|\leq
 \alpha} \frac{c^2\alpha^2}{|k|^2} \d k + \int_{|k|>\alpha}
 \frac{|k|\, \|a_\lambda(k)\Rnvar\|^2 \chi_\Lambda(|k|)}{|k|} \d
 k \\
 & \leq c^2\alpha^3 + c'\alpha^{-1} \|H_f^\frac12\Rnvar\|^2 =
 \mathcal{O}(\alpha^3) \ ,
\end{split}
\end{equation}
where in the last equality we used the estimate $\|\Rnvar\|_*^2
 = \|(H_f+P_f^2)^\frac12\Rnvar\|^2 =\mathcal{O}(\alpha^4)$ proved in
\cite[Theorem~3.2]{BCVVi}.
\end{proof}


\begin{lemma}\label{lem:appendix-3}
We have
\begin{equation}\nonumber
 P.\Aa \Sgsa  \cEm  =0,
\end{equation}
and
\begin{equation}\label{eq:0-with-Pf}
\begin{split}
      \la \cEpp , \zeta(H_f, P_f^2) P_f^i\cEpp \ra = 0\ ,
      \ \mbox{and}\ \la \cEpp , \zeta(H_f,P_f^2) W_i\ra =0 \ ,
\end{split}
\end{equation}
for any function $\zeta$ for which the scalar products are
defined. Similarly, we have
\begin{equation}\label{eq:Pf-phi2-F}
 \la \cEpp  \Proj^0 \gsps , \Ac . P_f \fF \ra = 0 .
\end{equation}
\end{lemma}
%
%
%
\begin{proof}
Straightforward computations using the symmetries of $\Aa \cEm $
and $\cEpp $.
\end{proof}

\begin{lemma}\label{appendix:lem-A0-ter}
We have
\begin{equation}\nonumber
      P.\Aa \fF = \sqrt{\alpha} a_0
     \Delta \Sgsa   ,
\end{equation}
where
\begin{equation}\nonumber
     a_0 = \int
     \frac{k_1^2+k_2^2}{4\pi^2|k|^3} \frac{2}{|k|^2
     +|k|} \chi_\Lambda (|k|)\, \d k_1 \d k_2 \d k_3.
\end{equation}
\end{lemma}
%
%
%
\begin{proof}
Straightforward computations using the symmetries of $\Aa \fF$.
\end{proof}

\begin{lemma}\label{appendix:lem-A0}
\begin{eqnarray}
 & & \forall\varphi\in\mathfrak{F} ,\
 \la P.P_f \Proj^1\varphi \Sgsa , \fFz \ra = 0\ ,
 \label{eq:app-0} \\
 & & \|\fFz \|^2 =\mathcal{O}(\alpha^3\log\alpha^{-1})\ ,\label{eq:app-i}\\
 & & \|\fFz \|_*^2 = \mathcal{O}(\alpha^3) \ ,\label{eq:app-ii}\\
 & & \|\, |k|^{-\frac12} \fFz \,\|^2 =
 \mathcal{O}(\alpha)\label{eq:app-iii}\ .
\end{eqnarray}
\begin{eqnarray}
 & & \| P\fFz  \|^2 = \mathcal{O}(\alpha^5\log\alpha^{-1})\ , \label{appendix:lem7-1}\\
 & & \| P \fFz \|_*^2 = \mathcal{O}(\alpha^5)\
 ,\label{appendix:lem7-2} \\
 & & \|\, |k|^{\frac16} P \fFz \|^2 =\mathcal{O}(\alpha^5)\ .
 \label{appendix:lem7-2}
\end{eqnarray}

\end{lemma}
%
%
%
\begin{proof}
The proof of \eqref{eq:app-0} is as follows
\begin{equation}
\begin{split}
 & \la P.P_f \Proj^1 \varphi \Sgsa , \fFz \ra\\
 & = \int \sum_{i=1}^3 k^i \frac{\partial \Sgsa }{\partial x_i}
 \varphi(k) \frac{1}{k^2 + |k| +h_\alpha+e_0}\sum_{j=1}^3 \sum_{\lambda=1,2}
 \frac{\epsilon_\lambda^j(k) \chi_\Lambda(|k|)}{2\pi |k|^\frac12}
 \overline{\frac{\partial \Sgsa }{\partial x_j}} \mathrm{d}x \mathrm{d}k\\
 & =  \sum_{i=1}^3  \int \mathrm{d} k \left(\int\mathrm{d} x
 \frac{\partial \Sgsa }{\partial x_i}\frac{1}{k^2 + |k| +h_\alpha+e_0}
 \overline{\frac{\partial \Sgsa }{\partial x_i}}\right)
 \sum_{\lambda=1,2}
 \frac{k^i \epsilon_\lambda^i(k) \chi_\Lambda(|k|)}{2\pi
 |k|^\frac12}\varphi(k)
  \mathrm{d}k =0,
\end{split}
\end{equation}
using that the integral over $x$ is independent of the value of
$i$, and since $k.\epsilon_\lambda(k)=0$.

To prove \eqref{eq:app-iii}, we note that
\begin{equation}\nonumber
\begin{split}
 \| \, |k|^{-\frac12} \fFz  \|^2 & \leq c\alpha \int\left|
 \frac{\chi_\Lambda(|k|)}{|k| \, (|k| + k^2 + h_\alpha + e_0)}
 \nabla \Sgsa \right|^2 \d k \d x \\
 & \leq c\alpha^3 \int \frac{\chi_\Lambda(|k|)^2}{|k|^2  (|k| + \frac{3}{16}
 \alpha^2)^2}\d k =\mathcal{O}(\alpha) .
\end{split}
\end{equation}

The proofs of \eqref{eq:app-i} and \eqref{eq:app-ii} are similar,
but simpler.

We next prove \eqref{appendix:lem7-1}.
\begin{equation}\label{appendix:lem7-3}
 \| P \fFz \|^2
 = c\alpha \sum_{i=1}^3 \int \|P
 \frac{\chi_\Lambda(|k|)}{ |k|^\frac12 (|k| + k^2 + h_\alpha + e_0)}
 \frac{\partial \Sgsa }{\partial x_i} \|_{L^2(\R^3)}^2 \d k .
\end{equation}
The function $\partial \Sgsa /\partial x_i$ is odd. On the
subspace of  antisymmetric functions on $L^2(\R^3)$, one has that
$-(1-\gamma_0) \Delta -\frac{\alpha}{|x|} > - e_0$ for some
$\gamma_0>0$, which implies on this subspace that $h_\alpha + e_0
> \gamma_0 P^2$, and thus,
\begin{equation}\label{appendix:lem7-4}
  P^2 < \gamma_0^{-1} (h_\alpha + e_0) .
\end{equation}
The relation \eqref{appendix:lem7-4} yields, for all $k$
\begin{equation}\label{appendix:lem7-5}
\begin{split}
 & \| P\frac{\chi_\Lambda(|k|)}{|k|^\frac12 (|k| + k^2 +h_\alpha
 +e_0)}\frac{\partial \Sgsa }{\partial x_i} \|^2 \\
 & \leq \gamma_0^{-1} \| \frac{\chi_\Lambda(|k|)}{|k|^\frac12
 (|k| + k^2 +h_\alpha
 +e_0)} (h_\alpha +e_0)^\frac12
 \frac{\partial \Sgsa } {\partial x_i}\|^2 \\
 & \leq \gamma_0^{-1} c \alpha^4 \left(
 \frac{\chi_\Lambda(|k|)}
 {|k|^\frac12 (|k| + k^2 +\frac{3}{16}\alpha^2)}\right)^2.
\end{split}
\end{equation}
Substituting \eqref{appendix:lem7-5} into \eqref{appendix:lem7-3}
and integrating over $k$ proves \eqref{appendix:lem7-1}.

The proof of \eqref{appendix:lem7-2} is similar.
\end{proof}

\begin{lemma}\label{appendix:lem-A1}We have
\begin{equation}\nonumber
\begin{split}
 & \| \fF\|_*^2 - \|\fFz \|_{\sharp}^2 =
 \frac{1}{3\pi} \|(h_1+\frac14)^\frac12 \nabla \Sgso\|^2
 \alpha^5\log\alpha^{-1} + o(\alpha^5\log\alpha^{-1})\ , \\
 & \| \fF - \fFz \|_*^2 =\mathcal{O} (\alpha^5)\ ,\\
 & \| \Aa (\fFz  -\fF)\|^2 =\mathcal{O}(\alpha^7\log\alpha^{-1}) .
\end{split}
\end{equation}
\end{lemma}
%
%
%
\begin{proof}
We have
\begin{equation}\label{eq:appendix-lem7-1}
\begin{split}
 & \| \fF\|_*^2 - \| \fFz \|_{\sharp}^2 \\
 & =
 \frac{\alpha^5}{(2\pi)^2} \frac43\left\|
 \frac{\chi_\Lambda(|k|)}
 {|k|^\frac12 (|k| + k^2)^\frac12 (|k| + k^2 + \alpha^2(h_1
 +\frac14))^\frac12} (h_1 +\frac14)^\frac12 \nabla
 \Sgso\right\|^2\ ,
\end{split}
\end{equation}
where the norm here is obviously taken on $L^2(\R^3,\, \d
x)\otimes L^2(\R^3,\, \d k)$. Since $(h_1 +\frac14) \nabla \Sgso
\in L^2$, it implies that for sufficiently large $c>0$ independent
of $\alpha$,
\begin{equation}\nonumber
 \| \chi( h_1 > c ) (h_1+\frac14)^\frac12 \nabla \Sgso\|^2 <\epsilon
 ,
\end{equation}
and
\begin{equation}\label{appendix:lem6-1}
\begin{split}
 & \left\| \frac{\chi_\Lambda(|k|) }{ |k|^\frac12 (|k| + k^2)^\frac12
 (|k| + k^2 +\alpha^2 (h_1+\frac14))^\frac12} \chi(h_1 >c)
 (h_1 + \frac14)^\frac12 \nabla \Sgso \right\|^2 \\
 & \leq \epsilon \int_0^\infty \frac{ \chi_\Lambda^2( t )} {t +
 c\alpha^2} \d t =\epsilon\log\alpha^{-1} +\mathcal{O}(1) .
\end{split}
\end{equation}
For the contribution of $\chi(h_1\leq c) (h_1 +\frac14)^\frac12
\nabla \Sgso$ in \eqref{eq:appendix-lem7-1}, the following
inequalities hold,
\begin{equation}\label{appendix:lem6-2}
\begin{split}
 & (\, \frac{1}{\pi} \log\alpha^{-1} + \mathcal{O}(1) \,) \frac13 \|
 \chi(h_1<c) (h_1+\frac14)^\frac12\nabla \Sgso\|^2 \\
 & = \frac{1}{(2\pi)^2} \frac13 \left\| \frac{\chi_\Lambda(|k|) }{
 |k|^\frac12 (|k| + k^2)^\frac12
 (|k| + k^2 + (c+\frac14)  \alpha^2 ))^\frac12} \chi(h_1 <c)
 (h_1 + \frac14)^\frac12 \nabla \Sgso
 \right\|^2 \\
 & \leq
 \frac{1}{(2\pi)^2} \frac13 \left\| \frac{\chi_\Lambda(|k|) }{
 |k|^\frac12 (|k| + k^2)^\frac12
 (|k| + k^2 + (h_1 + \frac14)  \alpha^2 ))^\frac12} \chi(h_1 <c)
 (h_1 + \frac14)^\frac12 \nabla \Sgso
 \right\|^2 \\
 & \leq \frac{1}{(2\pi)^2} \frac13 \left\| \frac{\chi_\Lambda(|k|) }{
 |k|^\frac12 (|k| + k^2)^\frac12
 (|k| + k^2 + \frac{3}{16}  \alpha^2 ))^\frac12} \chi(h_1 <c)
 (h_1 + \frac14)^\frac12 \nabla \Sgso
 \right\|^2 \\
 & \leq (\, \frac{1}{\pi} \log\alpha^{-1} + \mathcal{O}(1) \,)
 \frac13 \| \chi(h_1<c) (h_1+\frac14)^\frac12\nabla \Sgso\|^2 .
\end{split}
\end{equation}
The inequalities \eqref{appendix:lem6-1} and
\eqref{appendix:lem6-2} prove the first equality of the Lemma.

The proofs of the last two equalities are similar but simpler.
\end{proof}

\begin{lemma}\label{appendix:lem-A4}
\begin{eqnarray}
 & & \| (H_f + P_f^2)^{-1} \Ac .\Ac \fFz \|_{*}^2
 - \|\cEpp \|_*^2 \|\fFz \|^2 = o(\alpha^3\log\alpha^{-1}) ,\label{appendix:lem9-1}\\
 & & \| (H_f + P_f^2)^{-1} \Ac .\Ac \fFz \|^2 =
 \mathcal{O}(\alpha^3\log\alpha^{-1}) \label{appendix:lem9-1-bis} \\
 & & \|\, |k_1|^{\frac16}|k_2|^{\frac16}|k_3|^{\frac16}
  (H_f + P_f^2)^{-1}
  \Ac .\Ac \fFz \|^2 = \mathcal{O}(\alpha^3)
  \label{appendix:lem9-1-ter} \\
 & & \|  (h_\alpha + e_0)
 (H_f+P_f^2)^{-1}\Ac .\Ac \fFz \|^2 =
 \mathcal{O}(\alpha^7\log\alpha^{-1}) \label{appendix:lem9-1-last}.
\end{eqnarray}
\end{lemma}
%
%
%
\begin{proof}
Denoting by $\sigma_n$ the set of all permutations of
$\{1,2,...,n\}$, we have
\begin{equation}\nonumber
\begin{split}
  & \| (H_f + P_f^2)^{-1} \Ac .\Ac \fFz \|_*^2 \\
  & = \frac{4\alpha}{(2\pi)^3} \Big\| \frac{1}{\sqrt{6}} \sum_{(i,j,n)\in\sigma_n}
  \frac{\sum_{\lambda=1,2} \varepsilon_\lambda(k_i) .
  \sum_{\nu=1,2} \varepsilon_\nu(k_j)}
  {\left(\sum_{p=1}^3 |k_p| + (\sum_{p=1}^3 k_p)^2\right)^\frac12}
  \\
  & \times \frac{\sum_{\muvar=1,2}\varepsilon_\muvar(k_n)
  \chi_\Lambda(|k_1|)
  \chi_\Lambda(|k_2|) \chi_\Lambda(|k_3|)}{|k_i|^\frac12 |k_j|^\frac12 |k_n|^\frac12
  (|k_n| + k_n^2 +(h_\alpha+e_0))} \nabla \Sgsa \Big\|^2
\end{split}
\end{equation}
If we pick two triples $(i,j,n)$ and $(i',j',n')$, such that
$n\neq n'$, then we get a product which is integrable at $k_1 =
k_2 = k_3 =0$, even without the term $(h_\alpha +e_0)$. The
contribution of such terms is $c\alpha \|\nabla \Sgsa \|^2 =
\mathcal{O}(\alpha^3)$. Moreover, using the symmetry in $k_1$,
$k_2$, $k_3$, the twelve remaining  terms give the same
contribution. This yields
\begin{equation}\label{appendix:lem9-3}
\begin{split}
  & \| (H_f + P_f^2)^{-1} \Ac .\Ac \fFz \|_*^2 \\
  & = \frac{8 \alpha}{(2\pi)^3} \Big\|
  \frac{\sum_{\lambda} \varepsilon_\lambda(k_3) .
  \sum_{\nu} \varepsilon_\nu(k_2)
  \sum_{\muvar}\varepsilon_\muvar(k_1)
  \chi_\Lambda(|k_1|)
  \chi_\Lambda(|k_2|) \chi_\Lambda(|k_3|)}
  {\left(\sum_{i=1}^3 |k_i| + (\sum_{i=1}^3 k_i)^2\right)^\frac12
  |k_3|^\frac12 |k_2|^\frac12 |k_1|^\frac12
  (|k_1| + k_1^2 +(h_\alpha+e_0))}
  .\nabla \Sgsa \Big\|^2 \\
  & + \mathcal{O}(\alpha^3).
\end{split}
\end{equation}
On the other hand, we have
\begin{equation}\nonumber
\begin{split}
  & \| \cEpp  \|_*^2 \|\fFz \|^2 \\
  & = \frac{8 \alpha}{(2\pi)^3} \Big\|
  \frac{\sum_{\lambda} \varepsilon_\lambda(k_3) .
  \sum_{\nu} \varepsilon_\nu(k_2)
  \sum_{\muvar}\varepsilon_\muvar(k_1)
  \chi_\Lambda(|k_1|)
  \chi_\Lambda(|k_2|) \chi_\Lambda(|k_3|)}
  {\left(|k_2| +|k_3|  + (k_2 + k_3)^2\right)^\frac12
  |k_3|^\frac12 |k_2|^\frac12 |k_1|^\frac12
  (|k_1| + k_1^2 +(h_\alpha+e_0))}
  .\nabla \Sgsa \Big\|^2 .
\end{split}
\end{equation}
Therefore, we obtain
\begin{equation}\label{eq:appendix-fin_1}
\begin{split}
  & \| (H_f + P_f^2)^{-1} \Ac.\Ac \fFz \|^2_* - \| \cEpp  \|_*^2 \|\fFz \|^2 \\
  & = - \frac{8 \alpha}{(2\pi)^3} \int \bigg( \frac{(|k_1|^2 + 2|k_1|\,
  |k_3| + 2 |k_1|\, |k_2|)}{|k_2|\, |k_3|\,
  (|k_1|+|k_2|+|k_3| + |k_1+k_2+k_3|^2) (|k_2|+|k_3| + |k_2+k_3|^2)
  } \\
  & \times \sum_{\lambda} \varepsilon_\lambda(k_3) .
  \sum_{\nu} \varepsilon_\nu(k_2)\chi_\Lambda(|k_1|) \chi_\Lambda(|k_2|)^2
  \chi_\Lambda(|k_3|)^2 |u(k_1,x)|^2\d k_1 \d k_2 \d k_3 \d x\bigg) \ ,
\end{split}
\end{equation}
where
\begin{equation}\nonumber
\begin{split}
 u(k_1,x) =  \frac{\sum_{\muvar}\varepsilon_\muvar(k_1)
  \chi_\Lambda(|k_1|)^\frac12}
  {|k_1|^\frac12 (|k_1| + k_1^2 +(h_\alpha+e_0))}
  .\nabla \Sgsa \ .
\end{split}
\end{equation}
For fixed $\delta$, we first compute the integral in
\eqref{eq:appendix-fin_1} over the regions $|k_1|>\delta$. This
yields a term $c_\delta \alpha^3$, where $c_\delta$ is independent
on $\alpha$.

Next, integrating \eqref{eq:appendix-fin_1} over the regions
$|k_1|\leq \delta$ yields a bound $\mathcal{O}(\delta)
\alpha^3\log\alpha$, with $\mathcal{O}(\delta)$ independent of
$\alpha$.

This concludes the proof of \eqref{appendix:lem9-1}.

The proof of \eqref{appendix:lem9-1-bis} is a straightforward
computation showing
\begin{equation}\nonumber
\begin{split}
  & \| (H_f + P_f^2)^{-1} \Ac .\Ac \fFz \|^2 \\
  & = \frac{8 \alpha}{(2\pi)^3} \Big\|
  \frac{\sum_{\lambda} \varepsilon_\lambda(k_3) .
  \sum_{\nu} \varepsilon_\nu(k_2)
  \sum_{\muvar}\varepsilon_\muvar(k_1)
  \chi_\Lambda(|k_1|)
  \chi_\Lambda(|k_2|) \chi_\Lambda(|k_3|)}
  {\left(\sum_{i=1}^3 |k_i| + (\sum_{i=1}^3 k_i)^2\right)
  |k_3|^\frac12 |k_2|^\frac12 |k_1|^\frac12
  (|k_1| + k_1^2 +(h_\alpha+e_0))}
  .\nabla \Sgsa \Big\|^2 \\
  & = \mathcal{O}(\alpha^3\log\alpha^{-1}) \ .
\end{split}
\end{equation}

The proofs of \eqref{appendix:lem9-1-ter} and
\eqref{appendix:lem9-1-last} are similar to the above.
\end{proof}

\begin{lemma}\label{appendix:lem-A5}
For $\gsps $ and $\gspl $ given in
Definition~\ref{def:G-improved}, we have
\begin{equation}
  |2\alpha\Re\la\Aa.\Aa \Proj^3 \gsps , \Proj^1 \gspl \ra|
  \leq c\alpha^5 + \epsilon \|H_f^\frac12 \gspl \|^2\ .
\end{equation}
\end{lemma}
%
%
%
\begin{proof}
Using $2\alpha\Re\la\Aa.\Aa \Proj^3 \gsps , \Proj^1 \gspl \ra =
2\alpha\Re\la \Proj^3 \gsps , \Ac.\Ac \Proj^1 \gspl \ra$, we
obtain that $2\alpha\Re\la\Aa.\Aa \Proj^3 \gsps , \Proj^1 \gspl
\ra$ can be rewritten as a linear combination of the following two
integrals $I_1$ and $I_2$
\begin{equation}
\begin{split}
  & I_1 = \alpha^\frac52 \int \frac{\d k\, \d k'\,
  \d k''\, \d x}{|k|\!+\!|k'|\!+\!|k''| + |k\!+\!k'\!+\!k''|^2}
  \sum_{\muvar, \nu} \frac{ \epsilon_\muvar (k') . \epsilon_\nu (k'') }
  {|k'|^{\frac12} \, |k''|^\frac12}\\
  & \times  \Big( \frac{1}{|k|\!+\!|k|^2 +
  (h_\alpha\! + \!e_0)} \sum_\lambda
  \frac{\epsilon_\lambda(k)}{|k|^\frac12} \nabla \Sgsa  \Big)
  \sum_{\kappa, \eta}
  \frac{ \epsilon_\kappa (k') . \epsilon_\eta (k'') }
  {|k'|^{\frac12} \, |k''|^\frac12} \overline{(\Proj^1 \gspl )(k, x)} ,
\end{split}
\end{equation}
and $I_2$ is defined as $I_1$, except that in the last sum, we
reverse the role of $k$ and $k''$, namely
\begin{equation}
\begin{split}
  & I_2 = \alpha^\frac52\int \frac{\d k\, \d k'\,
  \d k''\, \d x \chi_\Lambda(|k|) \chi_\Lambda(|k'|) \chi_\Lambda(|k''|)}
  {|k|\!+\!|k'|\!+\!|k''| + |k\!+\!k'\!+\!k''|^2}
  \sum_{\muvar, \nu} \frac{ \epsilon_\muvar (k') . \epsilon_\nu (k'') }
  {|k'|^{\frac12} \, |k''|^\frac12}\\
  & \times  \Big( \frac{1}{|k|\!+\!|k|^2 +
  (h_\alpha\! + \!e_0)} \sum_\lambda
  \frac{\epsilon_\lambda(k)}{|k|^\frac12} \nabla \Sgsa  \Big)
  \sum_{\kappa, \eta}
  \frac{ \epsilon_\kappa (k') . \epsilon_\eta (k) }
  {|k'|^{\frac12} \, |k|^\frac12} \overline{(\Proj^1 \gspl )(k'', x)} .
\end{split}
\end{equation}
To bound $I_2$, we use the Schwarz inequality, $|ab| \leq
\frac{1}{2\epsilon}a^2 + \frac{\epsilon}{2} b^2$. This yields
\begin{equation}\nonumber
\begin{split}
  & |I_2| \leq \alpha^\frac52 \Big(\int
  \frac{\d k\, \d k'\,
  \d k''\, \d x\chi_\Lambda^2(|k|)\chi_\Lambda^2(|k'|)\chi_\Lambda^2(|k''|)}
  {(|k|\!+\!|k'|\!+\!|k''| + |k\!+\!k'\!+\!k''|^2)^2}
  \frac{ 2 }
  {|k'| \, |k''|}\\
  & \times  \Big| \frac{1}{|k|\!+\!|k|^2 +
  (h_\alpha\! + \!e_0)}
  \frac{1}{ |k|^\frac12 } \nabla \Sgsa  \Big|^2 |k|^{2-\epsilon}
  |k'|^{2-\epsilon} \frac{1}{|k''|} \Big)^\frac12 \\
  & \times \Big( 2 \int \frac{\d k\, \d k'\,
  \d k''\, \d x}{ |k'||k|} |(\Proj^1 \gspl )(k'', x)|^2
  |k''| \frac{\chi_\Lambda^2(|k|)\chi_\Lambda^2(|k'|)\chi_\Lambda^2(|k''|)}
  {|k|^{2-\epsilon}
  |k'|^{2-\epsilon}} \Big)^\frac12 \\
  & \leq c\alpha^7 + \epsilon
  \|H_f^\frac12 \gspl \|^2 ,
\end{split}
\end{equation}
Similarly, we bound $I_1$ as follows,
\begin{equation}\nonumber
\begin{split}
  & |I_1| \leq \alpha^\frac52 \Big(\int
  \frac{\d k\, \d k'\,
  \d k''\, \d x}{(|k|\!+\!|k'|\!+\!|k''| + |k\!+\!k'\!+\!k''|^2)^2}
  \frac{ 2 }
  {|k'| \, |k''|}\\
  & \times  \Big| \frac{1}{|k|\!+\!|k|^2 +
  (h_\alpha\! + \!e_0)}
  \frac{1}{ |k|^\frac12 } \nabla \Sgsa  \Big|^2 |k'|^{2-\epsilon}
  |k''|^{2-\epsilon} \frac{\chi_\Lambda^2(|k|)\chi_\Lambda^2(|k'|)
  \chi_\Lambda^2(|k''|)}{|k|} \Big)^\frac12 \\
  & \times \Big( 2 \int \frac{\d k\, \d k'\,
  \d k''\, \d x}{ |k'||k''|} |(\Proj^1\gspl )(k, x)|^2 |k|
  \frac{\chi_\Lambda^2(|k|)\chi_\Lambda^2(|k'|)\chi_\Lambda^2(|k''|)}
  {|k'|^{2-\epsilon}
  |k''|^{2-\epsilon}} \Big)^\frac12  \\
  & \leq c\alpha^5 \int \d k' \, \d k''
  |k'|^{-\epsilon} |k''|^{-\epsilon}\chi_\Lambda^2(|k'|)\chi_\Lambda^2(|k''|)
  \int \d k \, \frac{1}{|k|^2} \frac{\chi_\Lambda^2(|k|)}{(|k| + \frac{1}{16}
  \alpha^2)^2} \|\nabla \Sgsa \|^2 \\
  & + \epsilon \|H_f^\frac12 \Proj^1 \gspl \|^2
  \leq c\alpha^5 + \epsilon \|H_f^\frac12 \gspl \|^2 ,
\end{split}
\end{equation}
where we took into account that $\frac{\partial \Sgsa }{\partial
x_i}$ is orthogonal to $\Sgsa $, and on the subspace of such
functions we have $(h-\alpha + e_0)\geq \frac{1}{16} \alpha^2$.

\end{proof}

\end{appendix}

$\;$\\

\subsection*{Acknowledgements}
J.-M. B. thanks V. Bach for fruitful discussions. The authors
gratefully acknowledge financial support from the following
institutions: The European Union through the IHP network Analysis
and Quantum HPRN-CT-2002-00277 (J.-M. B., T. C., and S. V.), the
French Ministry of Research through the ACI jeunes chercheurs
(J.-M. B.), and the DFG grant WE 1964/2 (S. V.). The work of T.C.
was supported by NSF grant DMS-0704031.


\bibliographystyle{plain}

\end{document}